\documentclass{amsart}
\usepackage{amssymb, amsthm, amsmath, graphicx, latexsym}

\begin{document}

\newcommand{\barint}{\overline{\hspace{.65em}}\!\!\!\!\!\!\int}
\newcommand{\ep}{\epsilon}
\newcommand{\ue}{u_\epsilon}
\newcommand{\ve}{v_\epsilon}
\newcommand{\jep}{j_\epsilon}
\newcommand{\loc}{ {\mbox{\scriptsize{loc}}} }
\newcommand{\R}{{\mathbb R}}
\newcommand{\C}{{\mathbb C}}
\newcommand{\T}{{\mathbb T}}
\newcommand{\Z}{{\mathbb Z}}
\newcommand{\N}{{\mathbb N}}
\newcommand{\Q}{{\mathbb Q}}
\newcommand{\Hdf}{{\mathcal{H}}}
\newcommand{\calA}{{\mathcal{A}}}
\newcommand{\calI}{{\mathcal{I}}}
\newcommand{\calR}{{\mathcal{R}}}
\newcommand{\calD}{{\mathcal{D}}}
\newcommand{\calL}{{\mathcal{L}}}
\newcommand{\calF}{{\mathcal{F}}}
\newcommand{\calG}{{\mathcal{G}}}
\newcommand{\calH}{{\mathcal{H}}}
\newcommand{\calJ}{{\mathcal{J}}}
\newcommand{\calM}{{\mathcal{M}}}
\newcommand{\massnorm}{{\bf{M}}}
\newcommand{\flatnorm}{{\bf {F}}}
\newcommand{\dist}{\operatorname{dist}}
\newcommand{\sign}{\operatorname{sign}}
\newcommand{\spt}{\operatorname{spt}}
\def\rest{\hskip 1pt{\hbox to 10.8pt{\hfill
\vrule height 7pt width 0.4pt depth 0pt\hbox{\vrule height 0.4pt
width 7.6pt depth 0pt}\hfill}}}
\newcommand{\bd}{\partial}

\newcommand{\beq}{\begin{equation}}
\newcommand{\eeq}{\end{equation}}

\def\logeps{{|\!\log\epsilon|}}

\theoremstyle{plain}
\newtheorem{theorem}{Theorem}
\newtheorem{proposition}{Proposition}
\newtheorem{lemma}{Lemma}
\newtheorem*{lemma2prime}{Lemma 2\,$'$}
\newtheorem{corollary}{Corollary}

\theoremstyle{definition}
\newtheorem{definition}{Definition}

\theoremstyle{remark}
\newtheorem{remark}{Remark}
\newtheorem{example}{Example}
\newtheorem{warning}{Warning}

\numberwithin{equation}{section}
\setcounter{tocdepth}{3}

%%% 
%%% end of newcommands etc
%%%

\title[$\Gamma$-convergence of 3d Ginzburg-Landau functionals]
{Convergence of Ginzburg-Landau functionals in 3-d superconductivity}

\begin{abstract} In this paper we consider the asymptotic behavior of the Ginzburg-Landau model for superconductivity in 3-d, in various energy regimes. We rigorously derive, through an analysis via $\Gamma$-convergence, a reduced model for the vortex density, and deduce a curvature equation for the vortex lines. In the companion paper \cite{bjos} we describe further applications to superconductivity and superfluidity, such as general expressions for the first critical magnetic field $H_{c_1}$, and the critical angular velocity of rotating Bose-Einstein condensates.
\end{abstract}

\author{S. Baldo \and  R.L. Jerrard \and G. Orlandi \and H.M. Soner}

\address{Department of Mathematics, University of Verona, Verona, Italy}\email{sisto.baldo@univr.it}
\address{Department of Mathematics, University of Toronto,
Toronto, Ontario, Canada}\email{rjerrard@math.toronto.edu}
\address{Department of Mathematics, University of Verona, Verona, Italy}\email{giandomenico.orlandi@univr.it}
\address{Department of Mathematics, ETH Z\"urich, Z\"urich, Switzerland}\email{mete.soner@math.ethz.ch}
\maketitle
%\centerline{S. Baldo, R.L. Jerrard, G. Orlandi and H.M. Soner}

%\tableofcontents

%\end{document}
%\centerline{\today}

\section{INTRODUCTION}

In this paper we investigate the asymptotic behavior as $\ep\to 0$ of the functionals
\[
E_\ep(u) \equiv E_\ep(u;\Omega)= \int_\Omega e_\ep(u)\ dx=\int_\Omega \frac 12 |Du|^2 +\frac 1{\ep^2}W(u)\ dx ,
\]
where $\ep>0$, $\Omega$ is a bounded Lipschitz domain in $\R^3$, $u=u^1+iu^2\in H^1(\Omega; \C)$,
$W:\R^2\simeq\C\to \R$ is nonnegative and continuous, $W(u)= 0 \iff |u|=1$, and is assumed to satisfy some growth condition at infinity and around its zero set (see hypothesis $(H_q)$ below).

In the case $W(u)=\frac{(1-|u|^2)^2}{4}$, one usually refers to $E_\ep$ as the Ginzburg-Landau functional. This model is relevant to a variety of phenomena in quantum physics and in fact, as corollaries of its asymptotic analysis we will derive, here and in the companion paper \cite{bjos}, reduced models for density of vortex lines (or curves) in 3-d superconductivity and Bose-Einstein condensation.
In these physical application, $\ep$ represents a (small) characteristic length, $u$ corresponds to a wavefunction, $|u|^2$ to the density of superconducting or superfluid material contained in $\Omega$. Moreover, the {\it momentum}, defined as the 1-form 
$$ju\equiv(iu, du)\equiv u^1du^2 - u^2 du^1\, ,$$
represents the superconducting (resp. superfluid) current, and hence it is natural to interpret the Jacobian $Ju\equiv du^1\wedge du^2$ as the {\it vorticity}, since $2Ju=d(ju)$. 
We refer the reader to the Appendix for  notation used throughout this paper and background on differential forms and related material.

\medskip
\noindent
In the 2-d case it has been recognized since \cite{bbh} that for minimizers $u_\ep$ of $E_\ep$ (subject to appropriate boundary conditions), as $\ep\to 0$, typically the energy scales like $\logeps$ and there are a finite number of singular points, called {\it vortices}, where the energy density $e_\ep(u_\ep)dx$ and the vorticity $Ju_\ep$ concentrate. Moreover, the rescaled energy $\frac{E_\ep(u_\ep)}{\logeps}$ controls the total vorticity. These phenomena are robust, in the sense that analogous results hold in higher dimensions (see \cite{lr,bbo}, where the limiting vorticity is supported in a codimension 2 minimal surface) and under weaker assumptions on $u_\ep$, as stated in the following $\Gamma$-convergence result:

\begin{theorem}[\cite{js,abo}] \label{thm:1} Let $K>0$, $n\ge 2$, $\Omega\subset\R^n$ be a bounded Lipschitz domain, and the potential $W$ satisfy the growth condition\footnote{cf. condition (2.2) in \cite{abo}.}
$$
\liminf_{|u|\to\infty}\frac{W(u)}{|u|^q}>0\, ,\qquad\liminf_{|u|\to 1}\frac{W(u)}{(1-|u|)^2}>0\, ,
 \leqno{(\text{$H_q$})}
 $$ 
for some $q\ge 2$. Then the following statements hold:

\smallskip
\noindent
 {\rm (i)  Compactness and lower bound inequality.} For any  sequence $u_\ep\in H^1(\Omega,\C)$ such that 
$$
E_\ep(u_\ep)\le K\logeps\, ,\leqno{(\text{$H_0$})}
$$ 
we have, up to a subsequence, $Ju_\ep\to J$ in  $W^{-1,p}$ for every  $p<\frac{n}{n-1}$\, ,
where $J$ is an exact measure-valued 2-form in $\Omega$ with finite mass $||J||\equiv|J|(\Omega)$, and $J$ has the structure of an $(n-2)$-rectifiable boundary with multiplicities in $\pi\cdot\Z$. Moreover,
 \beq\liminf_{\ep\to 0} \frac{{E}_\ep(u_\ep)}{\logeps}\ge ||J||.\eeq

\noindent
{\rm (ii) Upper bound (in)equality.}
For any exact measure-valued 2-form $J$ having the structure of an $(n-2)$-rectifiable boundary in $\Omega$ with multiplicities in $\pi\cdot\Z$, there exist $u_\ep\in H^1(\Omega,\C)$ s.t. $Ju_\ep\to J$ in $W^{-1,p}$ for every $p<\frac{n}{n-1}$, and
\beq\lim_{\ep\to 0} \frac{{E}_\ep(u_\ep)}{\logeps}= ||J||.\eeq

\end{theorem}

Other energy regimes arise naturally for $E_\ep$ and are interesting for applications. In particular the energy regime $E_\ep(u_\ep) \approx \logeps^2$ corresponds to the onset of the mixed phase in type-II superconductors, and to the appearance of vortices in Bose-Einstein condensates. These situations have been extensively studied in the 2-d case, especially by Sandier and Serfaty in the case of superconductivity (see \cite{ss} and references therein).  In this energy regime, the number of vortices is of order $\logeps$, hence unbounded as $\ep\to 0$. Another feature is that the contribution of the vortices to the energy is of the same order as the contribution of the momentum, so that the limiting behavior can be described in term of this last quantity, suitably normalized. A $\Gamma$-convergence result for $\frac 1{g_\ep}E_\ep$ for general energy regimes $E_\ep(u_\ep) \lesssim g_\ep  \ll \ep^{-2}$ has been  proved, in the 2-d case, in \cite{js2}, see also \cite{ss}.

\subsection{Main results}
A first result of this paper extends the asymptotic analysis of \cite{js2} to the 3-d case. We write $f_\ep\ll h_\ep$ (or $h_\ep\gg f_\ep$) to express $f_\ep=o(h_\ep)$ as $\ep\to 0$.
We will use the notation
\beq
 \calA_0 := \{(J,v) : J\mbox{ is an exact measure-valued 2-form in $\Omega$,}\ 
v\in L^2(\Lambda^1\Omega)\}
\label{eq:calA0def}\eeq 
Measure-valued $k$-forms are discussed in the Appendix, see in particular Sections \ref{sect:mvf} and \ref{sect:mvftop}. Our conventions imply that a measure-value form $J$ has finite
mass, so that  $\| J \| := |J|(\Omega)<\infty$, where $|J|$ denotes the total variation measure associated with $J$. We say that a measure-valued $k$-form $J$  is {\em exact} if
$J = dw$ in the sense of distributions for some measure-valued $k-1$-form $w$.
We show in Lemma \ref{lem:dp} that a measure-valued $(n-1)$-form
$J$ on a smooth bounded open $\Omega\subset \R^n$ is exact if and only if $dJ=0$ and the associated flux through each component of the boundary $\partial \Omega$ vanishes. The latter condition follows automatically from the former if $\partial \Omega$ is connected.

\begin{theorem}\label{thm:2} Let $\Omega$ be a bounded Lipschitz domain in $\R^3$,  $W(u)$ satisfy $(H_q)$ for some $q\ge 2$, and $\logeps\ll g_\ep\ll\ep^{-2}$. Then the following statements hold:

\smallskip
\noindent
 {\rm (i)  Compactness and lower bound inequality.} For any  sequence $u_\ep\in H^1(\Omega,\C)$ such that
$$
E_\ep(u_\ep)\le Kg_\ep\, ,\leqno{(\text{$H_g$})} \qquad \mbox{ for some }K>0,
$$ 
there exist $(J,v)\in \calA_0$ such that after passing to a subsequence if necessary,
 \beq\label{eq:convju} |\ue|\to 1 \quad\text{in } L^q(\Omega),\qquad
 \frac{ju_\ep}{|u_\ep|\sqrt{g_\ep}}\rightharpoonup v \quad\text{weakly in } {L^2(\Lambda^1\Omega)\, , }
 \eeq
 \beq\label{eq:convu}
\frac{ju_\ep}{\sqrt{g_\ep}}\rightharpoonup v \quad\text{weakly in } {L^{\tfrac{2q}{q+2}}(\Lambda^1\Omega)\, }.
%{L^p(\Lambda^1(\Omega)}\quad\forall\, p<1.4.
 \eeq
If $g_\ep \le \logeps^2$, then in addition
% For the Jacobians $Ju_\ep$, normalized by the factor $\tfrac{\logeps}{g_\ep}$, we have the convergence
\beq\label{eq:convJ}
\frac{\logeps}{g_\ep}Ju_\ep=\frac{\logeps}{2g_\ep}d(ju_\ep) \to J\qquad\text{in }  W^{-1,p}(\Lambda^2\Omega)\, \quad\forall\, p<3/2 .
\eeq
 %where $B\in L^p(\Lambda^1(\Omega))$  $\forall p<3/2$, and where $J$ is an exact measure-valued 2-form in $\Omega$ with finite mass $||J||\equiv|J|(\Omega)$. 
The convergences in \eqref{eq:convu} and \eqref{eq:convJ} yield, in different scaling regimes,
\[
\mbox{if $\logeps \ll  g_\ep \ll \logeps^2$ then $(J,v)\in 
\calA_1 := \{ (J, v) \in \calA_0:  dv = 0\}$,}
\leqno{(\text{$S_1$})} 
\]
\[
\mbox{ if $ g_\ep = \logeps^2$ then $(J,v)\in \calA_2:= \{ (J, v)\in \calA_0 :J = \frac 12 dv \in H^{-1}(\Lambda^2\Omega)  \}$,}
\leqno{(\text{$S_2$})} 
\]
\[
\mbox{ if $\logeps^2\ll g_\ep \ll \ep^{-2}$ then $(J,v)\in \calA_3 := \{ (J, v)\in \calA_0 : J=0 \}$}.
\leqno{(\text{$S_3$})} 
\]
% \beq\label{eq:cases}
% dv\equiv 0\quad \text{if } g_\ep\ll \logeps^2,\qquad J\equiv 0 \quad\text{if } g_\ep\gg\logeps^2\, ,
% \eeq
% \beq\label{eq:criticalcase}
% 2\cdot J=dv\qquad\text{if } g_\ep=\logeps^2\, .
% \eeq
and in every case, % Moreover, in every case
\beq\label{eq:gammaliminf}
\liminf_{\ep\to 0} \frac{{E}_\ep(u_\ep)}{g_\ep}\ge ||J|| +\frac{1}{2}|| v ||^2_{L^2(\Lambda^1\Omega)}\, .
\eeq

\noindent
{\rm (ii) Upper bound (in)equality.}
Assume that $(g_\ep)_{\ep>0}$ satisfies one of the scaling conditions $(S_k)$, $k\in \{1,2,3\}$, identified above, and that $(J,v)\in \calA_k$.
Then $\exists\,  U_\ep\in H^1(\Omega;\C)$ such that \eqref{eq:convju}, \eqref{eq:convu}, \eqref{eq:convJ} hold, and
\beq\label{eq:gammalimsup}
\lim_{\ep\to 0} \frac{{E}_\ep(U_\ep)}{g_\ep}=||J|| +\frac{1}{2}|| v ||^2_{L^2(\Lambda^1\Omega)}.\eeq
\end{theorem}

The compactness and lower bound assertions are either very easy, already  known, see for example \cite{ss-jfa}, or are proved almost exactly as in the 2d case. %We do not however know any source that gives a complete proof of \eqref{eq:gammaliminf} in the critical case $g_\ep = \logeps^2$.
The upper bound \eqref{eq:gammalimsup} is the main new part of the theorem,  and constitutes the most difficult part of the theorem. 

\begin{remark} Assume that $(g_\ep)_{\ep>0}$ satisfies one of the scaling conditions $(S_k)$, $k\in \{1,2,3\}$, identified above, and for $(J,v)\in \calA_0$, 
set
\beq\label{eq:EJv}
E(J,v):=||J||+\frac{1}{2}|| v||^2_{L^2(\Lambda^1\Omega)}\quad\quad  \mbox{ if }(J,v)\in \calA_k,
\eeq
and $E(J,v) := +\infty$ if $(J,v)\not \in \calA_k$. 
We express the $\Gamma$-convergence result of Theorem \ref{thm:2} using the notation
\beq\label{eq:gammalim0}
\frac{E_\ep(u_\ep)}{g_\ep}\xrightarrow{\Gamma} E(J,v),
\eeq
where the $\Gamma$-limit is intended with respect to  the convergences \eqref{eq:convju},\eqref{eq:convu},\eqref{eq:convJ}.
Notice that the contributions of vorticity and momentum are decoupled in the $\Gamma$-limit, due to the different scaling factors in \eqref{eq:convu}, \eqref{eq:convJ}, except for the critical regime $g_\ep=\logeps^2$,
 where the scalings of $Ju_\ep$ and $ju_\ep$ coincide, and the limits
satisfy $2J = dv$ (see section \ref{sect:E} below). In particular, Theorem \ref{thm:2} expresses the fact that for regimes $g_\ep\ll \logeps^2$, the contribution to the energy is given by the vorticity and the curl-free part of the momentum, while for $g_\ep\gg\logeps^2$ the contribution of the vorticity  vanishes asymptotically.
\end{remark} 
\begin{remark} As observed in \cite{js,abo}, replacing $W(u)$ by $\sigma\cdot W(u)$, $\sigma>0$, and letting $\sigma\to 0$, the lower bound \eqref{eq:gammaliminf} can be sharpened to
\beq\label{eq:gammaliminf'}
\liminf_{\ep\to 0}\int_\Omega \frac{|\nabla u_\ep|^2}{2g_\ep}\ge ||J|| +\frac{1}{2}|| v ||^2_{L^2(\Lambda^1\Omega)}\, .
\eeq
Moreover, for a sequence $u_\ep$ satisfying \eqref{eq:gammalimsup}, the potential part of the energy is a lower order term, i.e. 
\beq\label{eq:negligible}
\int_\Omega\frac{W(u_\ep) }{\ep^2}=o(g_\ep)\qquad\text{as } \ep\to 0.
\eeq
Inequality \eqref{eq:gammaliminf'} is also proved in \cite{ss-jfa}.
%%in the case $g_\ep\ll \ep^{-\alpha}$ for any $\alpha>0$ and $W(u)=\frac{(1-|u|^2)^2}{4}$.
\end{remark}
\begin{remark} In the 2-d case the $\Gamma$-convergence result of \cite{js2} is formulated exactly as Theorem \ref{thm:2} above, except for the convergence of the normalized Jacobians $\frac{\logeps}{g_\ep}Ju_\ep$, that takes place there in $W^{1,p}$ for any $p<2$.  
\end{remark}
\begin{remark}\label{rem:localization} By localization, Theorem \ref{thm:2} implies the following: for any $u_\ep$ satisfying $(H_g)$, the rescaled energy densities $\frac{e_\ep(u_\ep)dx}{g_\ep}$ converge
weakly as measures in $\Omega$, upon passing to a subsequence, to a limiting measure $\mu$,  with $|J|+\frac{v^2}{2}dx \le \mu$. It then follows that 
$\mu=|J|+\frac{v^2}{2}dx$ for any sequence $(u_\ep)$ such that the convergences \eqref{eq:convju}, \eqref{eq:convu}, \eqref{eq:convJ} and the upper bound equality \eqref{eq:gammalimsup} hold. 
\end{remark}
\begin{remark}
The final compactness assertion \eqref{eq:convJ} is proved by establishing convergence
in $W^{-1,1}$, and then interpolating, using the easy estimate $\| Ju_\ep \|_{L^1} \le \| Du\|_{L^2}$.
For $\logeps\ll g_\ep \ll \ep^{-2}$, \eqref{eq:convu} already implies that $\frac{\logeps}{g_\ep}Ju_\ep\to 0$ in $W^{-1, \frac{2q}{q+2}}$. This can also be improved by interpolating with $L^1$ estimates
(which imply $W^{-1, 3/2}$ estimates)
if $\frac {2q}{q+2} < \frac 32$.
\end{remark}
\begin{remark}\label{rem:convJn} The convergences \eqref{eq:convju},\eqref{eq:convu},\eqref{eq:convJ} have been already established in the analysis of \cite{js,abo,js2}. In particular, for a domain $\Omega\subset\R^n$ with $n\ge 4$, \eqref{eq:convju} and \eqref{eq:convu} still hold true, while the normalized Jacobians converge to $J$ in $W^{-1,p}$ for any $p<\frac{n}{n-1}$. Moreover, assuming $g_\ep\le \ep^{-\gamma}$ for some $0<\gamma<2$, the convergence in \eqref{eq:convu} can be improved according to $\gamma$, see \cite{js2}.
In \cite{bos-app}, following \cite{bbm}, the convergence in \eqref{eq:convJ} has been proved also to hold in $W^{1,\frac{n}{n-1}}$ (as well as in fractional spaces $W^{s,p}$ with $sp=n/(n-1)$) for $n\ge 4$, and even in the case $n=3$, assuming the condition $u\in L^q(\Omega)$ for $q>6$ (see \cite{bos-app}, Theorem 1.3 and Remark 1.6). 
\end{remark}
%\begin{remark}\label{rem:ss-jfa} \end{remark}
\begin{remark}\label{rem:logeps}
In the scaling $g_\ep = \logeps$ studied in Theorem \ref{thm:1}, 
%it is easy to see that in addition to the compactness described in the theorem,
%we also may assume after passing to a subsequence that 
%\eqref{eq:convju}, \eqref{eq:convu} hold, where any limit $v\in L^2(\Lambda^1\Omega)$must satisfy $dv=0$. Moreover, in this case, 
arguments in the proof of Theorem \ref{thm:2} can easily be adapted to show that 
$
\frac{E_\ep(u_\ep)}{g_\ep}\xrightarrow{\Gamma} E(J,v),
$
where the $\Gamma$-limit is again intended with respect to  the convergences \eqref{eq:convju},\eqref{eq:convu},\eqref{eq:convJ}, and where $E(J,v)$
is defined exactly as in \eqref{eq:EJv}, except that  $E(J,v)$
is set equal to $+\infty$ unless $dv=0$ {\em and } $J$ has the structure of a rectifiable boundary. This is an improvement
over Theorem \ref{thm:1} (cf. analogous results in \cite{bos-rings} for critical points of $E_\ep$, and in \cite{bow} for minimizers with local energy bounds), and in fact is valid in $\R^n$ for any $n\ge 3$. 
\end{remark}
\begin{remark}\label{rem:open} The validity of \eqref{eq:gammaliminf}, \eqref{eq:gammalimsup} in dimension $n\ge 4$ remains an open issue for energy regimes $g_\ep \gg\logeps$. A major difficulty is to determine the correct generalization of the total variation term $\|J\|$ in \eqref{eq:EJv}. Different candidates include the total variation
with respect to the comass norm, the Euclidean norm, and the mass norm, see
\cite{F}. For measure-valued $2$-forms in $\R^3$, all of these coincide. 

The most reasonable conjecture is that the mass norm is the suitable one for the higher-dimensional generalization of Theorem \ref{thm:2}, but this seems difficult to prove.
The arguments we give to prove \eqref{eq:gammaliminf} are in fact presented in $\R^n$,
and for $n\ge 4$ prove that \eqref{eq:gammaliminf} holds  with $\|J\|$ replaced by the {\em comass} of $J$, which in general is strictly less than the mass of $J$. Lower bounds involving the comass norm
in $\R^n, n\ge 4$, are also proved in \cite{ss-jfa}.

By way of illustration,  for the (constant) measure-valued $2$-form $J= dx^1\wedge dx^2 + dx^3 \wedge dx^4$  on an open set $\Omega\subset \R^4$, 
one has comass$(J) = |\Omega|$, the  Euclidean total variation of $J$ is $\sqrt{2} |\Omega|$, and mass$(J) = 2|\Omega|$. 

For $\logeps^2 \ll g_\ep \ll \ep^{-2}$,  the total variation term does not appear in the limiting functional,
so the issue of mass versus comass does not arise, and the proof of the lower bound \eqref{eq:gammaliminf} is straightforward; in fact it follows from arguments we give here. 
The upper bound \eqref{eq:gammalimsup} is probably also easier  in this case 
than for $\logeps \ll g_\ep \le \logeps^2$.

\end{remark}

Replacing assumption $(H_q)$ for $W(u)$ with the  following one (verified in particular for sequences of minimizers)
$$
\exists \, C> 1 \quad\text{such that }   |u_\ep|\le C \qquad\forall \ep<1\, ,\leqno{(\text{$H_\infty$})}
$$
and taking into account Remark \ref{rem:convJn}, a variant of Theorem \ref{thm:2} can be formulated as follows:
%
% I would prefer to call it as Theorem 2 bis. How can one do this with LaTeX ?
%
\begin{theorem}\label{thm:2bis} In the hypotheses of Theorem \ref{thm:2}, we have

\smallskip
\noindent
{\rm (i) Compactness.} For any  sequence $u_\ep\in H^1(\Omega,\C)$ verifying $(H_g)$ and $(H_\infty)$
 we have, up to a subsequence,
 \beq\label{eq:convjubis}
 \frac{ju_\ep}{\sqrt{g_\ep}}\rightharpoonup v \ \text{weakly in } L^2(\Lambda^1\Omega)\, ,\quad
\frac{\logeps}{g_\ep}Ju_\ep\to J\ \text{in }  W^{-1,3/2}(\Lambda^2\Omega)\,  ,
\eeq
where $J$ is an exact measure-valued 2-form in $\Omega$, with finite mass $||J||\equiv|J|(\Omega)$. 

\smallskip
\noindent
{\rm (ii) $\Gamma$-convergence.}  Assuming that $g_\ep$ respects one of the scaling conditions $S_k$ from Theorem \ref{thm:2},
we have
\beq\label{eq:gammalimbis}
\frac{E_\ep(u_\ep)}{g_\ep}\xrightarrow{\Gamma} E(J,v),
\eeq
with respect to the convergence \eqref{eq:convjubis},
where $E(J,v)$ is defined in \eqref{eq:EJv},
taking into account the relevant scaling regime.
\end{theorem}
%In Section \ref{sect:gammaliminf} we will deduce \eqref{eq:gammaliminf} as a direct consequence, %with minor modification, of our previous works \cite{js,abo}. 

%\begin{remark} A $\Gamma$-convergence result for $E_\ep$ for general energy regime 
%$g_\ep=o(\ep^{-2})$ is given  in Theorem \ref{thm:otherregimes} below.
%\end{remark}

\subsection{The critical regime \mathversion{bold}$g_\ep=\logeps^2$\mathversion{normal}}\label{sect:E}
Let us specialize the statements of Theorems \ref{thm:2} and  \ref{thm:2bis} to the critical regime $g_\ep=\logeps^2$, where the scaling factors in \eqref{eq:convju}, \eqref{eq:convu},\eqref{eq:convJ} are equal, and hence %\eqref{eq:criticalcase} holds, i.e. 
the normalized vorticity is related to the momentum by the formula $2 J=dv$. We then have
\beq\label{eq:gammalimE}
\frac{E_\ep(u_\ep)}{\logeps^2}\xrightarrow{\Gamma} E(v),
\eeq
where, for $v\in L^2(\Lambda^1\Omega)$, we define 
\beq\label{eq:ev}
E(v):=E(\frac{dv}{2},v)=\frac{1}{2}||dv||+\frac{1}{2}||v||^2_{L^2(\Lambda^1\Omega)}
\eeq
if the mass $||dv||\equiv|dv|(\Omega)$ is finite, $E(v)=+\infty$ otherwise.
The $\Gamma$-limit is intended with respect to  the convergences \eqref{eq:convju},\eqref{eq:convu},\eqref{eq:convJ}.

Clearly Theorem \ref{thm:2bis}  yields the same conclusion \eqref{eq:gammalimE}, this time with respect to the convergence \eqref{eq:convjubis}, which in this case reads
\beq\label{eq:convbis}
\frac{ju_\ep}{\logeps}\rightharpoonup v \ \text{weakly in } L^2(\Lambda^1\Omega),\quad\frac{2Ju_\ep}{\logeps}\to dv \ \text{in } W^{-1,3/2}(\Lambda^2\Omega)\ .
\eeq

\medskip
\subsection{Applications to superconductivity}\label{S:superc}

As a first application of the above results in the energy regime $g_\ep=\logeps^2$, we describe the asymptotic behavior of the Ginzburg-Landau functional tor superconductivity 
$$
{{\mathcal F}}_\ep(u,A)=\int_{\Omega}\frac{|du -iAu|^2}{2}+\frac{1}{\ep^2}W(u)\, dx+
\int_{\R^3}\frac{|dA-h_{ex}|^2}{2}\, dx
$$
defined for $\Omega\subset\R^3$,  where the 2-form $h_{ex}\in L^2_{loc}(\Lambda^2\R^3)$ is an external applied magnetic field, the 1-form $A\in H^1(\Lambda^1R^3)$ is the induced vector potential (gauge field).
It does not change the problem to assume that $h_{ex}$ has the form $h_{ex} = dA_{ex}$ for some $A_{ex}\in H^1_{loc}(\Lambda^1\R^3)$, and we will always make this assumption.

Let 
$\dot H^1_{*} (\Lambda^1\R^3) :=  \{ A\in \dot H^1(\Lambda^1\R^3) : d^* A = 0 \}$, and define
the
inner product $(A,B)_{\dot H^1_{*} (\Lambda^1\R^3)}:= (dA, dB)_{L^2(\Lambda^2\R^3) }$.
This makes $\dot H^1_{*} (\Lambda^1\R^3)$ into a Hilbert space, 
satisfying in addition the Sobolev inequality
\[
\| A \|_{L^6(\Lambda^1\R^3)} \le C \| A \|_{\dot H^1_{*}(\Lambda^1\R^3)} \, .
\]
We will study $\calF_\ep(v,A)$ for $(v,A)\in H^1(\Omega;\C)\times [A_{ex}+\dot H^1_*(\Lambda^1\R^3)]$;
this is reasonable in view of the gauge-invariance of $\calF_\ep$, that is, the fact
that
\beq\label{eq:gauge}
{\mathcal F}_\ep(u,A)={\mathcal F}_\ep(u{\cdot e^{i\phi}},A{+d\phi}) \qquad\forall \phi\in H^1(\R^3)\, .
\eeq
%(We remark that this holds not only for real-valued $\phi$, but also for $\phi$ taking values in $\R/2\pi\Z$.) 
% note from Bob: commented out because in $\R^3$, \phi is always real-valued, since R^3 is simply connected.

It is useful to decompose  $\mathcal F_\ep$ as follows (see e.g. \cite{br}):
\beq\label{eq:br}
{\mathcal F}_\ep(u,A)=E_\ep(u)\, +\,  {\mathcal I}(u,A) +  {\mathcal M}(A,h_{ex})+{\mathcal R}(u,A)\, , 
\eeq
with
\beq\label{eq:br1}
{\mathcal I}(u,A) :=\, - \, \int_\Omega A\cdot ju\, dx ,
\eeq
\beq\label{eq:br2}
\mathcal M(A,h_{ex}) : =\int_\Omega \frac{|A|^2}{2}dx \, +\int_{\R^3} \frac{| dA - h_{ex}|^2}{2}\, dx
\ = \ \frac 12 \| A\|_{L^2(\Lambda^1\Omega)}^2 + \frac 12 \| A - A_{ex}\|_{\dot H^1_*(\Lambda^1\R^3)}^2.
\eeq
 and $\mathcal R(u,A)=\frac{1}{2}\int_\Omega(|u|^2-1)|A|^2dx$ is a remainder term of lower order. Thus $\mathcal F_\ep(u,A)$ may be written as a continuous perturbation of $E_\ep(u)+\mathcal M(A,h_{ex})$, and using the stability properties of $\Gamma$-convergence we deduce, as in \cite{js2} for the 2-d case, the  $\Gamma$-convergence for the functionals $\mathcal F_\ep$ in the critical energy regime $g_\ep=\logeps^2$:
 
\begin{theorem}\label{thm:3} Let $\Omega\subset\R^3$ be a bounded Lipschitz domain, $W(u)$ satisfy $(H_q)$ with $q\ge 3$, and assume $h_{ex} = dA_{ex,\ep}$ and that there exists $A_{ex,o}\in H^1_{loc}(\Lambda^1\R^3)$
such that $\frac {A_{ex,\ep}} {\logeps} - A_{ex,0}\to 0$ in $\dot H^1_*(\Lambda^1\R^3)$.
Then the following hold.

\medskip
\noindent
{\rm (i) Compactness.} For any sequence $(u_\ep,A_\ep)\in H^1(\Omega;\C)\times [A_{ex,0} +  \dot H^1_*(\Lambda^1\R^3)]$ such that  ${\mathcal F}_\ep(u_\ep,A_\ep)\le K\logeps^2$, we have, up to a subsequence,
\beq\label{eq:convA}
\frac{A_\ep}{\logeps}- A\rightharpoonup 0 \quad\text{weakly in } \dot H^1_*(\Lambda^1\R^3)\, 
\eeq
for some $A\in A_{ex,0} +  \dot H^1_*(\Lambda^1\R^3)$
as well as the convergences \eqref{eq:convju},\eqref{eq:convu},\eqref{eq:convJ} of Theorem \ref{thm:2} in the case $g_\ep=\logeps^2$. 

\medskip
\noindent
{\rm (ii) $\Gamma$-convergence.} 
For $v\in L^2(\Lambda^1\Omega)$ and $A\in  A_{ex,0} +  \dot H^1_*(\Lambda^1\R^3)$, define
 \beq\label{eq:F}
 {\mathcal F}(v,A)=\frac{1}{2}||dv||+\frac{1}{2}|| v-A ||^2_{L^2(\Lambda^1\Omega)}+\frac{1}{2}||dA-dA_{ex,0}||^2_{L^2(\Lambda^2\R^3)}\, 
 \eeq
if $||dv||=|dv|(\Omega)$ is finite, ${\mathcal F}(v,A)=+\infty$ otherwise. 

\medskip
\noindent
Then under the convergences \eqref{eq:convA}, \eqref{eq:convju},\eqref{eq:convu},\eqref{eq:convJ}, we have
\beq\label{eq:gammalimF}\frac{{\mathcal F}_\ep(u_\ep,A_\ep)}{\logeps^2}\xrightarrow{\Gamma}\mathcal F(v,A).
\eeq
\end{theorem}

\begin{remark} Assuming $(H_\infty)$, the $\Gamma$-limit \eqref{eq:gammalimF} is obtained with respect to the convergences
\eqref{eq:convA}, \eqref{eq:convbis}.
\end{remark}

\begin{remark}
The statement of Theorem \ref{thm:3} is not gauge-invariant, as
the condition that $A_\ep\in A_{ex, \ep} + H^1_*(\Lambda^1\R^3)$ uniquely determines the function 
$\phi$ in \eqref{eq:gauge}. Fixing this degree of freedom is clearly necessary for compactness.
Note however that the limiting functional $\calF$ has a gauge-invariance property:
$\calF(v,A) = \calF(v+\gamma|_\Omega, A+\gamma)$ whenever $d\gamma = 0$.
\end{remark}

The Euler-Lagrange equations of the functional $\mathcal F$ consist in  the Amp\`ere law $d^*H=j$ for the resulting magnetic field $H=dA-h$, generated by the (gauge-invariant) super-current $j=v-A$ in $\Omega$ (see \eqref{eq:E-L-F1}), and a curvature equation for the vortex filaments, i.e. the streamlines of the limiting vortex distribution (see \eqref{eq:E-L-F2}), which reads, in the regular case,
\beq\label{eq:curvature}
\begin{cases}
\vec\kappa=2\vec\tau\times\vec\jmath &\text{in }\Omega,\\
\vec\tau_\top=0 &\text{on }\bd\Omega\, .
\end{cases}
\eeq
whith $\vec\kappa$ and $\vec\tau$ denoting respectively the curvature vector and the unit tangent to the vortex filament, $\vec \jmath$ the vector field corresponding to the super-current $j=v-A$, and $\times$ the exterior product in $\R^3$.
Formula \eqref{eq:curvature} generalizes the corresponding law in the case of a finite number of vortices (see \cite{bos-rings}, Theorem 3 (iv),  and \cite{chi}).

\begin{remark}  In \cite{bjos} we analyze in more detail the properties of minimizers of the limiting functional $\mathcal F$ through the introduction of a dual variational problem. %The dual problem has an interpretation as  a vectorial obstacle-type problem, which we believe is by itself of independent interest. 
We use this description to characterize to leading order the first critical field $H_{c_1}$.

These results extend to 3 dimensions facts about 2-d models of superconductivity first established 
by Sandier and Serfaty \cite{ss00}, see also \cite{ss} and other references cited therein. Following the initial work of Sandier and Serfaty, it was shown in \cite{js2} that their results can be recovered via the 2-d analog of the procedure we follow here and in \cite{bjos}.

As far as we know, the relevance of convex duality in these settings was first pointed out by  Brezis and Serfaty \cite{BrezisSerfaty}.
%
% via a  certain non local $L^\infty$ norm of the induced magnetic field $H=dA-h$, and deduce the Meissner effect in an asymptotic sense.
%
%, in particular, by  deducing the $\Gamma$-limit  $\mathcal F$ for the sequence of functionals $({\mathcal F}_\ep)_{\ep\in (0,1]}$ from the $\Gamma$-limit for the reduced functionals $(E_\ep)_{\ep\in (0,1]}$.
%
%In the 2-d case the limiting functional has the same expression as $\mathcal F$, and allows also to deduce quickly (see \cite{js2}) some features of the original model, namely the leading term in the asymptotic expansion of the first critical field and the appearance of free boundaries in the vortex distribution, through the interpretation as an obstacle problem. For a comprehensive discussion of these and other 2-d phenomena, see the treatise \cite{ss}.
\end{remark}

\begin{remark} In \cite{bjos} we also apply Theorem \ref{thm:2} to study the $\Gamma$-limit of the Gross-Pitaevskii functional for superfluidity, and derive in particular a reduced vortex density model for rotating Bose-Einstien condensates, deducing the corresponding curvature equations and an expression for the critical angular velocity.
%cite Aftalion Jerrard? Alama-Bronsard-Montero?
\end{remark}

\begin{remark}
Theorem \ref{thm:3} is concerned with the description of the behavior of {\it global} minimizers. The convergence of {\it local} minimizers with bounded vorticity has been studied, under various assumptions, in \cite{JMS,MSZ,m-z}, relying on techniques related to Theorem \ref{thm:1}.
\end{remark}

%\subsection{ OPEN PROBLEMS}
%:  higher dimensions, bounded vorticity at the first critical field, more on meissner effect

\subsection{Plan of the paper} This paper is organized as follows: in Section \ref{sect:gammaliminf} we prove the lower bound and compactness statement (i) of Theorem \ref{thm:2}, while Section \ref{sect:upperbound} is devoted to the proof of the upper bound statement (ii). In Section \ref{sect:supercond} we prove Theorem \ref{thm:3} and derive the Euler-Lagrange equations of the $\Gamma$-limit, obtaining in particular formula \eqref{eq:curvature}.
Section \ref{sect:appendix} is an Appendix that collects some notation and the proofs of some auxiliary results.

\bigskip
{\bf Acknowledgements.} This research has been partially funded by G.N.A.M.P.A. of Istituto Nazionale di Alta Matematica (visiting professor program),  Universit\`a di Verona (funding program Cooperint), and the National Science and Engineering Research Council of Canada 
under operating Grant 261955.  We warmly thank these institutions for support and kind hospitality.
We are also grateful to Giovanni Alberti for numerous helpful discussions.

%%%%%%%%%%%%%%%%%%%%%%%%%%%%%%%%%%%%

%%%%%%%%%%%%%%%%%%%%%%%%%%%%%%%%%%%%%%%

%%%%%%%%%%%%%%%%%%%%%%%%%%%%%%%%%%%%%%%%%%%

\section{LOWER BOUND AND COMPACTNESS}\label{sect:gammaliminf}
 
 %%%%%%%%%%%%%%%%%%%%%%%%%%%%%%%%%%%%%%%%%%%%%

In this section we prove statement (i) of Theorem \ref{thm:2}, relying largely on our previous works \cite{js2,abo}. We prove everything in $\Omega\subset \R^n$, for arbitrary $n\ge 3$. We note however
that the lower bound inequality \eqref{eq:gammaliminf} is not expected to be sharp when $n\ge 4$, see Remark \ref{rem:open}.

 We first derive \eqref{eq:convju} and \eqref{eq:convu}.
Then, assuming  \eqref{eq:convJ}, we derive the characterization of the limiting spaces $\calA_k$ corresponding to the scaling regimes $S_k$ identified in the statement of the Theorem.
%\eqref{eq:cases} and \eqref{eq:criticalcase}. 
We next turn to the proof of the lower bound  \eqref{eq:gammaliminf}.
The compactness statement \eqref{eq:convJ} in the case $p=1$ will be 
obtained during the proof of \eqref{eq:gammaliminf}, and the case $1<p<\frac{n}{n-1}$, (see Remark \ref{rem:convJn}) will from the case $p=1$  by a short interpolation argument.

\medskip
\noindent
{\bf Proof of \eqref{eq:convju}, \eqref{eq:convu}.} Observe first that $|\ue|\to 1$ in $L^q(\Omega)$ by assumptions $(H_q)$ on $W(u)$  and $(H_g)$ on $E_\ep$, since
$$
\int_\Omega\left|1-|\ue |\right|^q\le C \int_\Omega W(u_\ep)\le C\ep^2 E_\ep(\ue)\le C\ep^2g_\ep\to 0\, .
$$
From the identity
$|u|^2|\nabla u|^2=|u|^2|\nabla |u| |^2+|ju|^2$ we deduce that
\beq\label{eq:boundju}
 \int_\Omega\frac{|ju_\ep|^2}{|u_\ep|^2g_\ep}\le 2\cdot\frac{E_\ep(u_\ep)}{g_\ep}\le 2K,
\eeq
which yields, up to a subsequence, $\tfrac{j\ue}{|\ue|\sqrt{g_\ep}}\rightharpoonup v$ weakly in $L^2(\Omega)$, completing the proof of  \eqref{eq:convju}. 
Now write
$$
\frac{j\ue}{\sqrt{g_\ep}}=\frac{j\ue}{|\ue|\sqrt{g_\ep}}+{(|\ue|-1)}\cdot\frac{j\ue}{|\ue |\sqrt{g_\ep}}\, .
$$
Using \eqref{eq:convju} we deduce  that  ${(|\ue|-1)}\cdot\frac{j\ue}{|\ue |\sqrt{g_\ep}}\rightharpoonup 0$    weakly in $L^{\tfrac{2q}{q+2}}(\Omega)$. This yields $\frac{j\ue}{\sqrt{g_\ep}}\rightharpoonup v$ weakly in $L^{\tfrac{2q}{q+2}}(\Omega)$, i.e. \eqref{eq:convu}.

%To obtain \eqref{eq:convu}, we follow \cite{js2} and use Sobolev-Gagliardo Nirenberg inequality to %deduce that actually, $|u_\ep|\to 1$ in $L^r$ for any $r<4+2/3$, which
%yields $ju_\ep\rightharpoonup v$ weakly in $L^p$ for any $p<1.4$. VERIFY!
%i.e. \eqref{eq:convu}.

\qed

Next,  the characterization of the limiting spaces $\calA_k$
follows  from \eqref{eq:convju},\eqref{eq:convu} and \eqref{eq:convJ}, since by \eqref{eq:convu} we deduce that $d(\frac{ju_\ep}{\sqrt g_\ep})\rightharpoonup dv$ weakly in $W^{-1,\frac{2q}{q+2}}(\Omega)$, hence, in the case $g_\ep\gg \logeps^2$,
\beq\label{eq:cases'}
\frac{\logeps}{g_\ep}Ju_\ep=\left(\frac{\logeps}{\sqrt g_\ep}\right)d\left(\frac{ju_\ep}{\sqrt g_\ep}\right)\rightharpoonup 0\cdot dv=0\quad\text{in }W^{-1,\frac{2q}{q+2}}(\Omega)\, .
\eeq
In view of \eqref{eq:convJ}, this implies $J=0$ by uniqueness of the weak limit.
On the other hand, in the case $g_\ep\ll \logeps^2$, 

$$
d\left(\frac{ju_\ep}{\sqrt g_\ep}\right)=2(\frac{\sqrt g_\ep}{\logeps})\cdot\left(\frac{\logeps}{g_\ep}Ju_\ep\right)\to 0\cdot J=0\quad\text{in }W^{-1,p}(\Omega),\ p< \frac n{n-1}\, ,
$$
which implies $dv=0$, again by uniqueness of the weak limit. The above formulas, in the case $g_\ep=\logeps^2$, imply that $dv = 2J$.

\medskip

\noindent
We turn to the proof of \eqref{eq:gammaliminf} distinguishing two cases, namely $\logeps \ll g_\ep\le \logeps^2$, and $ \logeps^2\ll g_\ep \ll \ep^{-2}$. We begin with the latter case.

\medskip
\noindent
{\bf Proof of \eqref{eq:gammaliminf} in the case \mathversion{bold}${g_\ep\gg\logeps^2}$.\mathversion{normal}}  In this energy regime, we have just shown that $J=0$, and \eqref{eq:convju} and \eqref{eq:boundju} immediately imply
\beq
\liminf_{\ep\to 0}\frac{E_\ep(u_\ep)}{g_\ep}\ge \frac{1}{2}\int_\Omega |v|^2\, ,
\eeq
yielding conclusion \eqref{eq:gammaliminf}.

\qed

If it is not true that $g_\ep \gg \logeps^2$, then by passing to a subsequence we may suppose that $g_\ep \le C \logeps^2$. By renaming the constant $K$ in $(H_g)$ we may also assume that $C=1$. Thus the proof of \eqref{eq:gammaliminf}
will be completed by the following.

\medskip
\noindent{\bf Proof of \eqref{eq:gammaliminf} in the case \mathversion{bold}${\logeps \ll g_\ep\le \logeps^2}$.\mathversion{normal}} The main step in the proof is the following  improvement of \cite{abo}, Proposition 3.1. We establish it in greater generality than
is needed for the proof of \eqref{eq:gammaliminf}.

We remark that \eqref{eq:gammaliminf} in the scaling ${\logeps \ll g_\ep\le \logeps^2}$
is already established in \cite{ss-jfa}, and moreover that a key point in the proof there is a result similar to the following proposition. 

\begin{proposition}\label{prop:3.1} Let $u_\ep$ be a sequence of smooth maps on $\Omega\subset \R^n$, $n\ge 2$,
such that $(H_g)$ holds, with $\logeps \le g_\ep \le \logeps^2$.
Then we have, up to a subsequence,
\beq\label{eq:convflat}
\frac{\logeps}{g_\ep}Ju_\ep\to J\, \qquad\text{in } W^{-1,1}(\Lambda^2\Omega)\, ,  
\eeq
where $J$ is an exact measure-valued 2-form\footnote{In the case $g_\ep=\logeps$, $J$ has the structure of a rectifiable boundary with multiplicities in $\pi\cdot\Z$, according to Theorem \ref{thm:1}.} 
with finite mass in $\Omega$. Moreover, there exists a closet set $C_\ep\subset \Omega$ such that $|C_\ep|\to 0$,  and such that for every simple 2-covector $\eta$ such that $|\eta|=1$ and for every open set $U\Subset\Omega$, it holds
\beq\label{eq:prop:3.1-0}
\liminf_{\ep\to 0} \frac{E_\ep(u_\ep; C_\ep)}{g_\ep}\ge |( J\, , \eta) |(U)\, ,
\eeq
where $( J\, , \eta)$ is the signed measure defined according to \eqref{eq:mueta}.

\end{proposition}

Our proof of Proposition \ref{prop:3.1} differs  from that of the corresponding point (Proposition IV.3) in \cite{ss-jfa}. One feature of our proof is that the set $C_\ep$ that
we construct is manifestly a closed set, whereas in the construction of \cite{ss-jfa}, a
certain amount of work is required even to see that the corresponding set is measurable.

Taking for granted Proposition \ref{prop:3.1}, we complete the proof of \eqref{eq:gammaliminf}. First,  a standard localization argument (see \cite{abo}, p. 1436) gives, for any finite collection of pairwise disjoint open sets $U_j\Subset\Omega$ and simple unit 2-covectors $\eta_j$, 
\beq\label{eq:prop:3.1loc}
\sum_j|(J\, ,\eta_j)|(U_j)\le \liminf_{\ep\to 0}\frac{E_\ep(u_\ep; C_\ep)}{g_\ep}
\eeq

Taking the supremum over all choices of pairwise disjoint open sets $U_j$ and unit simple 2-covectors $\eta_j$ on the l.h.s. of \eqref{eq:prop:3.1loc}  yields  the total {\it comass norm} of $J$ in the sense of 
\cite{F}, section 1.8.1. In the 3-dimensional case\footnote{and for any $n\ge 3$ if $g_\ep=\logeps$, then $J$ is obtained as a limit of polygonal currents with uniformly bounded mass, and hence is rectifiable by the Federer-Fleming closure theorem.} 
this coincides with the total variation (or $L^1$, accordingly) norm of $J$, since all 2-covectors in $\R^3$ are necessarily simple. Hence we may write, for $n=3$,
\beq\label{eq:prop:3.1final}
|J|(\Omega)\le \liminf_{\ep\to 0}\frac{E_\ep(u_\ep; C_\ep)}{g_\ep}\, .
\eeq

Let now $\Omega_\ep\equiv\Omega\setminus C_\ep$, and $\chi_{\ep}(x)$ be the characteristic function of $\Omega_\ep$. We may assume after passing to a subsequence that $\chi_{\ep}(x)\to 1$ as $\ep\to 0$ for a.e. $x\in\Omega$, since $|C_\ep|\to 0$.
Then  for any 
$h\in L^2$, $\chi_\ep\cdot h\to h$ in $L^2$ by the dominated convergence theorem,
% (cf. \cite{js2}, p. 541), 
and so it follows from \eqref{eq:convju} that
$$
\int_\Omega h\cdot\chi_\ep\cdot\frac{ju_\ep}{|u_\ep| \sqrt{g_\ep}}\to\int h\cdot v\qquad\text{as }\ep\to 0.
$$
That is, $\chi_\ep\cdot\frac{ju_\ep}{|u_\ep|\sqrt{g_\ep}}\rightharpoonup v$ weakly in $L^2$. Since
$$
\int_{\Omega_\ep} e_\ep(u)\ge \frac{1}{2}\int_{\Omega}{\chi}_{\Omega_\ep}\frac{|ju_\ep|^2}{|u_\ep|^2}
$$
we deduce that
\beq\label{eq:liminf-v}
\liminf_{\ep\to 0} \frac{E_\ep(u_\ep;\Omega_\ep)}{g_\ep}\ge \liminf_{\ep\to 0}\frac{1}{2}\int_{\Omega}{\chi}_{\Omega_\ep}\frac{|ju_\ep|^2}{|u_\ep|^2g_\ep}\ge \frac{1}{2}\int_\Omega v^2\, .
\eeq

To conclude observe that $E_\ep(u_\ep;\Omega)=E_\ep(u_\ep;C_\ep)+E_\ep(u_\ep;\Omega_\ep)$, so that
\beq\label{eq:fatou}
\liminf_{\ep\to 0}\frac{E_\ep(u_\ep;\Omega)}{g_\ep}\ge \liminf_{\ep\to 0}\frac{E_\ep(u_\ep;C_\ep)}{g_\ep}+\liminf_{\ep\to 0}\frac{E_\ep(u_\ep;\Omega_\ep)}{g_\ep}\, .
\eeq
Combining \eqref{eq:fatou} with \eqref{eq:liminf-v} and \eqref{eq:prop:3.1final} we obtain 
\eqref{eq:gammaliminf}

\qed

We now supply the

\medskip
\noindent
{\bf Proof of Proposition \ref{prop:3.1}.} 
We will proceed in two steps: first, we apply the discretization procedure of \cite{abo}, Section 3 at a suitable scale $\ell_\ep$ to deduce \eqref{eq:convflat} and to obtain a identify a small set $C_\ep^\prime\subset\Omega$ where the Jacobian $Ju_\ep$ is essentially confined. Second, we apply the cited  procedure again, this time imposing
an additional condition that yields good control of the resulting 2-form $\nu_\ep'$ (a discretization of the Jacobian)  in a small neighborhood $C_\ep$ of $C_\ep^\prime$ by the Ginzburg-Landau energy {\em in the same small neighborhood} $C_\ep$. We then argue that the restriction of  $\nu_\ep'$ to a suitable subset of $C_\ep$ converges to the same limit as $Ju_\ep$, so that from lower semicontinuity,  bounds on $(\nu_\ep', \eta) \rest C_\ep$ yield estimates on $(J,\eta)$, thereby
proving \eqref{eq:prop:3.1-0} .

We carry out these arguments in detail in the case $n=3$ and then we  discuss the general case.
% (simply to make it clear that it is essentially identical to the case $n=3$.)

\medskip
\noindent
{\bf Step 1.} 
%{\it There exists a set $C_\ep^\prime\subset\Omega$, such that $|C_\ep^\prime|\to 0$, and 
%$|| Ju_\ep ||_{W^{-1,1}(U\setminus C_\ep^\prime)}\to 0$ as $\ep\to 0$, for any open set %$U\Subset\Omega$. Moreover, \eqref{eq:convflat} holds true.}
%\begin{proof} 
We follow \cite{abo}, Section 3. Fix a unit simple 2-covector $\eta$, and an orthonormal basis $(\vec e_i)$ of $\R^3$ satisfying $\eta(\vec e_2\wedge \vec e_3)=1$. Consider a grid $\mathcal G=\mathcal G(a, \vec e_i,\ell)$, given by the collection of cubes with edges  of size $\ell$, and vertices having coordinates (with respect to a reference system with origin in $a\in\R^3$ and orthonormal directions $(\vec e_i)_{i=1,2,3}$) which are integer multiples of $\ell$. For $h=1,2$ denote by $R_h$ %=R_h(a,\vec e_i,\ell)$ 
 the $h$-skeleton of $\mathcal G$, i.e. the union of all $h$-dimensional faces of the cubes of $\mathcal G$. 
 %and by $\tilde R_2$ the union of the faces orthogonal to say $\vec e_1$.
 Consider also the dual grid having vertices in the centers of the cubes of $\mathcal G$, and denote by $R^\prime_h$ for $h=1,2$, its $h$-skeleton.
From $(H_g)$ and the assumption that $g_\ep \le \logeps^2$ we have
\beq\label{eq:3.22}
E_\ep(\ue;\Omega)\le K  \logeps^2\, , \quad\text{and we set }\ell\equiv\ell_\ep:=\logeps^{-10}\, .
\eeq
Observe that \eqref{eq:3.22} replaces (3.22) and (3.23) in \cite{abo}. Choose $a\equiv a_\ep$ by a mean-value argument in such a way that Lemma 3.11 of \cite{abo} holds, so that in particular, the restriction of the energy on the 2-d and 1-d skeleton of $\mathcal G$ is controlled by
\beq\label{eq:meanvalue}
\int_{R_h\cap\Omega} e_\ep(u_\ep) d\mathcal H^h\le C_0\ell^{h-3} E_\ep(u_\ep;\Omega)\, ,\quad h=1,2\, ,
\eeq
for a suitable constant $C_0>1$, and moreover
\beq\label{eq:meanvalue1}
\ell\int_\Omega \frac{e_\ep(u_\ep)}{|\text{dist}(x,R_1)|}dx\le C_0 E_\ep(u_\ep;\Omega).
\eeq
In view of \eqref{eq:3.22}, Lemma 3.4 in \cite{abo} is satisfied, hence $|u_\ep|\to 1$ uniformly on $R_1\cap\Omega$. In particular, for any face $Q\in R_2$, the topological degree $d_Q:=\,$deg$\, (\frac{u_\ep}{|u_\ep|},\partial Q, S^1)\in\Z$ is well-defined (modulo the choice of an orientation of $Q$ in $\R^3$).

The discretization procedure of \cite{abo}, Lemmas 3.7 to 3.10, may then take place on any fixed open set $U\Subset\Omega$, yielding an oriented polyhedral 1-cycle (actually, a relative boundary in $\bar U$) $M_\ep=\sum (-1)^{\sigma_i}d_{Q_i}\cdot Q^\prime_i$, where $Q^\prime_i\subset R^\prime_1$ is the unique edge of the cubes of the dual grid intersecting the face $Q_i\subset R_2$, the sign $(-1)^{\sigma_i}$ depends on the orientations of both $Q_i$ and $Q'_i$, and the sum is extended to any $Q_i\subset R_2$ such that $Q_i\cap U\neq\emptyset$. Notice that $M_\ep$ is supported in $R^\prime_1\cap U^{\sqrt 3\ell}$, where $U^{\sqrt 3\ell}$ denotes the tubular neighborhood of $U$ of thickness $\sqrt 3\ell$. The cycle $M_\ep$ gives rise to a (measure-valued) 2-form
$\nu_\ep$, whose action on 2-forms in $C^\infty_c(\Lambda^2\Omega)$ is defined by
\beq\label{eq:nu}
\langle\nu_\ep, \varphi\rangle\, :=\pi\cdot\sum_{\substack{Q_i\subset R_2\\ Q_i\cap U\neq\emptyset}} (-1)^{\sigma_i}d_{Q_i}\int_{Q^\prime_i}\star \varphi\, .
\eeq
The 2-form $\nu_\ep$ is exact in $U$, since $M_\ep$ is a relative boundary in $\bar U$, and enjoys the following properties:
it is a measure-valued 2-form supported in $R^\prime_1\cap U^{\sqrt 3\ell}$, such that its total variation $|\nu_\ep|$ is bounded 
on $U$ by\footnote{cf. \cite{abo}, (3.29)}
\beq\label{eq:massboundnu}
|\nu_\ep|(U)=\sum_{\substack{Q_i\subset R_2\\ Q_i\cap U\neq\emptyset}} \pi \ell\cdot |d_{Q_i}|\le C \frac{E_\ep(u_\ep;\Omega)}{\logeps}
\eeq
with $C>0$ independent of $U\Subset\Omega$, and such that $\nu_\ep$ is close to $Ju_\ep$ in the $W^{-1,1}$ norm, namely\footnote{combine \eqref{eq:3.22} and \eqref{eq:meanvalue1} with  (3.7) and (3.14) of \cite{abo}.}
\beq\label{eq:flatU}
|| Ju_\ep -\nu_\ep ||_{W^{-1,1}(\Lambda^2U)}\le C\ell\cdot E_\ep(u_\ep; \Omega) .%\quad\forall\, U\Subset\Omega\, .
\eeq
Moreover, the support of $\nu_\ep$  is contained in the interior of a set $C^\prime_\ep\subset U^{\sqrt 3\ell}$ given by the union of those cubes of the grid $\mathcal G$ having at least one face $Q\subset R_2$, $Q\cap  U\neq\emptyset$, such that $d_Q\neq 0$. Denote by $I$ the set of indices $i$ in \eqref{eq:massboundnu} for which $d_{Q_i}\neq 0$, or equivalently, $|d_{Q_i}|\ge 1$. By \eqref{eq:massboundnu} we have
\beq\label{eq:smallvolume}
|C^\prime_\ep|\le \ell^3\cdot |I|\le \sum_{i\in I}\ell^3\cdot |d_{Q_i}|\le C\ell^2 \frac{E_\ep(u_\ep;\Omega)}{\logeps}\, , %\le C\frac{\ell^2\cdot g_\ep}{\logeps}\, ,
\eeq
so that by \eqref{eq:3.22}, $|C^\prime_\ep|\to 0$ as $\ep\to 0$.
%Moreover, by \eqref{eq:flatU},
%\beq\label{eq:smallflatnorm}
%|| Ju_\ep||_{W^{-1,1}(\Lambda^2(U\setminus C^\prime_\ep))}=|| Ju_\ep-\nu_\ep||_{W^{-1,1}%(\Lambda^2(U\setminus C^\prime_\ep))}\le K\ell\cdot g_\ep\, , %\quad\forall\, U\Subset\Omega\, ,
%\eeq
%so that by \eqref{eq:3.22}, $|| Ju_\ep||_{W^{-1,1}(\Lambda^2(U\setminus C^\prime_\ep))}\to 0$ as $\ep\to 0$ for any $U\Subset\Omega$.

Notice moreover that \eqref{eq:massboundnu} and $(H_g)$ imply that 
$
\frac{\logeps}{g_\ep}\cdot\nu_\ep\rightharpoonup J  
$
weakly as measures, where $J$ is a measure-valued 2-form in $\Omega$, which is exact and has total variation $\displaystyle{|J|(\Omega)\le C\liminf_{\ep\to 0}\frac{E_\ep(u_\ep;\Omega)}{g_\ep}}$. By \eqref{eq:flatU} we finally deduce
that $\frac{\logeps}{g_\ep}\cdot Ju_\ep\to J$ in $W^{-1,1}(\Lambda^2U)$ for any $U\Subset\Omega$, which yields \eqref{eq:convflat}

\noindent
{\bf Step 2.} For $N>0$ to be chosen below, define $C_\ep\equiv C_{N,\ep}:=\{x\in\Omega,\ \text{dist}(x,C^{\prime}_\ep)\le 2N\ell\}$ to be the tubular neighborhood of $C^\prime_\ep$ of thickness $2N\ell$ intersected with $\Omega$.  By \eqref{eq:smallvolume} we have
\beq\label{eq:cep}
|C_\ep|\le 8N^3|C^\prime_\ep|\le CN^3\ell^2\frac{g_\ep}{\logeps}\to 0\qquad\text{as }\ep\to 0\, ,
\eeq 
as long as $N^3\le \ell^{-1}$. In view of \eqref{eq:3.22}, \eqref{eq:cep} is verified for instance by fixing 
\beq\label{eq:3.21}
N\equiv N_\ep:=\logeps^{3}.
\eeq
 Observe moreover that 
 \beq\label{eq:3.22''}
 E_\ep(u_\ep; C_\ep)\le E_\ep (u_\ep;\Omega)\le Kg_\ep\le \logeps^2\, .
 \eeq
%We set\footnote{in the general case choose $\ell^\prime_\ep\ll \ell_\ep$ satisfying $\ell^\prime_\ep\cdot g_\ep\to 0$ and  $(\ell^\prime_\ep)^n\ll\ep^\alpha\cdot g_\ep$ $\forall\,\alpha>0$, cf. previous footnote}
%\beq\label{eq:3.22'}
%\ell^\prime\equiv\ell^\prime_\ep:=\logeps^{-5}\, ,
%\eeq
Consider the grid $\mathcal G^*_\ep=\mathcal G(b_\ep,\vec e_i,\ell)$, where $\ell = \ell_\ep = \logeps^{-10}$ as above and $b_\ep$ is chosen such that  for an arbitrarily fixed $\delta>0$, (3.18), (3.19) and (3.20) in Lemma 3.11 of \cite{abo} hold true, and moreover (3.17) holds true with $\Omega$ replaced by $C_\ep$. 
In other words, denoting  by $R^*_h$ the $h$-skeleton of $\mathcal G^*_\ep$ , $h=1,2$, and $\tilde {R^*_2}$ the union of the faces of the 2-skeleton of $\mathcal G^*_\ep$ orthogonal to $\vec e_1$ we have, 
\beq\label{eq:meanvalueC}
\int_{\tilde R^*_2\cap(C_\ep)} e_\ep(u_\ep) d\mathcal H^2\le (1+\delta)\ell^{-1} E_\ep(u_\ep;C_\ep),
\eeq
\beq\label{eq:meanvalueprime}
\int_{R^*_h\cap\Omega} e_\ep(u_\ep) d\mathcal H^h\le C_0\delta^{-1}\ell^{h-3}E_\ep(u_\ep;\Omega)\, ,\quad h=1,2\, ,
\eeq
\beq\label{eq:meanvalueE}
\ell\int_\Omega \frac{e_\ep(u_\ep)}{|\text{dist}(x,R^*_1)|}dx\le C_0\delta^{-1} E_\ep(u_\ep;\Omega).
\eeq
% One readily verifies that conditions \eqref{eq:3.22} and \eqref{eq:3.22''} still allow to apply once more the procedure of \cite{abo}, Section 3 as described in Step 1, with $\Omega$ replaced by $C_\ep^N\cap U^\prime$  and $U$ replaced by $U\cap C^{\prime}_\ep$. 
Fix an open subset $U\Subset\Omega$. As in Step 1, the procedure of \cite{abo} yields a polyhedral cycle 
\beq\label{eq:Mepprime}
M'_\ep=\sum_{\substack{Q_i\subset R^*_2\\ Q_i\cap U\neq\emptyset}} (-1)^{\sigma_i}d_{Q_i}\cdot Q^\prime_i\, ,
\eeq 
which is a relative boundary in 
$\bar{U}$ and is supported in ${R^*_1}'\cap U^{\sqrt 3\ell}$, where  ${R^{*}_1}'$ is the 1-d skeleton  of the dual grid to ${\mathcal G}^*$. The correponding measure-valued 2-form $\nu_\ep^\prime$,  defined as in \eqref{eq:nu} by
\beq\label{eq:nu'}
\langle\nu_\ep^\prime , \varphi\rangle \, :=\pi\cdot\sum_{\substack{Q_i\subset R^*_2\\ Q_i\cap U\neq\emptyset}}(-1)^{\sigma_i}d_{Q_i}\int_{Q'_i}\star\varphi\, ,  
 \qquad\forall\, \varphi\in C^\infty_c(\Lambda^2(\Omega))\, ,
 \eeq
 is exact on $U$ and verifies $|\nu'_\ep|(U)\le C\frac{E_\ep(u_\ep;\Omega)}{\logeps}$ with $C>0$ independent of $U$. 
 %and
 %\beq\label{eq:flatU'}
 %||Ju_\ep-\nu_\ep'||_{W^{-1,1}(U)}\le \ell E_\ep (u_\ep;\Omega)\le K\ell g_\ep\, .%\quad\forall\, U\Subset\Omega\, .
 %\eeq
 %In particular, $\frac{\logeps}{g_\ep}\cdot \nu'_\ep\rightharpoonup J$ weakly as measures in $\Omega$.
% Moreover, either $M'_\ep\equiv M_\ep$ in $U$, or their respective supports are disjoint (since either $R_1\equiv R^*_1$ or $R_1\cap R^*_1=\emptyset$). 
%Moreover, from \eqref{eq:flatU} and \eqref{eq:flatU'} we deduce 
 %\beq\label{eq:nuclosenu'}
 %||\nu'_\ep-\nu_\ep||_{W^{-1,1}(U)}\le 2K\ell g_\ep\, .%\qquad\forall \, U\Subset \Omega . 
 %\eeq

For $x\in\Omega$ define $f(x):=\dist(x, M_\ep)$, so that $f$ is 1-Lipschitz. Denoting by $C^t=\{x:\ f(x)\le t\}\cap\Omega$, we have that $C^{2N\ell}\subset C_\ep$.
%, and for a.e. $t>0$ we have that $\bd C^t=f^{-1}(t)$ and $\bd C^t$ is a regular surface. 
 
 \begin{lemma}\label{lem:nurest} There exists $t:=t_\ep< N\ell$ such that 
 \beq\label{eq:nurest}
 ||\nu'_\ep\rest C^t-\nu_\ep||_{W^{-1,1}(U)}\le C(\ell+N^{-1})g_\ep\, ,%\qquad\forall\, U\Subset\Omega
 \eeq
 with $C>0$ independent of $\ep$ and $U$. In particular, the choices  of $\ell$ and $N$ (see \eqref{eq:3.22}  and \eqref{eq:3.21}) imply that 
 \beq\label{eq:convnu'}
 \frac{\logeps}{g_\ep}\cdot\nu'_\ep\rest C^t\to J \qquad \text{in } W^{-1,1}(\Lambda^2U) %\qquad \forall\, U\Subset\Omega
 \eeq
 and, for any 2-covector $\eta$,
 \beq\label{eq:convnu'eta}
 (\frac{\logeps}{g_\ep}\cdot\nu'\rest C^t,\eta)\to (J,\eta)\qquad\text{in } W^{-1,1}(U). %\qquad \forall\, U\Subset\Omega.
 \eeq
 \end{lemma}
 
 We postpone the proof of Lemma \ref{lem:nurest} to Section \ref{sect:nurest} of the Appendix. By \eqref{eq:convnu'eta} and lower semicontinuity of total variation we deduce
 \beq\label{eq:lsc1} 
 \begin{aligned}
 |(J,\eta)|(U) &\le\liminf_{\ep\to 0}|(\frac{\logeps}{g_\ep}\cdot\nu'_\ep\rest C^t\, ,\eta)|(U)\\
 &\le \liminf_{\ep\to 0}|(\frac{\logeps}{g_\ep}\cdot\nu'_\ep\rest C^{N\ell},\eta)|(U)\, .
 \end{aligned}
 \eeq
 %\beq\label{eq:massboundmprime}
 %|M'_\ep\cap C^N_\ep\cdot \eta|\le (1+\delta)\frac{E_\ep(u_\ep; C^N_\ep\cap U)}{\logeps}.
 %\eeq

 %\beq\label{eq:boundsnuprime}
 %||Ju_\ep-\nu_\ep^\prime||_{W^{-1,1}(\Lambda^2(U\cap C_\ep^\prime))}\le \ell^\prime E_\ep(u_\ep; C_\ep\cap U^\prime)\le K\ell\cdot g_\ep\, ,
 %\eeq
 %\beq\label{eq:boundsnuprime'}
 %\quad |\nu^\prime_\ep|(U)\le C\frac{E_\ep(u_\ep; C_\ep\cap U^\prime)}{\logeps}\, ,
 %\eeq
%with $C>0$ independent of $U$ and $U^\prime$. Combining \eqref{eq:boundsnuprime'}, \eqref{eq:boundsnuprime}, \eqref{eq:smallflatnorm} of Step 1 %and $(H_g)$ we deduce
%$
%\frac{\logeps}{g_\ep}\nu^\prime_\ep\rightharpoonup J
%$
%weakly as measures. 
Observe that specializing \eqref{eq:nu'} to the case $\varphi=\psi \, \eta$, with $\psi\in C^\infty_c(\Omega)$, and letting $\psi$ approach the characteristic function of $C^{N\ell}\cap U$, 
 we have
 \beq\label{eq:nueta}
 |(\nu^\prime_\ep\rest C^{N\ell},\eta)|(U)=|(\nu'_\ep,\eta)|(C^{N\ell}\cap U)=\pi\cdot\sum_{\substack{Q_i\subset R^*_2\\ Q_i\cap U\neq\emptyset}} \left|d_{Q_i}\int_{Q'_i\cap C^{N\ell}\cap U}\star\eta\right|\, .
 \eeq
 Notice that for any $Q'\subset {R^*_1}'$ such that $Q'\cap C^{N\ell}\neq\emptyset$, its dual element $Q$ is contained in the tubular nighborhood of thickness $\sqrt 3\ell$ of $C^{N\ell}$, which is a subset of $C^{2N\ell}$, so that in particular $Q\subset C_\ep$.
Recalling from the definitions that $\star \eta = dx^1 $, which is the oriented arclength element along $Q_i'$ for $Q_i\in \tilde {R^*_2}$, we obtain from
\eqref{eq:nueta} that
\beq\label{eq:boundetanu}
|(\nu^\prime_\ep \rest C^{N\ell},\eta) |(U)\le\sum_{Q\subset \tilde {R^*_2}\cap C_\ep}\pi\ell 
\cdot |d_{Q}|.
\eeq
One readily verifies, following \cite{abo}, p. 1435, that \eqref{eq:3.22} and \eqref{eq:3.22''} allow to  apply Lemma 3.10 there (which relied in turn on a fundamental estimate in \cite{je,sa}), to efficiently estimate
the sum of the degrees $|d_{Q}|$ in terms of $E_\ep(u_\ep ; C_\ep)$. Namely, for any $r>0$, and any $Q\subset R^*_2\cap\Omega$  we have
\beq\label{eq:estimatedeg0}
(1-c_r(\ep))\pi\cdot |d_Q|\le \frac{1}{\logeps}\int_Q e_\ep(u_\ep)d\mathcal H^2 \, +\frac{Kr\ell}{\logeps}\int_{\partial Q} e_\ep(u_\ep) d\mathcal H^1\, ,
\eeq
where $c_r(\ep)$ is independent of $Q$, and $c_r(\ep)\to 0$ as $\ep\to 0$ (see \cite{abo}, p. 1435).
 We may thus write 
\beq\label{eq:estimatedeg}
(1-c_r(\ep))\sum_{Q\subset\tilde R^*_2\cap C_\ep}\!\!\!\!\!\!\!\!\! \pi \cdot |d_{Q}|\le \frac{1}{\logeps}\int_{\tilde {R^*_2}\cap C_\ep}\!\!\!\!\!\!\!\!\!\!\!\!\!\! e_\ep(u_\ep) d\mathcal H^2\, +\frac{Kr\ell}{\logeps}\int_{R^*_1\cap C_\ep} \!\!\!\!\!\!\!\!\!\!\!\!\!\! e_\ep(u_\ep) d\mathcal H^1.
\eeq
 Combining \eqref{eq:boundetanu} with \eqref{eq:estimatedeg}, and taking into account \eqref{eq:meanvalueC}, \eqref{eq:meanvalueprime}, we are led to
\beq\label{eq:estimatemass}
(1-c_r(\ep))|(\frac{\logeps}{g_\ep}\cdot\nu^\prime_\ep\rest C^{N\ell},\eta)|(U)\le (1+\delta+\frac{Kr}{\delta}) \frac{E_\ep(u_\ep; C_\ep)}{g_\ep}\, .
\eeq
Passing to the limit as $\ep\to 0$, we have, in view of \eqref{eq:lsc1},
\beq\label{eq:estimatemass'}
|(J,\eta)|(U)\le (1+\delta+\frac{Kr}{\delta})\liminf_{\ep\to 0} \frac{E_\ep(u_\ep; C_\ep)}{g_\ep}\, .
\eeq
Taking $r<\delta^2$ and $\delta$ arbitrarily small yields \eqref{eq:prop:3.1-0}.

\qed

\noindent
{\bf Proof in the general case $n\ge 3$.}
The main tool used above is the algorithm from \cite{abo} for constructing a polyhedral 
approximation of the Jacobian $Ju$, and hence a measure-valued $2$-form $\nu_\ep$, with good
estimates of $\|Ju - \nu_\ep\|_{W^{-1,1}}$ and of $|(\nu_\ep, \eta)|(W)$ for suitable subsets $W\subset \Omega$. The procedure in \cite{abo} in fact is presented
in $\R^n$, $n\ge 3$, and so can be employed in the general case as for $n=3$, with purely cosmetic differences.
For example,  in $\R^n$, the analog of $Q_i'$ in \eqref{eq:nu} and elsewhere is 
now the unique  $n-2$ face of the dual grid that intersects $Q_i$. Also, different scalings make it convenient to 
choose $\ell = \logeps^{-(3n+1)}$, say, while we still take $N= \logeps^3$. Then it remains true that 
$g_\ep \ll N$, which is needed for the proof of Lemma \ref{lem:nurest}, and that $|C_\ep'|\to 0$,
which follows from the fact that $N^n \ell^2 \frac{g_\ep}{\logeps}\to 0$ as $\ep \to 0$, compare \eqref{eq:cep}.
Modulo   changes of this sort, the argument is identical in the general case.
\qed

\noindent
{\bf Proof of \eqref{eq:convJ}.} Recall that we have assumed that $g_\ep \le \logeps^2$.
Since 
\beq\label{eq:c00}
|| Ju_\ep-\nu_\ep ||_{L^1(\Lambda^2U)}\le || Ju_\ep||_{L^1(\Lambda^2U)}+|| \nu_\ep ||_{L^1(\Lambda^2U)}\le C E_\ep(u_\ep;\Omega)\le C g_\ep 
\eeq
 for any $U\Subset\Omega$, we deduce, by interpolation with \eqref{eq:flatU},
\beq\label{eq:c0alpha}
|| Ju_\ep-\nu_\ep||_{W^{-1,p}(\Lambda^2U)}\le C (\ell_\ep\cdot g_\ep)^{1-\frac{n(p-1)}{p}}g_\ep^{\frac{n(p-1)}{p}}\le C\ell_\ep^{1-\frac{n(p-1}{p}}\cdot \logeps^2\, .
\eeq

The conclusion \eqref{eq:convJ} follows by choosing  $\ell_\ep=\ell_{\ep,p}=\logeps^{-\frac{3p}{n-p(n-1)}}$, so that the r.h.s of \eqref{eq:c0alpha} vanishes.

\qed
% \left( || Ju_\ep-\nu_\ep||_{W^{-1,1}(U)}\right)^{1-\frac{np}{p-1}}\cdot\left(|| Ju_\ep-\nu_\ep %||_{L^1(U)}\right)^{\frac{np}{p-1}}
%\eeq

%
%
%
%
%
%%%%%%%%%%%%%%%%%%%%%%%%%%%%%%%%%%%%%%%%%%%%
%
%
%
%
%
%
%
%
%
%
%
%
%

%
%

%%%%%%%%%%%%%%%%%%%%%%%%%%%%%%%%%%%%%%%%%%%%

\section{UPPER BOUND}\label{sect:upperbound}

%%%%%%%%%%%%%%%%%%%%%%%%%%%%%%%%%%%%%%%

In this section we prove statement (ii) of Theorem \ref{thm:2}.

\subsection{Strategy of proof} The proof is subdivided in various steps. First of all, we reduce in Section \ref{S3.2} to considering an appropriate dense class of the domain of the $\Gamma$-limit, using a suitable finite elements approximation. The construction of the recovery sequence will be based on a Hodge decomposition of the limiting momentum $p$, described in Section \ref{S:hodge} and a discretization of the limiting vorticity $dp$ in terms of a system of lines where the vorticity is concentrated and quantized; this, and associated estimates of the discretized vorticity and related quantities, are the main points in the proof. An argument {\it \`a la} Biot-Savart then allows us to construct $S^1$-valued maps whose Jacobian is concentrated precisely on the discretized vorticity lines, and we obtain our maps $u_\ep$ by adjusting the modulus around the vortex cores. The  proof is completed by the verification of the upper bound inequality, which relies crucially on good properties of the discretized vortex lines and estimates satisfied by associated auxiliary functions.

\subsection{Nice dense class}\label{S3.2}

We say that a $1$-form $p$ on a domain $\Omega\subset \R^3$
is rational piecewise linear if $p$
is continuous, and there exist
a family of closed polygons $\{ P_i\}$ with pairwise disjoint interiors
such that $\Omega\subset \cup P_i$ with $p$ linear
on each $P_i\cap \Omega$, and if the flux $\int_{T_j} dp$ is a rational
number
for every face $T_j$ of
every polyhedron $P_i$.

\begin{lemma}\label{lem:fem} Suppose that $\Omega\subset\R^3$ is a bounded open subset  and that
$\partial \Omega$ is of class $C^1$.
Given $p\in L^2(\Lambda^1(\Omega))$  such that $dp$ is a 
measure, and given $\delta>0$ small, there exists a polygonal set $\Omega_\delta^P$ with $\Omega\Subset\Omega_\delta^P\Subset\Omega^\delta=\{\dist(x,\Omega)<\delta\}$, and such that $\Omega\simeq\Omega_\delta^P\simeq\Omega^\delta$,  and a rational, piecewise linear $1$-form $p_\delta\in L^2(\Lambda^1 \Omega_\delta^P)$,
such that $d p_\delta\in L^1(\Lambda^2\Omega_\delta^P)$ and 
\beq
\| p - p_\delta\|_{L^2(\Omega)}  \le\delta
\label{ndc1}\eeq
\beq
\| p_\delta \|_{L^2(\Omega_\delta^P)}^2 \le 
\| p \|_{L^2(\Omega)}^2 + \delta
\label{ndc2}\eeq
\beq
\int_{\Omega_\delta^P}|d p_\delta| \ \le 
|d p|(\Omega) + \delta.
\label{ndc3}\eeq
\end{lemma}

\begin{proof}
{\bf Step 1}. 
We say that a simplex $P$ is {\em rational} if, whenever $p = \sum (a^{ij} x_j+b^i)dx_i$ 
is a linear $1$-form on $P$ with $a^{ij}, b^i$ rational for all $i,j$, the flux of $p$ through
every face of $P$ is rational.
We claim that the unit cube in $\R^n$ can be covered by closed {rational} simplices with 
pairwise disjoint interiors. Repeating the same construction in every integer
translate of the unit cube, we can cover $\R^3$ by closed rational simplices
with pairwise disjoint interiors. Note also that if we dilate the simplices by
any rational factor, the resulting simplices are still rational.

\smallskip
Let $S_0$ denote the standard simplex  co$\,\{ 0, e_1, e_2, e_3\}$ in $\R^3$,
where co$\,\{ \cdots\}$ denotes the convex hull. 
If $p$ is linear on $P$ with rational coefficients, 
then the flux $\int_T dp$ is a rational number  when $T = \mbox{co}\,(\{0, e_i, e_j\})$ (with either orientation),
for any choice of $i,j\in \{1,2,3\}$ with $i\ne j$. Since $\int_{\partial P} dp  =0$,
it follows that the flux through the fourth face is rational as well.
Thus $S_0$ is rational

Similarly, for any $i,j\in \{1,2,3\}$ with $i\ne j$, let
$S_{ij} = $ co$\,\{e_i+e_j, e_i, e_j, e_1+e_2+e_3\}$.
The same argument as above shows that $S_{ij}$ is rational.
We next claim that
\[
[0,1]^3 \setminus (S_0\cup S_{12}\cup S_{13}\cup S_{23}) =  \mbox{co}\,\{ e_1, e_2, e_3, e_1+e_2+e_3\}.
\]
This follows by noting  that $[0,1]^3\setminus (S_0\cup( \cup_{i,j} S_{ij}))$ is convex,
and that its extreme points are exactly $\{ e_1, e_2, e_3, e_1+e_2+e_3\}$.
Every face of $\mbox{co}\,\{ e_1, e_2, e_3, e_1+e_2+e_3\}$
is also a face of either $S_0$ or of $S_{ij}$ for some $i,j$, so it follows from what
we have
already said that $\mbox{co}\,\{ e_1, e_2, e_3, e_1+e_2+e_3\}$ is rational.

\medskip\noindent
{\bf Step 2}. 
By adapting standard approximation techniques for $BV$ functions as in \cite{gi}, we can
find a set $\Omega'$ such that $\Omega\Subset \Omega'$,
and a $1$-form $p'\in C^\infty(\Lambda^1(\Omega'))$, such that $\| p - p'\|_{L^2(\Omega)}  \le\delta/2$, $\|p'\|^2_{L^2(\Omega')}\le\|p\|^2_{L^2(\Omega)}+\delta/2$ and $|dp'|(\Omega')\le |dp|(\Omega)+\delta/2$.

Choose now a domain $\Omega_\delta$ such that $\Omega\Subset\Omega_\delta\Subset\Omega'$, and $\Omega_\delta$ is the union of a finite number
of cubes with pairwise disjoint interiors and rational edges.

By the discussion in Step 1 above, we can triangulate $\Omega_\delta$ with rational simplices. Performing dyadic subdivisions of each cube, 
we may also obtain rational triangulations with arbitrarily small mesh size 
(and with fixed geometry, since  the angles appearing in the triangulation 
will be precisely those in our original decomposition of the unit cube). 

By standard interpolation theory from the finite elements method (see for instance \cite{C}, Chapter 3), we can find piecewise linear $1$-forms which are arbitrarily close to $p'$ in $W^{1,2}(\Omega_\delta)$: it suffices to choose a sufficiently fine triangulation constructed as above, 
and to take the (unique) piecewise linear
form $p_\delta$ which interpolates $p'$ in the vertices of the triangulation. Moreover, an arbitrarily small change of $p_\delta$ in the vertices makes it rational. 
\end{proof}

We will also need the following variant of the above.

\begin{lemma2prime} Suppose that $\Omega\subset\R^n$ is a bounded open subset  and that
$\partial \Omega$ is of class $C^1$.
Given an exact measure-valued $2$-form $J$,
and given $\delta>0$ small, there exists a polygonal set $\Omega_\delta^P$ such that $\Omega\Subset\Omega_\delta^P\Subset\Omega^\delta=\{\text{dist}(x,\Omega)<\delta\}$, and such that $\Omega\simeq\Omega_\delta^P\simeq\Omega^\delta$,  and a rational, piecewise linear $1$-form $p_\delta'\in L^2(\Lambda^1 \Omega_\delta^P)$,
such that $d p_\delta'\in L^1(\Lambda^2\Omega_\delta^P)$ and such that
%\beq
%\| p - p_\delta\|_{L^2(\Omega)}  \le\delta
%\label{ndc1a}\eeq
%\beq
%\| p_\delta \|_{L^2(\Omega_\delta^P)}^2 \le 
%\| p \|_{L^2(\Omega)}^2 + \delta
%\label{ndc2a}\eeq
\beq
\| p - p_\delta\|_{W^{-1,1}(\Omega)} \le \delta, 
\quad\quad\quad
\int_{\Omega_\delta^P}|d p_\delta'| \ \le 
|J|(\Omega) + \delta.
\label{ndc3a}\eeq
\end{lemma2prime}

The proof is a straightforward modification of the proof of Lemma \ref{lem:fem}, once we note from 
Corollary \ref{cor:poincare} in the Appendix that any exact measure-valued $2$-form $J$ in $\Omega$ can be written in the form
$J=dp'$ for some $p'\in \cap_{1\le q< \frac n{n-1}} L^q(\Lambda^1\Omega)$ 

\subsection{Hodge decomposition of $p_\delta$} \label{S:hodge}
Here we refer for  notations and basic theory to section \ref{sect:hodge} of the Appendix. 
%\begin{remark} If $\Omega\subset\R^n$, $n\ge 3$, Lemma \ref{lem:dp} holds for $dp$ a piecewise %constant exact $(n-1)$-form. A similar statement holds for piecewise constant co-exact 1-forms $d^*p$.
%\end{remark}
We henceforth  write $p$ instead of $p_\delta$. 

Since basic results on Hodge theory to which we appeal require some smoothness of the domain, 
we  fix an open set $\Omega_\delta$ with smooth boundary, such that
$\Omega\Subset\Omega_\delta\Subset\Omega_\delta^P$, and such that 
$\Omega\simeq\Omega_\delta\simeq\Omega_\delta^P$. In particular we assume that 
if $\partial \Omega_\delta^P$ has connected components $(\partial \Omega_\delta^P)_i, i=1,\ldots, b$, then
there exist disjoint connected open sets $W_1,\ldots, W_b$ such that 
\beq\label{eq:Od-OdP}
\Omega_\delta^P \setminus \bar \Omega_\delta = \cup_{i=1}^b W_i,\quad\quad
\partial W_i = (\partial \Omega_\delta^P)_i \cup (\partial \Omega_\delta)_i \ \ \quad \forall i.
\eeq

Consider the Hodge decomposition $p =\gamma+ d\alpha+d^*\beta$ on $\Omega_\delta$ satisfying the boundary conditions \eqref{eq:hodgen}. Thanks to Corollary \ref{cor:poincare} in the Appendix, we know that $\beta=-\Delta^{-1}_N(dp)$, so that in particular $||\beta||_q\le C_q || dp||_1$ $\forall\, q<3/2$.
%\beq\label{eq:hodger3}
%p=\gamma+d\alpha+d^*\beta, \qquad\text{where } \quad\beta=-\Delta^{-1}_N(dp)
%\eeq
%and $\alpha=-\Delta^{-1}_N(d^*p)$. 
Recall that by $L^2$-orthogonality of the Hodge decomposition we have
\beq\label{eq:parseval}
\int_{\Omega_\delta}|p|^2=\int_{\Omega_\delta}|\gamma|^2+|d^*\beta|^2+|d\alpha|^2\, .
%\, \, ,\qquad 
%||d^*\beta||_{L^q(\R^3)}\le C_q \int_{\Omega_\delta}| dp|\quad \forall\, q<3/2\, .
\eeq

We emphasize that in what follows, we will carry out most geometric arguments on the polygonal set $\Omega_\delta^P$,
but the Hodge decomposition always refers to the smooth set $\Omega_\delta \subset \Omega_\delta^P$.
%\begin{remark}\label{rem:hodgeintrinsic0}{\rm
%For different purposes (boundary value problems, manifold ambient formulations) we are also interested
%in proving \eqref{eq:gammalimsup} using intrinsic Hodge decompositions of $p$ on $\Omega$, such as $p=\gamma'+d\alpha'+d^*\beta'$ with boundary conditions \eqref{eq:hodgen}. We will treat this case in section \ref{sect:hodgeintrinsic}.
%}
%\end{remark}

%We henceforth drop subscript and write $p$ instead of $p_\delta$. 

%We will write $\eta = d_N^{-1}\mu$ to mean that \eqref{meas3} holds.

\subsection{Discretization of $dp = dd^*\beta$.}\label{S3.4} 

We will use different arguments to approximate the different terms in the Hodge decomposition of $p$. Most of our effort will be devoted to $d^*\beta$.
As noted above, the first step in our construction is to discretize $dp = dd^*\beta$,
which one can think of as the vorticity.

\begin{proposition}\label{prop:1} Let $p$ be a rational 1-form supported on $\Omega_\delta^P\subset\R^3$, and fix $\eta\in(0,1)$. For any $h\le \eta^2$ there exists an exact measure-valued 2-form $q_h$ in $\Omega_\delta^P$ such that:

\smallskip
\smallskip
\noindent
{\rm (i)}  \qquad\qquad\qquad\quad $q_h=dd^*\beta_h$, where $\beta_h=-\Delta^{-1}_Nq_h$\ in $\Omega_\delta$.
$$
|| q_h - dp ||_{W^{-1,1}(\Omega_\delta^P)}\le C\eta ,\leqno{\rm(ii)}
$$
$$
|q_h|(\Omega_\delta^P) \le  ||dp||_{L^1(\Omega_\delta^P)} +C\eta\, ,\leqno{\rm(iii)}
$$
$$
\begin{aligned}
 ||d^*\beta_h||_{L^p(\Omega_\delta)} \le C_p | q_h|(\Omega_\delta)\, , &\quad d^*\beta_h\rightharpoonup d^*\beta^\eta\ \text{in }L^p(\Omega_\delta)\ \forall\, p<3/2,\\
 || d^*\beta^\eta-d^*\beta &||^2_{L^2(\Omega_\delta)} \le C\eta\, ,
 \end{aligned}
 \leqno{\rm(iv)}
$$
where $C>0$ is independent of $h,\eta, U$.
 For any $\varphi\in C^0(\Lambda^2\Omega_\delta^P)$ we have the integral representation
$$
\langle q_h\, ,\varphi\rangle=h\int_{\Gamma_h}\star\varphi={h}\sum_{\ell=1}^{m(h)}\int_{\Gamma_h^\ell}\star\varphi\, ,\leqno{\rm(v)} 
$$
where $\Gamma_h=\cup_{s=1}^{n(h)}L_h^s\subset\Omega_\delta^P$, $L_h^s$ an oriented line segment $\forall\, s,h$, $m(h)<n(h)\le Kh^{-1}$, and for any $\ell,h$, $\Gamma^\ell_h$ is an oriented simple piecewise linear curve in $\Omega_\delta^P$ such that  $\bd\Gamma_h^\ell\cap U=\emptyset$ $\forall\, U\subset\Omega_\delta^P$. In particular, we have $|q_h|(U)=h|\Gamma_h\cap U|$ for any $U\subset\Omega_\delta^P$. Moreover
$$
\dist (L_1\, , L_2)> c_0\eta h^{1/2}\qquad  \mbox{ if $L_1, L_2$ are disjoint closed line segments of 
$\Gamma_h$,}
\leqno{\rm(vi)}
$$
with $c_0>0$ independent of $h,\eta$. 

Finally, if $L_1, L_2$ be two line segments of $\Gamma_h^\ell$ with exactly one endpoint in common, and
$\tau_1,\tau_2$ are their respective unit tangents, then 
$$
\tau_1\cdot \tau_2>-1+C\eta^2\, ,\leqno{\rm(vii)}
$$
for some $C>0$ independent of $h,\eta$.

\end{proposition}

\begin{remark}The discretized vorticity $q_h$ has a $1$-dimensional character, in that it is supported on a union of line segments, so that in realizing it as a (measure-valued) $2$-form, rather than a $1$-form or vector field, we are departing both from the convention discussed in \eqref{eq:hatt} and from standard practice in geometric measure theory. However, this departure is natural in that $q_h$ is an approximation of the $2$-form $dp$, and it is very useful when we want to appeal to  Hodge Theory to solve
elliptic equations with $q_h$ on the right-hand side, as in conclusion $(i)$ above.
\end{remark}

\begin{remark}
The role of the parameter $\eta$ is to guarantee that  $q_h$
enjoys certain properties such as a  good lower bound on distance between
distinct piecewise linear curves in the support of $q_h$, see conclusion (vi) above.
These are essential for the verification of the upper bound inequality.
\end{remark}

\begin{remark}Our arguments (in particular the proof of (iv)) show that there exists $2$-form $q^\eta$ such that $q_h\rightharpoonup q^\eta$ weakly as measures. as $h\to 0$. In fact our
construction is designed to yield an explicit description of $q^\eta$, see \eqref{eq:qetaformula}. This complicates the construction of $q_h$ but immediately yields uniform estimates of $q^\eta$, needed for (iv), that would otherwise require some work to obtain.
\end{remark} 

\smallskip
%\begin{remark} By Remark \ref{rem:weaktrace} in the Appendix we may consider 
%\end{remark}

\begin{proof}
The proof of Proposition \ref{prop:1} will be divided in several steps.

\medskip
\noindent
{\bf  Proof of (v).}
We start by constructing $q_h$, which amounts to constructing a collection $\Gamma_h$ of line segments, see (v).
Let $\eta\in (0,1)$ be fixed, and
let $p$ be a piecewise linear rational 1-form with respect to the triangulation $\{ S_i\}$ of $\Omega_\delta^P$ as fixed in the proof of Lemma \ref{lem:fem}. 
In particular, for each $i$ there exists a vector $v_i=(v_i^1,v_i^2,v_i^3)$ such that
$dp\rest S_i=\sum_j v^j_i\star dx_j$. 
For any  simplex $S_i$, let $b_i$ its barycenter, and let 
\beq\label{eq:homo}
\tilde S_i=(1-\eta)\cdot S_i +\eta\cdot b_i\subset S_i
\eeq
be a homothetic copy of $S_i$, and let $T_{ij},\tilde T_{ij}$, $j=1,\ldots,4$ be the 2-faces of $S_i,\tilde S_i$ respectively, with the induced orientations.

We will arrange that within each $\tilde S_i$, our discretization of $dp$ is supported on a finite union of  line segments exactly parallel to $v_i$. In order to to this and to match fluxes across the faces of each $S_i$, we discretize the flux through the faces of each $S_i$ and each $\tilde S_i$ in related though different ways.

For every $i$ and for $j\ne k\in \{1,\ldots, 4\}$, define $T_{ijk}\equiv\pi^{-1}(\pi(T_{ij})\cap\pi (T_{ik}))\cap T_{ij}$ (with the orientation of $T_{ij}$), where $\pi\equiv\pi_i$ is  the projection on the 2-plane $(v_i)^\perp$. One may think of $T_{ijk}$ as the portion of $T_{ij}$ connected to $T_{ik}$ by flux lines of $dp$. Further define
\[
\phi_{ij} = \int_{T_{ij}} dp \in \Q, \quad\quad
\phi_{ijk}=\int_{T_{ijk}}dp \ = \ \frac{|T_{ij}| } {|T_{ijk}|} \phi_{ij}.
\]
Clearly $\phi_{ij}=\sum_{j\neq k}\phi_{ijk}$, and hence $\phi_{ijk}\in\Q$, as solutions of a linear systems with rational data. Let $\phi^{-1}$ be the least common denominator of $\, \{ |\phi_{ijk}|\}\in\N$, so that $\phi_{ijk}\phi^{-1}\in\Z$. 

%REMARK: Let $k=[\phi/h]$, and $\tilde p_h=\frac{\phi}{hk}\cdot p$.

For $N\in \N$, we define $h_N := \frac{\phi}{N}$, so that $\frac{\phi_{ijk}}{h_N}\in\Z$ for all $i,j,k$, and similarly $\frac{\phi_{ij}}{h_N} \in \Z $ for every $i,j$. We will prove the proposition for every $h_N$ such that  $h_N<\eta^2$; for arbitrary $h<\eta^2$, the conclusions of the proposition then hold if we define
$q_h := q_{h_N}$, $\beta_h := \beta_{h_N}$, for $N$ such that $h_N\le h < h_{N-1}$.

We henceforth fix an arbitrary $N$ such that $h_N<\eta^2$, and we drop the subscript and write simply $h$.

We first discretize $dp$ on every $T_{ij}$. In order to avoid discretizing any $2$-face twice in inconsistent ways, we define
\[
\mathcal T := \{ T_{ij}
 \ : \ \phi_{ij}>0 \mbox{ or }T_{ij}\subset \partial \Omega_\delta^P \}.
\]
For $T_{ij}\in \mathcal T$, let $m = m_{ij} := \frac{\phi_{ij}}h \in \Z$, 
and let $\ell = \ell_{ij}$ verify
$(\ell_{ij}-1)^2< m \le \ell_{ij}^2$. Now   partition $T_{ij}$ into $\ell_{ij}^2$ closed
triangular pieces $\{  T^a_{ij} \}_{a=1}^{\ell^2}$ with pairwise disjoint interiors, each one isometric to $\ell_{ij}^{-1}T_{ij}$. Select $m$ of these triangles, and let
$\{ s^a_{ij}\}_{a=1}^m$ be the barycentres of the chosen triangles.

If $T_{ij}\not\in \mathcal T$,  then $T_{ij} = -T_{i'j'}$ for some $T_{i'j'}\in \mathcal T$, we set 
$m = m_{ij} := m_{i'j'}$,  and $s^a_{ij} = s^a_{i'j'}$ for $a=1\ldots m_{ij}$.

Next we consider $\{\tilde T_{ij}\}$. For $i,j,k$, let 
$\tilde T_{ijk}\equiv (1-\eta)\cdot T_{ijk} + \eta \cdot b$ (with the orientation of $
T_{ijk}$) and define
\[
\tilde {\mathcal T}:= \{ \tilde T_{ijk}
 \ : \ \phi_{ijk}>0\}.
\]
Now proceed as above: for each $\tilde T_{ijk}\in \tilde {\mathcal T}$, let $ m =  m_{ijk} :=  \frac{\phi_{ijk}}h \in \Z$ and  $\ell_{ijk} := \lceil \sqrt{m}\,\rceil$,
and partition $\tilde T_{ijk}$ into $\ell_{ijk}^2$ closed
triangular pieces $\{ \tilde T_{ijk}^a\}_{a=1}^{\ell_{ijk}^2}$ with pairwise disjoint interiors, each one isometric to $\ell_{ijk}^{-1}\tilde T_{ijk}$. Select $m$ of these triangles, and let
$\{ \tilde s^a_{ijk}\}_{a=1}^m$ be the barycentres of the chosen triangles.

If $T_{ijk}\not\in \tilde {\mathcal T}$, then  $\phi_{ijk}\le 0$. If $\phi_{ijk}=0$  (which in particular happens if $T_{ijk}=\emptyset$) we do nothing.
If $\phi_{ijk}<0$, then noting that our orientation conventions imply that $\phi_{ijk} = -\phi_{ikj}$, we see that
$\tilde T_{ikj}\in \tilde {\mathcal T}$, and we define
$\tilde s^a_{ijk} = \pi^{-1}_i\pi_i (\tilde s^a_{ikj}) \cap T_{ijk}$.

We now define piecewise linear curves as follows. First, for every $T_{ijk}\in \tilde {\mathcal T}$, we define
\[
\tilde \Gamma^a_{ijk} := [\pi_i^{-1}(\pi(\tilde s^a_{ijk}))]\cap\tilde S_i,\quad\mbox{oriented so that }
\partial \tilde \Gamma^a_{ijk} = \tilde s^a_{ijk} - \tilde s^a_{ikj}.
\]
Here and below, if $c$ is an oriented piecewise smooth curve, we write $\partial c = p-q$ to mean
that $\int_c df = f(p) - f(q)$ whenever $f$ is a smooth function. We define $ \Gamma_i = \sum_{j,k,a} \tilde \Gamma^a_{ijk}$, so that $\Gamma_i \subset \tilde S_i$, and
\beq
\partial  \Gamma_i = \sum_{j,k,a} \sign(\phi_{ijk}) \tilde s^a_{ijk}.
\label{eq:minconbdy1}\eeq
Moreover that $\Gamma_i$ is the collection of segments with the smallest total arclength
satisfying this condition (as the segments of $ \Gamma_i$ are all parallel to each other.)

Now for each $i,j$, let $P_{ij} := \{ (1- \lambda) x + \lambda b_i : x\in T_{ij}, 0<\lambda <\eta\}$
be the pyramidal frustum having bases  $T_{ij}$ and $\tilde T_{ij}$, and let $\Gamma_{ij}$
be a collection of (oriented) line segments such that
\beq
\partial \Gamma_{ij} = \sum_{a} \sign(\phi_{ij}) s^a_{ij} - \sum_{k,a} \sign(\phi_{ijk}) \tilde s^a_{ijk},
\label{eq:minconbdy2}\eeq
and that minimizes the total arclength among the set of all collections of line segments satisfying the constraint \eqref{eq:minconbdy2}.  Such collections exist, since $\sign(\phi_{ij}) =  \sign(\phi_{ijk})$ and $m_{ij} = \sum_{k=\ne j} m_{ijk}$, so that  $\sum_{j,a} \sign(\phi_{ij}) - \sum_{j,k,a} \sign(\phi_{ijk}) =0$.
Hence $\Gamma_{ij}$ is well-defined, and clearly $\Gamma_{ij}\subset P_{ij}$.

We define $\Gamma_h$ to be the union $\cup \Gamma_i \cup \Gamma_{ij}$ of the families of segments constructed above, and $n(h)$ to be the total number of segments comprising $\Gamma_h$.
We also define $\Gamma^\ell_h$, for $\ell=1,...,m(h)$, where $m(h)\le N(h)$, to be the polyhedral curves realizing the connected components of $\Gamma_h$. It follows from \eqref{eq:qh.bdy1}, proved below, that
$\bd\Gamma^\ell_h=0$ in $\Omega_\delta^P$

Finally, we  define the measure-valued 2-form $q_h$ to satisfy statement (v).

In the following we will write `` a region'' to refer either to one of the $\tilde S_i$ or one of the $P_{ij}$.
We remark that the definition of $\Gamma_h$ states that,
in the language of Brezis, Coron, and Lieb \cite{bcl}, its restriction to
any region is a {\em minimal connection}, subject to the condition \eqref{eq:minconbdy1} 
in $\tilde S_i$ and \eqref{eq:minconbdy2} in $P_{ij}$.

%Now, for each fixed simplex $\tilde S_i$, denote $s_{ijk}^n=P_{ijk}^n$ if $\phi_{ij}>0$, %$s_{ijk}^n=N_{ijk}^n$ if $\phi_{ij}<0$. Notice that $\sharp\{P_{ij}^n\}=\sharp\{ N_{ij}^n\}$ because %$\sum_{j}\phi_{ij}=\int_{\bd S_i}dp=0$. Consider the minimal connection relative to the set $\{P_{ij}^n\, %,\ N_{ij}^n\}$ (see \cite{bcl} for   definitions), and denote it by $\{\Gamma_{i,\nu}\}_\nu\subset S_i$.  

%Since $P_{ij}^n=N_{kl}^m$ for some $l,m$ if $P_{ij}^n$ belongs to the common face of two adjacent %simplices $S_i$ and $S_k$, the construction of the minimal connection in every simplex $S_i$ gives %rise to a family of  disjoint polyhedral curves $\{\Gamma_k^h\}_k$ and a corresponding 2-form $q_h$ %verifying the conclusions of statement (v).

%%00393472300286 ganz

%%00393389955693 bardi

\medskip
\noindent{\bf Proof of (i).} 
By Lemma \ref{lem:dp} and Corollary \ref{cor:poincare} in the Appendix, it suffices to check that $dq_h=0$
in 
$\Omega_\delta$ and that $\int_{(\partial\Omega)_i} (q_h)_\top = 0$
for every connected component $(\partial \Omega_\delta)_i$ of $\partial \Omega_\delta$.

To do this, fix any $f\in C^\infty_c(\R^3)$, and note that 
(v),  \eqref{eq:minconbdy1}, \eqref{eq:minconbdy2}
imply that
\[
\begin{aligned}
\langle dq_h, \star f\rangle &= 
\langle q_h, d^* \star f\rangle = 
\langle q_h, \star d f\rangle = h \sum_i \int_{\Gamma_i} df + \sum_{i,j}\int_{\Gamma_{ij}} df  \\
&= \
 h \sum_{i, j, a} (\sign \phi_{ij} ) f(s^a_{ij}).
\end{aligned}
\]
Here all terms of the form $f(\tilde s^a_{ijk})$ have cancelled, since they occur twice, 
with opposite signs, in \eqref{eq:minconbdy1} and \eqref{eq:minconbdy2}. 
If $s^a_{ij} \in \Omega_\delta^P$, then our construction implies that there exists exactly
one $(i', j', a') \ne (i,j,a)$ such that $s^a_{ij} = s^{a'}_{i'j'}$, and moreover that 
$\sign \phi_{ij} = - \sign \phi_{i'j'}$. Thus all contributions from $\Omega_\delta^P$ vanish, and the above reduces to
\beq
\langle dq_h, \star f\rangle = 
h \sum_{\{ i, j, a: s^a_{ij}\in \partial \Omega_\delta^P\} } (\sign \phi_{ij} )\  f(s^a_{ij}).
\label{eq:qh.bdy1}\eeq
In particular, by considering $f\in C^\infty_c(\Omega_\delta)$ we see that $dq_h=0$ in $\Omega_\delta$.

Now fix some component $(\partial\Omega_\delta)_k$ of $\partial\Omega_\delta$. Then \eqref{eq:Od-OdP}
implies that 
\[
0 = \int_{W_k} d1 = \int_{\partial W_k} 1
= \int_{(\partial \Omega_\delta^P)_k} (q_h)_\top -
\int_{(\partial \Omega_\delta)_k} (q_h)_\top.
\]
Moreover, it follows from \eqref{eq:minconbdy2}, \eqref{eq:qh.bdy1}, and the definition of $(q_h)_\top$  
(see \eqref{qtop.def1} in the Appendix) 	that
\[
\int_{(\partial \Omega_\delta^P)_k} (q_h)_\top =
\sum_{ (i,j): T_{ij}\subset (\partial \Omega_\delta)_k}
h (\sign \phi_{ij}) m_{ij} .
\]
However, the definitions of $m_{ij}$ and  $\phi_{ij}$ imply that the above quantity 
equals
\[
\sum_{ (i,j): T_{ij}\subset (\partial \Omega_\delta^P)_k} \phi_{ij}
=
\sum_{ (i,j): T_{ij}\subset (\partial \Omega_\delta^P)_k} \int_{T_{ij}} dp 
=
\int_{(\partial\Omega_\delta^P)_k} dp = 0.
\]
Then, as remarked above, (i) follows from Lemma \ref{lem:dp} and Corollary \ref{cor:poincare}..

\smallskip
\noindent{\bf Proof of (iii).} 
We next estimate the mass of $q_h$. We will bound the mass on each region $R$,
and then sum up the estimates.
We begin by comparing the fluxes of $q_h$ and $dp$ across $\bd R$.

\begin{lemma}\label{lem:h1/4} Let $R$ be a region, and let $(dp)_\top$ and $(q_h)_\top$ be the tangential parts of $dp$ and $q_h$, respectively, on $\partial R$, ie, the measures in $\R^3$, supported in $\partial R$, defined as discussed in the Appendix, see \eqref{qtop.def1}.
Then there exists a constant $C = C(dp, \Omega_\delta^P)$, independent of $\eta$ and $h$,
such that 
\beq\label{eq:h1/4}
||(q_h - dp)_\top||_{W^{-1,1}(\R^3)}\le  C (\eta + h^{1/2}) \le C\eta\, ,
\eeq
\end{lemma}

\begin{proof}
First consider the case of  a pyramidal frustrum $P_{ij}$.

Arguing as in the proof of (i), we find from \eqref{eq:minconbdy2}  that $(q_h)_\top =  h \sum_{a} \sign(\phi_{ij}) \delta_{s^a_{ij}} -h \sum_{k,a} \sign(\phi_{ijk}) \delta_{\tilde s^a_{ijk}}$.
Similarly, the definition of $\phi_{ij}$ and the fact that $T_{ij}$ and $\tilde T_{ij}$ are parallel
implies that 
\[
\int_{\partial P_{ij}} f (dp)_\top = \
 \frac{\phi_{ij}} {|T_{ij}|} \int_{T_{ij}} f \,  d\calH^2 
 -
  \frac{\phi_{ij}} {|T_{ij}|} \int_{\tilde T_{ij}} f \,  d\calH^2 
  + O(\|f\|_\infty \eta)
\]
where the error term comes from neglecting $\partial P_{ij}\setminus (T_{ij}\cup \tilde T_{ij})$,
which has area bounded by $C\eta$.

Thus for any continuous $f$,
\begin{align*}
\int_{\partial P_{ij}} f (dp-q_h)_\top
&= 
\left[ \frac{\phi_{ij}} {|T_{ij}|} \int_{T_{ij}} f \,  d\calH^2 
 - h \sum_{a} \sign(\phi_{ij}) f(s^a_{ij})\right]\\
&\quad-
\left[ \frac{\phi_{ij}} {|T_{ij}|} \int_{\tilde T_{ij}} f \,  d\calH^2 
 - h \sum_{a,k} \sign(\phi_{ij}) f(\tilde s^a_{ijk})\right] +O(\|f\|_\infty \eta) .
 \end{align*}
We will consider only the second term on the right-hand side (which is slightly harder). We assume for simplicity
that $\phi_{ij}>0$; the case $\phi_{ij}<0$ is essentially identical.
Noting that $\frac {\phi_{ij}} {|\tilde T_{ij}|} = \frac {\phi_{ijk}} {|\tilde T_{ijk}|}$ and that $|\tilde T^a_{ijk}| = \ell_{ijk}^{-2}|\tilde T_{ijk}|$,
and using notation from the first step above, we have
\begin{align}
\int_{\tilde T_{ij}} f(dp-q_h)_\top 
\ &= \  \frac{\phi_{ij}} {|T_{ij}|} \int_{\tilde T_{ij}} f \,  d\calH^2  
- h \sum_{a,k} f(\tilde s^a_{ijk}) \nonumber\\
&= \ 
(\frac {\phi_{ij}} {|T_{ij}|} - \frac {\phi_{ij}}  {|\tilde T_{ij}|}) \int_{\tilde T_{ij}} f \,  d\calH^2  
 \ + \  
\sum_{k,a}
 \frac{\phi_{ijk}} {|\tilde T_{ijk}|} \int_{\tilde T^a_{ijk}} f  - f(\tilde s^a_{ijk})
  d\calH^2
\label{eq:bdyflatest} \\
&\ + \ 
\sum_{a,k } \left[ \frac{|\phi_{ijk}|}{\ell_{ijk}^2}  -  h \right]  f(\tilde s^a_{ijk})
+ 
\sum_{k }
 \frac {\phi_{ij}} {|T_{ij} |} \sum_k  \int_{\tilde T_{ijk}\setminus \cup_a  \tilde T^a_{ijk}} f \, \calH^2.\nonumber
\end{align}
It is clear from the definition of $\phi_{ij}$ that that $|\phi_{ij}| \le    \|dp\|_\infty | T_{ij}| \le C$,
and since by definition $ (\ell_{ijk}- 1)^2 <m_{ijk} = h^{-1}\phi_{ijk} \le \ell_{ijk}^2$,
\[
 | \frac { \phi_{ijk}}{\ell_{ijk}^2}  -  h|  \le \frac 2 {m_{ijk}} \frac {\phi_{ijk}}{\ell_{ijk}} \le \frac C{m_{ijk}}( h\phi_{ijk})^{1/2} \le C \frac{\sqrt{h}}{m_{ijk}}.
 \]
Similarly one checks that $|T_{ijk}\setminus \cup_a T^a_{ijk}| = |T_{ijk}| |1 - \frac{m_{ijk}}{\ell_{ijk}^2}| \le  %C |T_{ij}|\ell_{ij}^{-1}
 C |T_{ijk}| \sqrt{h}.
$
Note also that $|f(x) - f(\tilde s^a_{ijk})| \le \| df\|_\infty \mbox{diam} (\tilde T^a_{ijk}) \le C\| df\|_\infty  \sqrt h$ for $x\in  \tilde T^a_{ijk}$. Taking these into account, elementary calculations yield
\[
\left|\int_{ \tilde T_{ij}} f(dp-q_h)_\top 
 \right|
\le
C( \eta   +\sqrt{h})\|f \|_{W^{1.\infty}}.
\]
Since similar computations apply to $T_{ij}$, we deduce that
$|\int_{\partial P_{ij}}f (dp-q)_\top | \le C \eta\| f\|_{W^{1,\infty}}$ for every $P_{ij}$.
If the region $R$ is a simplex $\tilde S_i$,
then $\int_{\partial S_i} f(dp-_h)_\top$ is a sum of terms
of exactly the form $\int_{\tilde T_{ij}} f(dp-q_h)_\top$
already estimated (now with the opposite orientation) and so the conclusion follows
in this case as well.
\end{proof}

For future reference, we remark that the above proof shows that
that
\beq
\int_{T_{ij}} f(dp - q_h)_\top \le C \sqrt h\|f \|_{W^{1.\infty}},
\quad\quad\quad
\int_{\tilde T_{ij}} f(\frac{dp}{(1-\eta)^2} - q_h)_\top \le C \sqrt h\|f \|_{W^{1.\infty}}.
\label{eq:flatref1}\eeq
Indeed, every term on the right-hand side of 
\eqref{eq:bdyflatest} can be bounded by $Ch^{1/2}$ except for 
the term $(\frac {\phi_{ij}} {|T_{ij}|} - \frac {\phi_{ij}}  {|\tilde T_{ij}|}) \int_{\tilde T_{ij}} f \,  d\calH^2$.
This term is not present when one considers $T_{ij}$ rather than $\tilde T_{ij}$, 
and it is also not present if one considers $\tilde T_{ij}$, but
weighting the integrand as shown, since $(1-\eta)^2 = |\tilde T_{ij}|/|T_{ij}|$,
so that \eqref{eq:flatref1} follows from our earlier arguments.

We will need the following result about continuous dependence of the minimal connection upon its boundary datum.
\begin{lemma}\label{lem:mincon} Let $K$ be a compact convex  domain in $\R^3$, $\zeta$ a measure supported on $\partial K$ such that $\int_{\bd K}\zeta =0$. Then we have
$$
\min \{ ||\alpha||\equiv|\alpha |(K)\, , \ d\alpha=0\text{ in } K\, ,\ \alpha_\top=\zeta \ \text{on } \bd K\, \}\le C \, || \zeta ||_{W^{-1,1}(\R^3)}\, .
$$
\end{lemma}

The proof of this lemma is postponed to Section \ref{sect:mincon} in the Appendix. Let us apply
 Lemma \ref{lem:mincon} first with $K=P_{ij}$, $\zeta=(q_h-dp)_\top$ and let $\alpha_h$ be the measure 2-form that realizes the minimum.  By \eqref{eq:h1/4} and Lemma \ref{lem:mincon} we deduce $|\alpha_h|(P_{ij})\le C\eta$.
 
As remarked above, the restriction of $\Gamma_h$ to any region $R$ is a minimal connection, and as a consequence,
it follows from results proved in Brezis, Coron and Lieb  \cite{bcl} that 
$q_h\rest R$ has minimal mass among all $2$-form-valued measures $q'$  in
$R$ such that $(q')_\top = (q_h)_\top$ on $\partial R$ (not merely those
corresponding to a union of oriented line segments). We thus have
\beq\label{eq:qhequibd}
|q_h|(P_{ij})\le ||\alpha_h+dp||\le |\alpha_h|(P_{ij})+\int_{P_{ij}}|dp|\le \int_{P_{ij}}|dp|+C\eta\, .
\eeq
Next, applying Lemma \ref{lem:mincon} with $K=\tilde S_{i}$, $\zeta=(q_h-dp)_\top$
and arguing exactly as above, we obtain
 \beq\label{eq:qhequibdbis}
 |q_h|(\tilde S_i)\le  \int_{\tilde S_i}|dp|+C\eta\, .
 \eeq
 Statement (iii) follows by summing over all regions.

\smallskip

\noindent{\bf Proof of (ii).}
It suffices to show that for every region $R$,
\beq
\langle \varphi,   (dp - q_h)\rest R  \rangle\ = \ 
\int_R (\varphi, dp)  - \langle \varphi, q_h \rest R\rangle  \  \le \  C \eta \| \varphi \|_{W^{1,\infty}}
\label{eq:flatS}\eeq
for every $\varphi\in C^\infty_c(\Lambda^2\R^3)$. This is clear if $R = P_{ij}$, since $|P_{ij}| \le C\eta$ for all $i,j$, so that $\| dp\|_{L^1(P_{ij})} \le C \eta$, and hence
$|q_h|(P_{ij}) \le C\eta$ by \eqref{eq:qhequibd}.

If $R= \tilde S_i$ then we assume, after changing coordinates, that $dp = \lambda dx^2\wedge dx^3$ on $\tilde S_i$ for some $\lambda\in \R$. 
Now fix $\varphi \in C^\infty_c (\Lambda^2\R^2)$ 
and let $\Phi\in C^\infty_c(\R^3)$ be a function such that $(\star d \Phi, dx^2\wedge dx^3) = (\varphi ,dx^2\wedge dx^3)$ in $S_i$, and such that $\| \Phi\|_{W^{1,\infty}} \le C \| \varphi\|_{W^{1,\infty}}$. 
Indeed,  $(\star d\Phi,  dx^2\wedge dx^3) = \Phi_{x_1}$,
so we can take
\[
\Phi(x) := \chi(x) \int_{-\infty}^{x_1} \chi(x) \left(\varphi(s, x_2, x_3), dx^2\wedge dx^3\right) \,ds
\]
where $\chi\in C^\infty_c(\R^3)$ satisfies $\chi\equiv 1$ on $S_i$. 
Then clearly  $\langle dp \rest \tilde S_i, \varphi \rangle = 
\langle dp \rest \tilde S_i, \star d\Phi \rangle$  and it follows from the form of $dp$ and the definition (ie statement (v)) of $q_h$ that 
$\langle q_h \rest \tilde S_i, \varphi \rangle = 
\langle q_h \rest \tilde S_i, \star d\Phi \rangle$. Thus Lemma \ref{lem:h1/4} implies that
\begin{align*}
\langle \varphi,   (dp - q_h)\rest \tilde S_i \rangle
 = 
\langle \star d\Phi,   (dp - q_h)\rest \tilde S_i \rangle
 = 
\int_{\partial \tilde S_i} \Phi ( dp - q_h)_\top
\le C \eta \| \phi \|_{W^{1,\infty}} .
\end{align*}
Thus $\| (dp - q_h)\rest S_i \|_{W^{-1,1}(\R^3)} \le C  \eta$.

\medskip
\noindent{\bf Proof of (iv).}
The estimate $\| d^*\beta_h\|_{L^p(\Omega_\delta)} \le C_p |q_h|(\Omega_\delta) \le C$, $1\le p < 3/2$,
follows immediately from Corollary \ref{cor:poincare} in the Appendix.
Thus $d^*\beta_h$ is weakly precompact in these $L^p$ spaces, and we only need to identify the limit,
prove that it is unique,  and estimate its $L^2$ distance from $d^*\beta$.

To do this we will show that  $q_h\to q^\eta$ in $W^{-1,1}(\Omega_\delta)$, 
where $q^\eta=(1-\eta)^{-2}dp$ on $\tilde S_i$, while on $P_{ij}$, $q^\eta$ is defined to be the unique minimizer 
of the problem 
\beq
\min \{ |\alpha |(P_{ij})\, , \ d\alpha=0\text{ in } P_{ij}\, ,\ \alpha_\top=\zeta\ \text{on } \bd P_{ij}\, \}\, ,
\label{eq:qhmin}\eeq
where $\zeta=(dp)_\top$ on $T_{ij}$, $\zeta=(1-\eta)^{-2}(dp)_\top$ on $\tilde T_{ij}$ and $\zeta=0$ on the remaining faces of $\bd P_{ij}$. Since then $\beta^\eta = -\Delta^{-1}q^\eta$, the uniqueness of $\beta^\eta$ will follow,
and we will deduce the estimates of $\beta^\eta$ from the explicit form of $q^\eta$, which we find below.

We consider first a truncated pyramidal region $P_{ij}$, which is the harder case.
The uniform mass bounds \eqref{eq:qhequibdbis} imply that
$q_h\rest P_{ij}$ is precompact in $W^{-1,1}(\R^3)$. 
Let $q$ denote a limit of a convergent subsequence.
It follows from \eqref{eq:flatref1} that $(q_h)_\top$ on $\partial P_{ij}$
converges to  $\zeta$ as defined above,
and hence that $q_\top = \zeta$ on $\partial P_{ij}$. Next, if $q$ did not solve the minimization
problem \eqref{eq:qhmin},  we could use the estimate
$\| (q_h)_\top - \zeta \|_{W^{-1,1}} \le C \sqrt h$
(which is \eqref{eq:flatref1}) together with Lemma \ref{lem:mincon} to create a sequence  $q_h'$ such that 
$(q_h')_\top = (q_h)_\top$, and with $|q_h'|( P_{ij}) < |q_h|( P_{ij})$ for all small enough $h$, contradicting the minimality 
of $q_h$. Thus $q = q^\eta$, a minimizer of \eqref{eq:qhmin}.

We now argue that the  unique minimizer \eqref{eq:qhmin}  is given by
\beq
q^*(x) = a \frac{(x - b_i)_\ell}{( (x - b_i) \cdot \nu_{ij} )^3} \star dx^\ell
%\quad\quad\quad\mbox{ in }P_{ij}
\label{eq:qetaformula}\eeq
where $b_i$ denotes the barycenter of $S_i$, $\nu_{ij}$ is the unit normal to $T_{ij}$, and
$a\in \R$ is adjusted so that $q^*_\top = \zeta$.  (A  calculation
shows that such a number $a$ exists and also that $dq^* = 0$.) %, using the fact that $(dp)_T$ has constant density on $T_{ij}$ and $\tilde T_{ij}$, and that these densities differ only by a sign.)
The (unique) minimality of  $q^*$
now follows from  a calibration argument. We briefly recall the idea:
Let $f(x) = |x-b_i|$, so that $df = \sum\frac  {(x- b_i)_\ell}{|x-b_i|}dx^\ell$, and
$(\star df, q^*) = |q^*|$ in $P_{ij}$. For any other $2$-form valued measure $q'$ 
supported in $P_{ij}$
such that $dq'=0$ in $P_{ij}$ and $q'_\top =  \zeta $ on
$\partial P_{ij}$, we have
\[%\begin{align*}
|q^*|(P_{ij})  \ = \  \langle q^* \rest P_{ij}, \star d f\rangle \ = \ 
\int_{\bd P_{ij}} f \zeta
%= \int_{\bd P_{ij}} f q'_\top
 \ = \  \langle q' , \star d f\rangle \ \le  \ |q'|(P_{ij}),
\]%end{align*}
since $|\star df|\le 1$ everywhere. Hence $q^* $ is a minimizer. Furthermore,
if equality holds then, heuristically, $q'$ is parallel to $\star df$, or more precisely, $q'$ has the form
$\langle q', \psi\rangle = \int_{P_{ij}} (\frac{ (x-b_i)_\ell \star dx^\ell}{|x-b_i|}, \psi) d\mu'$ for some measure $\mu'$.
Then one can check that $q^*$ is the only measure-valued $2$-form  of this form such that $dq'=0$ in $P_{ij}$,
$q'_\top = \zeta$ on $\bd P_{ij}$. Hence $q^\eta = q^*$ as asserted. 

The proof that $q_h\rest \tilde S_i$ converges in $W^{-1,1}$ to $(1-\eta)^{-2}dp \rest  \tilde S_i$
can be carried out on exactly the same lines, except that the limit has a simpler form. It can also be proved
by arguing as in the proof of (ii), but using \eqref{eq:flatref1} instead of (iii).
Thus we have proved that $q_h \to q^\eta$ in $W^{-1,1}(\Omega_\delta^P)$.
 
From the explicit form of $q^\eta$, noting that $\sum_{i,j}|P_{ij}| \le C\eta$,
we see that
\beq\label{eq:qvsdp}
||q^\eta-dp||^2_{L^2(\Omega_\delta^P)}\le C\eta\, .
\eeq
Thus$\|d^*\beta^\eta-d^*\beta \|^2_2
= \|d^* \Delta_N^{-1} (q^\eta - dp) \|^2_2\le C\eta$, by \eqref{eq:qvsdp} and standard elliptic estimates. This concludes the proof of
statement (iv).

\medskip
\noindent{\bf Proof of (vi).} 
We prove now the separation properties of the polyhedral curves $\Gamma_h^\ell$. Let $L_1$ and $L_2$ be closed line segments of $\Gamma_h$,
with endpoints $s_1^\pm$ and $s_2^\pm$, and assume that $L_1$ and $L_2$ are disjoint, so that in particular $\{ s_1^\pm\} \cap \{s_2^\pm\} = \emptyset.$

If $L_1, L_2$ belongs to  non-adjacent regions of the family $\{ \tilde S_i, P_{ij} \}$ then
the conclusion is obvious, so we assume that this is not the case,
and we claim that 
\beq\label{eq:claim1}
\text{dist}\, (s_m^\pm, L_n)\ge c_2\eta h^{1/2}\, \ \quad\mbox{ for }m\ne n,  m,n\in \{1,2\}.
\eeq
To see this, let $F$ denote the face (some $T_{ij}$ or $\tilde T_{ij}$) containing $s_1^+$ say. If $F$ also contains an endpoint of $L_2$ (for example $s_2^+$) then by construction $L_2$ forms an angle of at least $c \eta$ with
$F$, and $|s_1^+ - s_2^+| \ge c h^{1/2}$, and so \eqref{eq:claim1} follows from elementary geometry.
The claim is still clearer if neither endpoint of $L_2$ is contained in
$F$.

It is evident that  \eqref{eq:claim1} implies (vi) if $L_1$ and $L_2$ belong to distinct but adjacent regions. If $L_1$ and $L_2$ belong to the same region, then in view of the minimality property of $q_h$, we obtain statement (vi) from \eqref{eq:claim1} and the following Lemma:

\begin{lemma}\label{lem:claim2} Let $\{ s_m^\pm\}_{m=1,2}$ 
satisfy  $|s^+_1 - s^-_1| +|s^+_2 - s^-_2| \le |s^+_1 - s^-_2|+|s^+_2 - s^-_1|$.
Also, let $L_m$ be the segment joining $s^+_m$ and $s^-_m$, for $m=1,2$. Then
\beq\label{eq:claim2}
\dist (L_1,L_2)\ge \frac 1 {\sqrt{2}}
 \min_{m\ne n}\ 
\dist (s_m^\pm, L_n).
\eeq
 \end{lemma}
\begin{proof}
Let $Q_m\in L_m, m=1,2$ be such that  dist$\,(L_1,L_2)=|Q_1-Q_2|=d$. If either $Q_m$ is an endpoint then the conclusion is clear, so we assume that both are interior points, in which case the segment from $Q_1$ to $Q_2$ is orthogonal to both $L_1,L_2$. We may then assume without loss of generality that
the midpoint $\frac{Q_1+Q_2}{2}$ is the origin,  and that $Q_1=(0,0,\frac{d}{2})$, $Q_2=(0,0,-\frac{d}{2})$, and  moreover that $L_1$ and $L_2$ are parallel to the directions $(\cos\theta,\sin\theta, 0)$, $(\cos\theta,-\sin\theta,0)$ respectively, for some $\theta$
. 
Define $\tilde s_1^\pm = (\pm \lambda\cos\theta, \pm\lambda\sin\theta,\frac{d}{2}), 
\tilde s_2^\pm=(\pm\lambda\cos\theta, \mp\lambda\sin\theta,-\frac{d}{2})$,
for $\lambda>0$ chosen so that one of the $\tilde s_m^\pm$ coincides with the closest point
to $0$ among the original endpoints. 

Our hypothesis and the triangle inequality imply that
$ |\tilde s^+_1 - \tilde s^-_2|+|\tilde s^+_2 - \tilde s^-_1|
\ge |\tilde s^+_1 - \tilde s^-_1| +|\tilde s^+_2 - \tilde s^-_2|$, which reduces to
$$
2\sqrt{ 4\lambda^2 \cos^2\theta + d^2} \ge 4 \lambda =  
2\sqrt{ 4\lambda^2( \cos^2\theta + \sin^2\theta)},\quad\quad\mbox{ so that }
d^2\ge  4 \lambda \sin^2\theta.$$
On the other hand, assuming for concreteness that $\tilde s^+_1$ agrees with the original endpoint
$s_1^+$, then since $\tilde s_2^+ \in L_2$, we use the above inequality to find that
we 
\[
\dist( s_1^+, L_2)\le | \tilde s_1^+ -  \tilde s_2^+|   \ = \ \sqrt {4 \lambda^2 \sin^2\theta + d^2}
\le   \sqrt{2} d.
\]
 \end{proof}
 
\noindent{\bf Proof of (vii).} Finally, suppose that $L_1$ and $L_2$ are adjacent, and that $L_1$ precedes $L_2$ in the ordering induced by their respective orienting unit tangents $\tau_1$, $\tau_2$. 
Decompose $\tau_i$ as $\tau_i^\perp + \tau_i^\parallel$, where for $i=1,2$, $\tau_i^\perp$
is orthogonal to the face $T_{ij}$ that contains the common endpoint of $L_1$ and $L_2$.
The orientation conventions imply that $\tau_1^\perp \cdot \tau_2^\perp >0$,
and, as noted above, each segment forms an angle of at least $c\eta$ with $T_{ij}$, which
implies that $|\tau_i^\perp|\ge c\eta$ for $i=1,2$. Statement (vii) follows directly.
 
\medskip
The proof of Proposition \ref{prop:1} is now complete.
\end{proof}

%%%%%%%%%%%%%%%%%%%%%%%%%%%%%%%%%%%%%%%%%%

\subsection{Pointwise estimates for $d^*\beta_h$}\label{S3.5}

%%%%%%%%%%%%%%%%%%%%%%%%%%%%%%%%%%%%%%%%%

Let $G(x)=(4\pi)^{-1}|x|^{-1}$ be the Poisson kernel in $\R^3$. We may write 
\beq\label{eq:decompdbeta}
d^*\beta_h=d^*(G*q_h)+\Psi_h\, \qquad \Psi_h=d^*(-\Delta^{-1}_Nq_h-G*q_h)\, .
\eeq
In view of statement (i), we deduce that $d\Psi_h=d^*\Psi_h=0$ in $\Omega_\delta$, i.e. $-\Delta\Psi=0$ in $\Omega_\delta$ and $\Psi_N=-d^*(G*q_h)_N$ on $\bd\Omega_\delta$.  From the decomposition \eqref{eq:decompdbeta} we will deduce pointwise and integral estimates for $d^*\beta_h$. 

We begin with the term $d^*(G*q_h)=G*d^*q_h$.
The integral representation of $d^*(G*q_h)$ through the Biot-Savart law takes the form
\beq\label{eq:biotsavart}
d^*(G*q_h)(x) =
%\left(h\ep_{ijk} \int_{\Gamma_h} G_{x_j}(\cdot - y) dx^k \right) dx^i =
h\sum_{\ell=1}^{m(h)} \sum_{i,j,k =1}^3 \frac1 {4\pi} dx^i 
\ep_{ijk} \int_{\Gamma_h^\ell}\frac{(x_j-y_j)dy^k}{|x-y|^3}\, ,
\eeq
where $\ep_{ijk}$ is the usual totally antisymmetric tensor. This can be justified for example
by noting that
 $\langle d^*(G*q_h), \varphi\rangle = \langle q_h, G * d\varphi\rangle$, since $G$ is even,
and then using  statement (v) of Proposition \ref{prop:1} to explicitly write out the right-hand side. 
%Alternatively, in view of statement (v)  and \eqref{eq:hatt} in the Appendix we have
%$\star q_h=h\widehat{\Gamma_h}=h\sum \widehat{\Gamma^\ell_h}$, so that
%\beq\label{eq:biotsavart0}
%G*d^*q_h(x)=G*(\star d\widehat{\Gamma_h})(x)=(\star dG)*\widehat{\Gamma_h(x)}=h\int_{\Gamma_h}\star dG(x-\cdot), 
%\eeq
%Either way, 
From \eqref{eq:biotsavart} we readily deduce 
\begin{lemma}\label{lem:pointwise1} Let $l_1,l_2>0$, $L=\{(0,0,z)\, ,\ -l_1\le z\le l_2\}\subset\R^3$, $q$ the associated measure 2-form, i.e. $\langle q\, ,\varphi\rangle=\int_{L}\star\varphi$ for $\varphi\in C^0(\Lambda^2\R^3)$. Then
\beq\label{eq:pointwise1}
d^*(G*q)=\frac{xdy-ydx}{4\pi (x^2+y^2)}\left(\frac{l_2-z}{\sqrt{x^2+y^2+(l_2-z)^2}} +\frac{l_1+z}{\sqrt{x^2+y^2+(l_1+z)^2}}\right)\, .
\eeq
As a result, %In particular, we have, for $p_0=(x_0,y_0,z_0)\in\R^3$, $-l_1\le z_0\le l_2$,
\beq\label{eq:pointwise1bis}
|d^*(G*q)(p_0)|\le  %\frac{1}{2\pi\sqrt{x_0^2+y_0^2}}=
\frac{1}{2\pi\cdot\text{\rm dist}\,(p_0,L)}\,  \quad\quad\quad \mbox{ for every $p_0\in \R^3$.}
\eeq
%while if $z_0\not\in (-l_1,l_2)$ we have
%\beq\label{eq:pointwise2}
%|d^*(G*q)(p_0)|\le \frac{\sqrt 2}{2\pi \cdot\text{\rm dist}\, (p_0,\bd L)}\le \frac{\sqrt 2}{2\pi\cdot\text{\rm dist}\,(p_0, L)}\, . \eeq
\end{lemma}
\begin{proof}
We obtain \eqref{eq:pointwise1} by particularizing \eqref{eq:biotsavart} to the case $\Gamma_h=L$. 
We easily deduce \eqref{eq:pointwise1bis} from \eqref{eq:pointwise1} if $p_0 = (x_0, y_0, z_0)$ with
$-l_1 \le z_0 \le l_2$, in which case $\dist(p_0, L ) = \sqrt{x_0^2 + y_0^2}$.
If $z_0>l_2$ then, writing $r_0 = (x_0^2+y_0^2)^{1/2}$,
since $\lambda \mapsto \frac \lambda{\sqrt{r_0^2 +\lambda^2}} < 1$ is an increasing function
and $0 < z_0-l_2 < z_0+l_1$,
we find from \eqref{eq:pointwise1} that
\begin{align*}
|d^*(G*q)(p_0)|
&\le
\frac 1 {4\pi r_0} \left( 1 - \frac{z_0- l_2}{\sqrt{r_0^2+(l_2-z_0)^2}}\right)\\
&=
\left(  \frac {\sqrt{r_0^2+(l_2-z_0)^2} - (z_0- l_2)}{r_0}\right)\left( \frac 1 {4 \pi \dist(p_0, L)}\right),
%\\&\le
%\frac 1 {4\pi r_0} \left( \frac {r_0^2}{2|z_0-l_2|}\right)\left( \frac 1 {\sqrt{r_0^2+(l_2-z_0)^2}}\right) 
%
%&\frac{\sqrt{x_0^2+y_0^2}}{2\pi |z_0-l_2|^{2}}\le \frac{\sqrt 2}{2\pi \sqrt{x_0^2+y_0^2+(z_0-l_2)^2}}\, ,
\end{align*}
and \eqref{eq:pointwise1bis} follows, since $\sqrt{a^2+b^2}  \le a+b$ for $a,b\ge 0$.
The same reasoning of course holds if $z_0<-l_1$.
\end{proof}
\begin{lemma}\label{lemma:pointwise2} Let $x\in\Omega_\delta$ be such that $\text{\rm dist}\,(x,\Gamma_h)\le \frac{c_0}{2}\eta h^{1/2}$, where $c_0>0$ is defined in statement (vi) of Proposition \ref{prop:1}. Then there exists a constant $K>0$ independent of $\eta,h$ such that  if $\eta\le 1$, then
\beq\label{eq:pointwise3}
|d^*\beta_h(x)|\le \frac{h}{2\pi\cdot\text{\rm dist}\,(x,\Gamma_h)}+\frac{K}{\eta^2}\quad\text{\rm if  }\ \dist(x,\cup_{i,j}\bd \tilde S_i\cup \bd P_{ij})\ge \frac{c_0}{2}\eta h^{1/2}\, ,
\eeq
\beq\label{eq:pointwise4}
|d^*\beta_h(x)|\le \frac{h}{\pi\cdot\text{\rm dist}\,(x,\Gamma_h)}+\frac{K}{\eta^2}
\quad\text{\rm if  }\ \dist(x,\cup_{i,j}\bd \tilde S_i\cup \bd P_{ij})< \frac{c_0}{2}\eta h^{1/2}\, .
\eeq
\end{lemma}
\begin{proof}The definition \eqref{eq:decompdbeta} of $\Psi_h$ implies that for any measure-valued $2$-form $q$, 
\beq
|d^* \beta_h| \le  |d^*(G*q)|+|d^*(G*q_h- G*q)| + |\Psi_h|\, .
\label{eq:pointwise5}\eeq
Fix $x\in \Omega_\delta \setminus \Gamma_h$ and let $r=\frac{c_0}{2}\eta h^{1/2}$. Define
a measure-valued $2$-form by
$\langle q, \varphi \rangle = h \sum_{ \{s : B_r(x)\cap L^s_h \ne \emptyset \}} \int_{L^s_h} \star\varphi $,
where $\{L^s_h\}$ is the collection of line segments whose union gives $\Gamma_h$,  see
Proposition \ref{prop:1} (v). 
By  Proposition \ref{prop:1} (vi), there is at most $1$ term in the sum that defines $q$ if $\dist(x, \cup_{i,j}\bd P_{ji}\cup \bd \tilde S_i) \ge r$, and otherwise at most $2$ terms. Then
$|d^*(G*q)|$ is estimated via Lemma \ref{lem:pointwise1} to give the first term on the right-hand sides of \eqref{eq:pointwise3} and \eqref{eq:pointwise4} respectively, and we must show that the other two terms in \eqref{eq:pointwise5} can be bounded by $K/\eta^2$.

Interior regularity for harmonic functions, %Sobolev embedding and $L^p$ estimates for the Laplace operator 
together with Proposition \ref{prop:1}, statements (iii), (iv) 
allow us to fix some $q\in (1, 3/2)$ and argue as follows:
%yield a uniform bound for $\Psi_h$ in $\Omega$ as follows:
\beq\label{eq:intregpsi}
\begin{aligned}
||\Psi_h||_{L^\infty(\Omega)} &\le C|| \Psi_h ||_{W^{2,2}(\Omega)}
\\
&\le C||\Psi_h||_{L^q(\Omega_{\delta})}\\
&=C ||d^*\beta_h-d^*(G*q_h)||_{L^q(\Omega_{\delta})}\\
&\le C(1+C\eta)||dp||_{L^1(\Omega_{\delta})}\le C\, .
\end{aligned}
\eeq
To estimate the remaining term in \eqref{eq:pointwise5}, observe that
\beq\label{eq:linfty}
\begin{aligned}
|d^*(G*q_h-  G*q)(x)| &\le \frac{6}{4\pi}h\sum_{k=1}^3\sum_{\ell=1}^{m(h)}\int_{\Gamma_h^\ell\cap B_r(x)^c}\frac{dy^{k}}{|x-y|^2}\, \\
&\le C\sum_{k=1}^3\int^{M}_{-M}\left(\sum_{\ell'=1}^{m'(h)}\frac{h}{|x-y^t_{\ell'}|^2}\right) dt\, \\
\end{aligned}
\eeq
where $M>0$ is such that $\Omega_\delta\subset B_M(0)$ and $\{y_{\ell'}^t\}_{\ell'}=\cup_\ell\Gamma_h^\ell\cap \{y_k=t, |y-x|>r\}$, for $|t|\le M$.
%%%%%%%%%%%%%%%%%%%%%%%%%%%%%%%%
%%%%%%%%%%%%%%%%%%%%%%%%%%%%%%%%
%%%%%%%%%%%%%%%%%%%%%%%%%%%%%%%%
%%%%%%%%%%%%%%%%%%%%%%%%%%%%%%%%
For every $k$ and $t$ , %, if we let $x^t^ denote the projection of $x$ onto the plane \{ x_k = t\}$, we compute
\begin{align*}
\sum_{\ell'=1}^{m'(h)}\frac{h}{|x-y^t_{\ell'}|^2}
\ &\le \ 
%\sum_{\ell'=1}^{m'(h)}\frac{h}{|x-x^t|^2+ |x^t-y^t_{\ell'}|^2}
%\\ &= \ 
\sum_{j=1}^{M/r}  \frac h { r^2j^2} \ \# \{ \ell' : jr \le  |x-y^t_{\ell'}| < (j+1)r \}.
\end{align*}
Consider the collection of (2 dimensional) balls
\[
\{ z : z^k = t, |z - y^t_{\ell'}| < r \},\quad\mbox{ for }y^t_{\ell'}\mbox{ such that }jr \le |x-y^t_{\ell'}| < (j+1)r.
\]
These balls are pairwise disjoint by Proposition \ref{prop:1} (vi), and are contained in the annulus
$\{ z : z^k = t, (j-1)r \le |x-z| < (j+2)r\}$, which has area $(6j+3) \pi r^2$. Thus
$\# \{ \ell' : jr \le  |x-y^t_{\ell'}| < (j+1)r \}\le 6j+3$ for all $j$. In addition, if we write 
$x^t$ for the projection of $x$ onto the plane $\{ z^k=t\}$, then
$\# \{ \ell' : jr \le  |x-y^t_{\ell'}| < (j+1)r \}=0$ if $(j+1)r < |x - x^t|$.
Then elementary estimates lead to the conclusion
\[
\sum_{\ell'=1}^{m'(h)}\frac{h}{|x-y^t_{\ell'}|^2}
\le C \frac h{r^2} \log(\frac M{|x-x^t|}).
\]
Substituting this into \eqref{eq:linfty},
we see that $|d^*(G*q_h-  G*q)(x)|  \le C\frac h {r^2} = C (c_0\eta)^{-2}$, completing the proof of the lemma
\end{proof}

The next lemma shows that we get uniform estimates of certain quantities  if we mollify on a scale comparable to the minimum distance between the discretized vortex lines.

\begin{lemma}\label{lem:convolbeta} 
Let $0<\mu<1$ and $r = \mu c_0\eta h^{1/2}$, for $c_0$ as in statement (vi) of Proposition \ref{prop:1}.
Then there exists a nonnegative radial function  $\phi$ supported in the unit ball, with 
$\int \phi = 1$, 
and such that in addition $\phi_r(x) := r^{-3}\phi(x/r)$ satisfies
\beq\label{eq:convolbeta}
||\phi_r*d^*\beta_h||_{W^{1,p}(\Lambda^1\Omega)}\le K
\eeq
for any $p<\infty$, where  $K=K(\mu,\eta, \|\phi \|_\infty,p )$ is
independent of $h$.
\end{lemma}

\begin{proof} 
First, let $\psi$ be any radial mollifier with support in the unit ball, such that $\psi \ge 0$ and  $\int\psi = 1$,
and let $\psi_r(x) := r^{-3}\psi(x/r)$.
Then for $x\in\Omega_\delta$, in view of statement (vi) of Proposition \ref{prop:1}, either $B_r(x)\cap\Gamma_h=\emptyset$ or $B_r(x)\cap\Gamma_h=B_r(x)\cap\{L_1\}$, or $B_r(x)\cap\Gamma_h=B_r(x)\cap\{L_1, L_2\}$, where $L_i$ are segments of $\Gamma_h$. Hence we have
\beq\label{eq:boundconvolq}
|\psi_r*q_h(x)|\le r^{-3}||\psi||_{\infty} \sum_i  h|L_i\cap B_r(x)|\le 4hr^{-2}||\psi||_{\infty}\le \frac{4}{(c_0\mu\eta)^2}|| \psi ||_\infty\, .
\eeq
Now fix open sets
$\Omega =\Omega_3 \Subset \Omega_2\Subset \Omega_1 \Subset \Omega_0=  \Omega_\delta$
and functions $\chi_m$ for  $m=1,2,3$ such that 
$\chi_m\in C^\infty_c(\Omega_{m-1})$
and $\chi_m\equiv 1$ on an open neighborhood of  $\bar \Omega_m$.
Fix a mollifier $\psi^1$ as above, but such that spt$(\psi^1)\subset B_{1/3}$,
and define $\psi^2 = \psi^1*\psi^1$ and $\psi^3 = \psi^1*\psi^2$. Thus $\psi^m$ is radial
with support in $B_1$ for $m=1,2,3$, so that \eqref{eq:boundconvolq} applies to $\psi^m_r$.
Now write $\zeta_0 = d^*\beta$, and
for $m=1,2,3$  define $\zeta_m = \psi^1_r* (\chi_m \zeta_{m-1})$.

If $h$, and thus $r$, is small enough (which we will henceforth take to be the case), then 
\beq
\mbox{$\zeta_m= \psi^1_r * \zeta_{m-1} = \psi^m_r*d^*\beta$ on $\Omega_m$, and
$\zeta_m$ has support in $\Omega_{m-1}$.}
\label{eq:zetam}\eeq
We claim that
\beq\label{eq:gaffney0}
\begin{aligned}
\|d\zeta_m\|_{L^p(\Omega_{m-1})}
&\le C_m \|  \zeta_{m-1}\|_{L^p(\Omega_{m-1})} + C(p, \mu, \psi^1, \Omega_\delta)\ , \\
\|d^*\zeta_m\|_{L^p(\Omega_{m-1})}&\le C_m \|  \zeta_{m-1}\|_{L^p(\Omega_{m-1})} .
\end{aligned}\eeq
To see these, note first that
$d\zeta_m =\psi^1_r* (d\chi_m \wedge \zeta_{m-1})+\psi^1_r* (\chi_m d\zeta_{m-1})$.
Then Jensen's inequality implies that
\[
\|\psi_1* (d\chi_m \wedge \zeta_{m-1})
\|_{L^p(\Omega_{m-1})} 
\le
\| d\chi_m \wedge \zeta_{m-1} \|_{L^p(\Omega_{m-1})} \le 
 C_m \|  \zeta_{m-1}  \|_{L^p(\Omega_{m-1})} .
\]
We estimate $\psi^1_r* (\chi_m d\zeta_{m-1})$ first in the case $m=1$, when it follows from statement (i) of Proposition \ref{prop:1} that $\psi^1_r* (\chi_1 d\zeta_{0})= \psi^1_r * (\chi_1 q_h)$.
Then arguing as in \eqref{eq:boundconvolq} we find that for any $p<\infty$,
\[
\| \psi^1_r* (\chi_1 q_h) \|_{L^p(\Omega)} \le 
C(p, \Omega_\delta)\| \psi^1_r* (\chi_1 q_h) \|_{L^\infty(\Omega)}  \le C(p,\psi^1, \Omega_\delta)(c_0\mu\eta)^{-2}.
\]
proving the first part of \eqref{eq:gaffney0} for $m=1$. 
For $m=2,3$, %since in general $\| f * (gh)\|_p \le \| g\|_\infty \|\, |f| *|h| \, \|_p$ 
\[
\| \psi^1_r* (\chi_m d\zeta_{m-1}) \|_{L^p(\Omega_{m-1})} \le 
%\| \chi_m  d\zeta_{m-1}   \|_{L^p(\Omega_{m-1})} \le 
\| d\zeta_{m-1}   \|_{L^p(\Omega_{m-1})} \overset{\eqref{eq:zetam}}=
\| \psi_{m-1}*q_h\|_{L^p(\Omega_{m-1})}
\]
and we conclude  \eqref{eq:gaffney0} much as in the case $m=1$. 
The second claim of \eqref{eq:gaffney0} is similar but easier, since \eqref{eq:zetam}
implies that $d^*\zeta_m = \psi^1_r* [\star d \chi_m \wedge \star \zeta_{m-1}]$,
so that $\| d^* \zeta_m \|_p \le \| |d\chi_m| \, |\zeta_{m-1}|\|_{L^p(\Omega_{m-1})} 
\le C_m  \|  \zeta_{m-1}\|_{L^p(\Omega_{m-1})} $.

Now recall  the 
Gaffney-G\aa rding inequality 
\beq\label{eq:gaffney}
\| \zeta  \|_{W^{1,p}(U)}\le C_p(U)\left(\| \zeta\|_{L^p(U)}+\|d\zeta\|_{L^p(U)}+\|d^*\zeta\|_{L^p(U)}\right)\, ,\ 1<p<+\infty\, ,
\eeq 
valid for a differential form $\zeta$ with compact support in $U\subset\R^n$.
Applying this to $\zeta_m$,  taking into account \eqref{eq:gaffney0} and
noting that $\| \zeta_m\|_{L^p} \le \|\zeta_{m-1}\|_{L^p}$, we find that
\begin{align}
\| \zeta_m \|_{W^{1,p}(\Omega_{m-1} )}
& \le
C \| \zeta_{m-1} \|_{L^p(\Omega_{m-1} )} +C.
\label{eq:gaffney1}
\end{align}
Recall that  Proposition \ref{prop:1}, statement (iv),
provides uniform
estimates of $\zeta_0 = d^*\beta$ in $L^p(\Omega_0)$ for every $p<3/2$,
so \eqref{eq:gaffney1} implies uniform estimates of $ \|\zeta_1\|_{W^{1.p}(\Omega_0)}$
for every $p<3/2$, and hence of $ \|\zeta_1\|_{L^p(\Omega_0)}$ for ever $p<3$.
Iterating this argument twice more and recalling \eqref{eq:zetam},
we find that \eqref{eq:convolbeta} holds with $\phi  = \psi^3$.
\end{proof}

%%%%%%%%%%%%%%%%%%%%%%%%%%%%%%%%%%%%
%
%                                       PREVIOUS VERSION
%
%%%%%%%%%%%%%%%%%%%%%%%%%%%%%%%%%%%%%

%%%%%%%%%%%%%%%%%%%%%%%%%%%%%%%%%%%%%%%%%%%
%
%    END PREVIOUS VERSION
%
%%%%%%%%%%%%%%%%%%%%%%%%%%%%%%%%%%%%%%%%%%

\subsection{Construction of the sequence $u_\ep$ in case $g_\ep\ge\logeps^2$}\label{sect:3.6}

Assume that the sequence $g_\ep$ satisfies either $g_\ep = \logeps^2$ or $\logeps^2\ll g_\ep \ll \ep^{-2}$.
Suppose that we are given $(J,v)\in \calA_0$ as defined in \eqref{eq:calA0def}, and moreover
that $J = \frac 12 dv$ if $g_\ep = \logeps^2$, and that $J = 0$ if $\logeps^2\ll g_\ep \ll \ep^{-2}$.

Set $p=\frac{1}{2\pi}v$. Fix $\delta>0$ and let $p_\delta$ be the piecewise linear approximation provided by Lemma \ref{lem:fem}, and recall  the Hodge decomposition  $p_\delta=\gamma+d\alpha+d^*\beta$ in $\Omega_\delta$ introduced in Section \ref{S:hodge}.
Fix $\eta>0$, and $h = h_\ep  = (g_\ep)^{-1/2}$, and let $q_h$ be the discretized vorticity, with support $\Gamma_h$, and $\beta_h=-\Delta^{-1}_Nq_h$ the approximation to $\beta$ constructed in  Proposition \ref{prop:1}. 

As we discuss in Remark \ref{rem:green}, if $c$ is any cycle in $\Omega_\delta\setminus\Gamma_h$, then  $h^{-1}\int_c d^*\beta_h$ is an integer for every $h$.
Thus, if we fix $\bar x\in\Omega$  and let $c_{\bar x,x}$ denote a path in
$\Omega_\delta\setminus \Gamma_h$
from $\bar x$ to $x$, it follows that
\beq\label{eq:rmodz}
\phi_h(x)
:=
\frac{1}{h}\int_{c_{\bar x,x}}d^*\beta_h\quad\text{is well-defined function }
\Omega_\delta\setminus \Gamma_h\to \R/\Z\, ,
\eeq
independent of the choice of $c_{\bar x,x}$, and is hence well-defined a.e. in $\Omega$. 

Moreover, according to Lemma \ref{L:H1N}, we may write
$\gamma=\sum_{j=1}^\kappa a_j\cdot d\phi_j$, where $\phi_{j}$ is well-defined in $\R/\Z$ for $j=1,...,\kappa$. For any $j$ let $n_j=[h^{-1}a_j]\in\Z$ be the integer part of $h^{-1}a_j$, and consider $h^{-1}\gamma_h\equiv d\psi_h=\sum_{j=1}^\kappa n_j d\phi_j$, so that $\psi_h$ is well-defined in $\R/\Z$.
Let finally $\alpha_h=h^{-1}\alpha$.
 The map
\beq\label{eq:vh}
v_h=\exp(i2\pi(\phi_h+\psi_h+\alpha_h))\, 
\eeq
is  thus a well-defined map $\Omega_\delta \to S^1$, with
\beq\label{eq:jvh}
jv_h=2\pi (d\phi_h+d\psi_h+d\alpha_h)=\frac{2\pi}{h}(d^*\beta_h+\gamma_h+d\alpha)
%\ + \ \sum_{j=1}^k (a_j - h n_j) d\phi_j
\eeq
and $Jv_h=\frac \pi h dd^*\beta_h= \frac \pi h \cdot q_h$. Let now
\beq\label{eq:roep}
\rho_\ep(x)\equiv\rho_{\ep,h}(x)=\min\{\frac{\text{dist}\,(x,\Gamma_h)}{\ep}\, ,\ 1\}\, ,
\eeq
for $\Gamma_h$ as in Proposition \ref{prop:1}, statement (v)
and set finally
\beq\label{eq:uep}
u_\ep\equiv u_{\ep,h}=\rho_\ep\cdot v_h\, .
\eeq

%Consider the Hodge decomposition \eqref{eq:hodger3} $p_\delta=d^*\beta+d\alpha$.  consider the discretized forms $d^*\beta_h$ constructed in Proposition \ref{prop:1}, and set $\alpha_h=h^{-1}\cdot\alpha$. For any simple closed loop $\gamma\subset\R^3$ we deduce, from previous Remark \ref{rem:biotsavart}, and \ref{rem:green} in the Appendix, more particularly from \eqref{eq:biotsavart0} and \eqref{eq:biotsavart3}, \eqref{eq:biotsavart4},
%\beq\label{eq:quantized}
%\begin{aligned}
%\int_\gamma d^*\beta_h &=h\int_{\gamma\times\Gamma_h}\!\!\!\!\!\!\star dG(x-y)\ \in h\Z\, .
%\end{aligned}
%\eeq

%%%%%%%%%%%%%%%%%%%%%%%%%%%%%%%%%%%%%%%%%%%%
\subsection{Completion of proof of \eqref{eq:gammalimsup} in case $g_\ep\ge\logeps^2$}\label{sect:completion1}
%%%%%%%%%%%%%%%%%%%%%%%%%%%%%%%%%%%%%%%%%%%%%

We first claim that
\beq
\frac{ju_\ep}{\sqrt{g_\ep}} \rightharpoonup 2\pi ( d\alpha + d^*\beta^\eta +\gamma) \quad\mbox{weakly in }L^q\mbox{ for every } q\in (1,  3/2).
\label{eq:convju2}\eeq 
for $\beta^\eta$ as in statement (iv) of Proposition \ref{prop:1}.
To see this we write
%\[
%\frac{ju_\ep}{\sqrt{g_\ep}} = h   \rho_\ep^2 j v_h = 2\pi  \rho_\ep^2(d^*\beta_h + \gamma_h + d\alpha).
%\]
% (where $h=h_\ep$) so that  %if we fix $p, q, r$ such that  $1\le p<q<3/2, \frac 1q + \frac 1r = \frac 1p$, 
 \beq
\frac{ ju_\ep}{\sqrt {g_\ep}}  = 2\pi ( d^*\beta_h +  \gamma_h +\alpha) + 
2\pi  ( \rho_\ep^2 - 1)( d^*\beta_h + \gamma_h + d\alpha) .
\label{eq:tst}\eeq
It is clear from the definition of $\gamma_h$ that $\gamma_h \to \gamma$
uniformly  as $\ep$ (and thus $h$) tend to $0$, and we know from Proposition \ref{prop:1}
that $d^*\beta_h \rightharpoonup d^*\beta^\eta$ in the relevant $L^q$ spaces. So we only need to show that the last term in \eqref{eq:tst} vanishes.
For this, we use statements (vi), (v), and (iii) of Proposition \ref{prop:1} to see that 
\beq
|\{\text{dist}\,(x,\Gamma_h)\le\ep\}| \le C \ep^2 |\Gamma_h|
=  C \frac{\ep^2}h |q_h|(\Omega_\delta) \le C  \frac{\ep^2}h . % ( \eta+  ||dp_\delta||_{L^1(\Lambda^2\Omega_\delta)}).
\label{eq:tube}\eeq
It easily follows from this and from the definition of $\rho_\ep$ that $(\rho_\ep^2-1)\to 0$ in $L^r$ for every $r<\infty$. 
Thus, fixing $q\in (1,3/2)$ and $r$ such that $\frac 1q + \frac 1r =1$, in view of uniform  estimates
of $\| d^*\beta_h\|_q$ in Proposition \ref{prop:1} (iv), we find from H\"older's inequality that  $( \rho_\ep^2 - 1)( d^*\beta_h + \gamma_h + d\alpha) \to 0$
in $L^1$ as $\ep \to 0$, proving \eqref{eq:convju2}.

We now turn to the proof of the upper bound. Since $h = g_\ep^{-1/2}$,
we have 
\beq\label{eq:euep}
\frac{E_\ep(u_\ep;\Omega)}{g_\ep}=\frac{h^2}{2}\int_{\Omega}|\nabla \rho_\ep|^2+\rho_\ep^{2}|jv_h|^2+\frac{W(\rho_\ep)}{\ep^2}\, .
\eeq
Let us estimate the various terms contributing to ${g_\ep}^{-1}E_\ep(u_\ep;\Omega)$. 
First note that
\[
\frac{h^2}{2}\int_{\Omega}|\nabla\rho_\ep|^2+\frac{W(\rho_\ep)}{\ep^2}
\le \frac{Ch^2}{\ep^2}|\{\text{dist}\,(x,\Gamma_h)\le\ep\}|
\]
for
$C = \frac 12(1+ \| W\|_{L^\infty(B_1)})$.
It follows from this and \eqref{eq:tube} that
\beq\label{eq;boundro}
%\begin{aligned}
\frac{h^2}{2}\int_{\Omega}|\nabla\rho_\ep|^2+\frac{W(\rho_\ep)}{\ep^2}
\ \le \ 
%\le \frac{h^2}{\ep^2}|\{\text{dist}\,(x,\Gamma_h)\le\ep\}|\\ &\le \frac{Ch^2|\Gamma_h|\ep^2}{\ep^2}
Ch\, .%( \eta + ||dp_\delta||_{L^1(\Lambda^2\Omega_\delta)}) \, .
%\end{aligned}
\eeq
Moreover,
\beq\label{eq:boundjvh}
\frac{h^2}{2}\int_{\Omega} \rho_\ep^2|jv_h|^2=2\pi^2\int_{\Omega}\rho_\ep^2(|d^*\beta_h|^2+|d\alpha+\gamma_h|^2+2\, d^*\beta_h\cdot (d\alpha+\gamma_h)),
\eeq
We have just shown in the proof of \eqref{eq:convju2} that 
$\rho_\ep^2(d\alpha+\gamma_h)\to d\alpha+\gamma$ in $L^p$ $\forall\, p<+\infty$ and that $d^*\beta_h\rightharpoonup d^*\beta^\eta$ weakly in $L^q$ $\forall\, q<3/2$. Thus, recalling %the mutual orthogonality of $d^*\beta^\eta$,  $d\alpha$ and $\gamma$ in $L^2(\Omega_\delta)$, as well as 
the estimate 
$\| d^*\beta^\eta- d^*\beta\|^2_2 \le C\eta$ from statement (iv) in Proposition \ref{prop:1}, we obtain
\beq\label{eq:convdalpha}
\lim_{\ep\to 0}\int_\Omega\rho_\ep^2(d\alpha+\gamma_h)\cdot d^*\beta_h
\ = \ 
\int_{\Omega}d^*\beta^\eta\cdot(d\alpha+\gamma)
\ = \ 
C\sqrt{\eta}+ 
\int_{\Omega}d^*\beta\cdot(d\alpha+\gamma)\, ,
\eeq
%\beq\label{eq:convdalpha}
%\left|\lim_{\ep\to 0}\int_\Omega\rho_\ep^2(d\alpha+\gamma_h)\cdot d^*\beta_h\right|=\left|\int_{\Omega_\delta\setminus\Omega} d^*\beta^\eta\cdot(d\alpha+\gamma)\right|\le C\eta+\int_{\Omega_\delta\setminus\Omega}|p_\delta|^2\, ,
%\eeq
\beq\label{eq:convdalpha0}
\lim_{\ep\to 0}\int_{\Omega}\rho^2_\ep|d\alpha+\gamma_h|^2\le\int_{\Omega_\delta}|d\alpha+\gamma|^2\, .
%=\int_{\Omega_\delta}|d\alpha|^2+|\gamma|^2.
\eeq

%%%%%%%%%%%%%%%%%%%%%%%%%%%%%%

For the remaining term, fix $0<\mu<1$ and set $r=c_0\mu\eta h^{1/2}$. Denote  $G_h^\lambda=\{ \text{dist}\,(x,\Gamma_h)\le\lambda\}\cap\Omega$. We have
\beq\label{eq:convbeta1}
2\pi^2\int_{\R^3}\rho^2_\ep|d^*\beta_h|^2=A_\ep+B_\ep+C_\ep\, ,
\eeq
where
\beq\label{eq:convbeta1bis}
A_\ep=2\pi^2\int_{G_h^\ep}\rho^2_\ep|d^*\beta_h|^2\, ,\ \  B_\ep=2\pi^2\int_{G_h^r\setminus G_h^\ep}|d^*\beta_h|^2\, ,\ \  C_\ep=2\pi^2\int_{\Omega\setminus G_h^r}|d^*\beta_h|^2\, .
\eeq
Let us estimate $A_\ep$. By \eqref{eq:pointwise3}, \eqref{eq:pointwise4}, and \eqref{eq:roep},
$
\rho_\ep^2|d^*\beta_h|^2 \le  \frac {h^2}{\ep^2} + \frac {2K^2}{\eta^4}
$  in $G^\ep_h$,
so \eqref{eq:tube} implies that
\beq\label{eq:convaep}
A_\ep\le |G_h^\ep|(\frac {h^2}{\ep^2} + \frac {2K^2}{\eta^4}) \le C(h + K \frac {\ep^2}{\eta^4 h})
%A_\ep\le h^2\frac{|G_h^\ep|}{\ep^2}+\frac{K^2|\Gamma_h|\ep^3}{\eta^4\ep^2}\le C(h+\frac{\ep}{h \eta^{4}})
%(1+C\eta)|| dp_\delta||_{L^1(\Lambda^2\Omega_\delta)}\, ,
\eeq
so that, since $h = g_\ep^{-1/2}$ and $\logeps^2 \le g_\ep \ll \ep^{-2}$,
we have
\beq\label{eq:convaepbis}
\limsup_{\ep\to 0} A_\ep=0\, .
\eeq
Let us turn to $C_\ep$. Let $\phi_r$ be the radial mollifier found in Lemma \ref{lem:convolbeta}.
Observe that $d^*\beta_h$ is harmonic on $\Omega\setminus G_h^r$, and hence coincides there with $\phi_r*d^*\beta_h$, by the mean-value property of harmonic functions. By \eqref{eq:convolbeta} and Rellich's Theorem we deduce that $\phi_r*d^*\beta_h$ is strongly compact in $L^2(\Omega)$, and hence by Proposition \ref{prop:1}, statement (iv) that  $\phi_r* d^*\beta_h\to d^*\beta^\eta$ in $L^2(\Omega)$
as $\ep\to 0$. We deduce that
\beq\label{eq:convcep}
\begin{aligned}
\limsup_{\ep\to 0}C_\ep \ = \ \limsup_{\ep\to 0}2\pi^2\int_{\Omega\setminus G^r_h}|\phi_r*d^*\beta_h|^2
&\le \lim_{\ep\to 0}2\pi^2\int_{\Omega}|\phi_r*d^*\beta_h|^2\\
&=2\pi^2\int_{\Omega}|d^*\beta^\eta|^2\\
&\le 2\pi^2\int_{\Omega}|d^*\beta|^2 + C\eta.
\end{aligned}
\eeq

\smallskip

To estimate $B_\ep$ we proceed as follows:  let $V_1=(G^r_h\setminus G^\ep_h)\setminus U_{r_0}$, where $U_{r_0}=\{\text{dist}\, (x, \cup_{i,j}\bd\tilde S_i\cup\bd P_{ij}) < r_0\}\cap\Omega$ and $r_0 = \frac{c_0}2 \eta h^{1/2}$, and set $V_2=(G^r_h\setminus G^\ep_h)\cap U_{r_0}$. For any $\sigma>0$ we have, using for $d^*\beta_h$ the bound \eqref{eq:pointwise3} on $V_1$ and \eqref{eq:pointwise4} on $V_2$,
\beq\label{eq:convbep1}
\begin{aligned}
2\pi^2\int_{V_1}|d^*\beta_h|^2&\le (1+\sigma)\frac{h^2}{2}\int_{V_1}\frac{dx}{|\text{dist}\,(x,\Gamma_h)|^2}\, +(1+\frac{1}{\sigma})\frac{2\pi^2K^2}{\eta^4}|V_1|\\
&\le (1+\sigma)h^2\pi\log\left(\frac{r}{\ep}\right)|\Gamma_h\setminus U_{r_0}|+(1+\frac{1}{\sigma})\frac{C\mu^2}{\eta^2}h|\Gamma_h\setminus U_{r_0}|\, ,
\end{aligned}
\eeq
\beq\label{eq:convbep2}
\begin{aligned}
2\pi^2\int_{V_2}|d^*\beta_h|^2 &\le 4(1+\sigma)\frac{h^2}{2}\int_{V_2}\frac{dx}{|\text{dist}\,(x,\Gamma_h)|^2}\, +(1+\frac{1}{\sigma})\frac{2\pi^2K^2}{\eta^4}|V_2|\, ,\\
&\le 4(1+\sigma)h^2\pi\log\left(\frac{r}{\ep}\right)|\Gamma_h\cap U_{r_0}|+(1+\frac{1}{\sigma})\frac{C\mu^2}{\eta^2}h|\Gamma_h\cap U_{r_0}|\, ,
\end{aligned}
\eeq
so that
\beq\label{eq:convbep3}
B_\ep\le (1+\sigma)h^2\pi\log\left(\frac{r}{\ep}\right)\left(|\Gamma_h|+3|\Gamma_h\cap U_{r_0}|\right)+(1+\frac{1}{\sigma})\frac{C\mu^2}{\eta^2}h|\Gamma_h|\, .
\eeq
If $g_\ep=h^{-2}=\logeps^2$ then
statements (iii), (v) of Proposition \ref{prop:1} and \eqref{eq:convbep3} give
\beq\label{eq:convbepcrit}
\limsup_{\ep\to 0} B_\ep\le \left[ (1+\sigma)\pi+(1+\frac{1}{\sigma})\frac{C\mu^2}{\eta^2}\right]\cdot (C\eta +||dp_\delta||_{L^1(\Omega_\delta)})\, ,
\eeq
while if $\logeps^2\ll g_\ep \ll \ep^{-2}$ (i.e. $\ep\ll h\ll \logeps^{-1}$), we have
\beq\label{eq:convbepnoncrit}
\limsup_{\ep\to 0} B_\ep\le (1+\frac{1}{\sigma})\frac{C\mu^2}{\eta^2}\cdot (C\eta + ||dp_\delta||_{L^1(\Omega_\delta)})\, .
\eeq
 We sum up all the contributions \eqref{eq;boundro}, \eqref{eq:convdalpha}, \eqref{eq:convdalpha0}, \eqref{eq:convaepbis},  \eqref{eq:convcep}, 
\eqref{eq:convbepcrit} and \eqref{eq:convbepnoncrit}, noting that the terms estimated  in 
\eqref{eq:convdalpha}, \eqref{eq:convdalpha0}, and \eqref{eq:convcep} add up to
$2\pi^2 \int_\Omega | d\alpha + \gamma + d^*\beta|^2 + C\sqrt{\eta} = 2\pi^2\int_\Omega |p_\delta|^2
+C\sqrt{\eta}$.
Thus, letting first $\mu\to 0$, then $\sigma\to 0$,  in \eqref{eq:convbepcrit} and \eqref{eq:convbepnoncrit}, we obtain
\beq\label{eq:glsupcrit}
\limsup_{\ep\to 0}\frac{E_\ep(u_\ep,\Omega)}{g_\ep}
\le 
\pi  \int_{\Omega_\delta}| dp_\delta|\, +2\pi^2\int_{\Omega}|p_\delta|^2
+ C \sqrt{\eta}
\eeq
if $g_\ep = \logeps^2$, and 
\beq\label{eq:glsupnoncrit}
\limsup_{\ep\to 0}\frac{E_\ep(u_\ep,\Omega)}{g_\ep}
\le 
2\pi^2\int_{\Omega_\delta}|p_\delta|^2
+ C \sqrt{ \eta}
\eeq
if $\logeps^2 \ll g_\ep \ll \ep^{-2}$. In these estimates  $C$ is independent of $\eta$.
Thus, since $p = 2\pi v$, and recalling \eqref{eq:parseval},  \eqref{ndc2}, \eqref{ndc3}, and
statement (iv) of Proposition \ref{prop:1},
we see that as first $\eta$ and then $\delta$ tend to $0$, the right-hand sides above converge to
$\frac 12 |dv|(\Omega) + \frac 12 \| v\|_{L^2(\Omega)}^2$ in the case $g_\ep = \logeps^2$, and
$ \frac 12 \| v\|_{L^2(\Omega)}^2$ in the case $\logeps^2\ll g_\ep \ll \ep^{-2}$. 
Thus, we can find sequences $\eta = \eta_\ep$ and
 $\delta = \delta_\ep$ tending to zero slowly enough that, if we define
$U_\ep := u_{\ep}$ with parameters $\delta_\ep$ in the piecewise linear approximation
(Lemma \ref{lem:fem}) and $\eta_\ep$ in the discretization of the vorticity (Proposition \ref{prop:1}) ,
then
\begin{align}
\label{eq:glsupcrit2}
\limsup_{\ep\to 0}\frac{E_\ep(U_\ep,\Omega)}{g_\ep}
& \le 
\frac 12 |dv|(\Omega) +  \frac 12 \| v\|_{L^2(\Omega)}^2
%=  |J|(\Omega) + \frac 12 \| v\|_{L^2(\Omega)}^2
\quad\quad
&\mbox{ if }
g_\ep = \logeps^2\\
\limsup_{\ep\to 0}\frac{E_\ep(U_\ep,\Omega)}{g_\ep}
& \le 
 \frac 12 \| v\|_{L^2(\Omega)}^2
%=  |J|(\Omega) + \frac 12 \| v\|_{L^2(\Omega)}^2
\quad\quad
&\mbox{ if }
 \logeps^2 \ll g_\ep \ll \ep^{-2}.
\label{eq:glsupnoncrit2}\end{align}
This finally proves the upper bound \eqref{eq:gammalimsup}, recalling that 
$J = \frac 12 dv$ for $g_\ep = \logeps^2$ and
$J= 0$ when $\logeps^2 \ll g_\ep \ll \ep^{-2}$.

Finally, having established the energy upper bound for $U_\ep$, the compactness assertions \eqref{eq:convju}, \eqref{eq:convu}, \eqref{eq:convJ} imply that $\frac 1{\sqrt{g_\ep}}jU_\ep, \frac 1{\sqrt{g_\ep}|U_\ep|}jU_\ep$ and 
$JU_\ep$ converge to limits in the required spaces, so it suffices only to idenfity the limits.
In fact,  it suffices to show for example
that $\frac 1{\sqrt{g_\ep}}jU_\ep \to v$ in the sense of distributions, and this follows (after taking $\eta_\ep$ in the definition of $U_\ep$  to converge to zero more slowly, if necessary) from \eqref{eq:convju2}.
\qed
%%%%%%%%%%%%%%%%%%%%%%%%%%%%%%%%%%%%%%%%%%%

\subsection{Construction of the sequence $u_\ep$ in case $g_\ep\ll \logeps^2$} 

Let  $J$ be an exact measure-valued 2-form in $\Omega$ and $v\in L^2(\Lambda^1\Omega)$ such that $dv=0$.
Fix $\delta>0$, and let  $p_\delta$ be the rational piecewise linear approximation of
$p := \frac v{2\pi}$ from Lemma \ref{lem:fem}. Furthermore, let $p_\delta'$
be the rational piecewise linear function from Lemma 2', so that $dp'$ approximates $J$.
%Notice that we may require w.l.o.g. $dp_\delta=0$ and $d^*p'_\delta=0$. 
Our Hodge decomposition gives
respectively $p_\delta=\gamma + d\alpha+d^*\beta'$, and $p'_\delta=\gamma'+d\alpha'+d^*\beta$. Let $h=\frac{1}{\sqrt{g_\ep}}$ and $h'=\frac{\logeps}{g_\ep}$, so that $h=h'\frac{\sqrt{g_\ep}}{\logeps}\ll h'$. Fix $\eta>0$, and for $h'<\eta^2$ let $d^*\beta_{h'}$ be the discretization of $d^*\beta$ via  Proposition \ref{prop:1}.
Let $\phi_{h'}$ be defined as in \eqref{eq:rmodz},
so that $d\phi_{h'} =\frac 1{ h'} d^*\beta_{h'}$, let $h^{-1}\gamma_h=d\psi_h$ be as in section \ref{sect:3.6},
and set $\alpha_h=h^{-1}\alpha$.
Finally, let $\rho_\ep$ be as in \eqref{eq:roep} and define
\beq\label{eq:vepsubcrit}
u_\ep=\rho_\ep  \exp(i2\pi\cdot(\phi_{h'}+\psi_h+\alpha_h))\, .
\eeq

%%%%%%%%%%%%%%%%%%%%%%%%%%%%%%%%%%%%%%%%%%%%%%

\subsection{Completion of proof of \eqref{eq:gammalimsup} in case $g_\ep\ll\logeps^2$}\label{sect:completion3}

%%%%%%%%%%%%%%%%%%%%%%%%%%%%%%%%%%%%%%
We have to estimate
\beq\label{eq:euep2}
\frac{E_\ep(u_\ep;\Omega)}{g_\ep}=\frac{h^2}{2}\int_{\Omega}|\nabla\rho_\ep|^2+\frac{W(\rho_\ep)}{\ep^2}+4\pi^2\rho_\ep^2\left|\frac 1{h'}d^*\beta_{h'}+ \frac 1h(\gamma_h+d\alpha)\right|^2\, .
\eeq
Then 
$ |\dist (x,\Gamma_h)\le\ep\}| \le \frac{\ep^2}{h'}$ as in \eqref{eq:tube}, so we find
as in \eqref{eq;boundro} that
\beq\label{eq:boundro2}
\frac{h^2}{2}\int_{\Omega}|\nabla\rho_\ep|^2+\frac{W(\rho_\ep)}{\ep^2}\le C\frac{h^2}{h'}
\longrightarrow 0
\eeq
For the remaining terms we have %, following our earlier computations, 
\beq\label{eq:boundjvh2}
2\pi^2 \int_{\Omega}\rho_\ep|d\alpha +\gamma_h|^2\to 2\pi^2\int_{\Omega}|d\alpha+\gamma|^2\le 2\pi^2\int_{\Omega_\delta}|p_\delta|^2\, ,
\eeq
\beq\label{eq:mixedterm}
2\pi^2\frac{h}{h'}\int_{\Omega}\rho_\ep^2d^*\beta_{h'}\cdot (d\alpha+\gamma_h)\to 0\, ,
\eeq
\beq\label{eq:convbeta2}
2\pi^2\frac{h^2}{{h'}^2}\int_{\Omega}\rho_\ep^2|d^*\beta_{h'}|^2=A'_\ep+B'_\ep + C'_\ep\, ,
\eeq
where, in the notation corresponding to \eqref{eq:convbeta2},
\beq\label{eq:aepbep}
\begin{aligned}
A'_\ep&=2\pi^2\frac{h^2}{{h'}^2}\int_{G^\ep_{h'}}\rho_\ep^2|d^*\beta_{h'}|^2\, , \ \  \\
B'_\ep&=2\pi^2\frac{h^2}{{h'}^2}\int_{G_{h'}^r \setminus G_{h'}^\ep}|d^*\beta_{h'}|^2\, , \\
C'_\ep&=2\pi^2\frac{h^2}{{h'}^2}\int_{\Omega\setminus G_{h'}^r}|d^*\beta_{h'}|^2\, \,
\end{aligned}
\eeq
for $r = c_0 \eta (h')^{1/2}$.
Reasoning as in \eqref{eq:convaep} and \eqref{eq:convcep} we deduce {\it a fortiori} that $\limsup A_\ep=\limsup_{\ep\to 0}C_\ep = 0$, while  following 
\eqref{eq:convbep1} and \eqref{eq:convbep2} we deduce
\beq\label{eq:convbep4}
B'_\ep\le (1+\sigma) h^2\pi \log( \frac r \ep) \left( |\Gamma_{h'}|+ C |\Gamma_{h'} \cap U_r| \right)+(1+\frac{1}{\sigma})\frac{h^2}{{h'}}|\Gamma_{h'}|\, ,
%\le (1+\sigma)\pi h'|\Gamma_{h'}|+o(1)\, ,
\eeq
so that $\limsup B'_\ep\le (1+\sigma)\pi \int_{\Omega_\delta}|dp_\delta'| + C\eta$ by Proposition \ref{prop:1} (iii). Summing up the various contributions and then letting  $\sigma\to 0$, we obtain
\beq\label{eq:glsupsubcrit}
\limsup_{\ep\to 0}\frac{E_\ep(u_\ep)}{g_\ep}\le \pi\int_{\Omega_\delta}|dp'_\delta|\, +2\pi^2\int_{\Omega_\delta}|p_\delta|^2\,  + C \eta.
\eeq
We conclude the proof as in the previous cases,
by defining 
% $\delta = \delta_\ep$ tending to zero slowly enough that, if we define
$U_\ep := u_{(\ep, \eta_\ep, \delta_\ep)}$ (that is, defining $u_\ep$ as above, but 
with parameters $\delta_\ep$ in the piecewise linear approximation
of Lemma \ref{lem:fem}, and $\eta_\ep$ in the discretization of the vorticity of Proposition \ref{prop:1})
for $\eta_\ep$ and $\delta_\ep$ converging to zero
sufficiently slowly, so that $U_\ep$ satisfies 
the Gamma-limsup inequality \eqref{eq:gammalimsup}, and
then verifying the convergence as before.
\qed

%%%%%%%%%%%%%%%%%%%%%%%%%%%%%%%%%%%%%%%%%%%%%%

%\subsection{Upper bound using intrinsic Hodge decompositions}\label{sect:hodgeintrinsic}

%%%%%%%%%%%%%%%%%%%%%%%%%%%%%%%%%%%%%%%%%%%%%%%%

%%%%%%%%%%%%%%%%%%%%%%%%%%%%%%%%%%%%%%%%%%%%%%%%%%%

\section{APPLICATIONS TO SUPERCONDUCTIVITY}\label{sect:supercond}

%%%%%%%%%%%%%%%%%%%%%%%%%%%%%%%%%%%%%%%%%%%%%%%

In this section we prove Theorem \ref{thm:3} and begin the analysis of the limiting functional $\mathcal F$, deriving the curvature equation for the vortex filaments. We use a good deal of notation that was introduced in Section \ref{S:superc}.

In the companion paper \cite{bjos} we analyze in more detail the properties of $\mathcal F$ and derive further applications such as a general expression for the first critical field $H_{c_1}$.

\subsection{Proof of Theorem \ref{thm:3}} 
First, recalling that $h_{ex} = dA_{ex, \ep}$, we see immediately from the definition of $\calF_\ep$ and
of the $\dot H^1_*(\Lambda^1\R^3)$ norm  that
\[
\| A_\ep - A_{ex, \ep}\|_{\dot H^1_*}^2 \le 2 \calF_\ep(u_\ep, A_\ep) \le K \logeps^2.
\]
It immediately follows that $\frac 1 \logeps (A_\ep - A_{ex, \ep})$ is weakly precompact in $\dot H^1_*(\Lambda^1\R^3)$,
and since $\logeps^{-1} A_{ex, \ep} \to A_{ex,0}$ in $\dot H^1_*(\Lambda^1\R^3)$, we deduce \eqref{eq:convA}.

The above bounds on $A_\ep$ and the Sobolev embedding $\dot H^1_* \hookrightarrow L^6$ implies that
\beq
\|  \logeps^{-1} A_\ep \|_{L^6(\Lambda^1\Omega)} \le K\, .
\label{eq:AL6}\eeq

In order to establish the remaining
compactness assertions, we use the decomposition \eqref{eq:br}, which
implies that
\[
E_\ep(u_\ep) \le  \calF_\ep(u_\ep, A_\ep) + |\int_\Omega A_\ep \cdot ju_\ep| \le K \logeps^2 + |\int_\Omega A_\ep \cdot ju_\ep| \, ,
\]
using the fact that $\calM(A; dA_{ex, \ep}) + \calR(u_\ep, A_\ep) \ge 0$. To estimate the right-hand
side, note that in general
\begin{align*}
|ju \cdot A| 
&\le 
|u| \, |Du| \, |A| \le \frac 14 |Du|^2 + |u|^2|A|^2 
\le \frac14 |Du|^2 + 2 |A|^2 + 2(|u|-1)^2|A|^2\\
&  
\le \frac14 |Du|^2 + 2 |A|^2 +  \frac c {\ep^2} \left| \, |u|-1\,\right|^3 + C \ep^2 |A|^6.
\end{align*}
And hypothesis $(H_q)$ with $q\ge 3$ implies that $c\left| \, |u|-1\,\right|^3 \le \frac 12 W(u)$ if $c$ is small enough, so that
\[
|\int_\Omega A_\ep \cdot ju_\ep| \le \frac 12 E_\ep(u_\ep) + C \int_\Omega |A_\ep|^2 +  \ep^2|A_\ep|^6 \ dx.
\]
By combining the above inequalities and using \eqref{eq:AL6}, we find that $E_\ep(u_\ep) \le K' \logeps^2$, which in view of
Theorem \ref{thm:2} implies that 
 \eqref{eq:convju}, \eqref{eq:convu}, \eqref{eq:convJ} hold with $g_\ep = \logeps$.

 To prove statement (ii), consider the decomposition of $\mathcal F_\ep$ given by \eqref{eq:br}, \eqref{eq:br1}, which may be rewritten 
\beq\label{eq:br3}
\frac{\mathcal F_\ep (u_\ep,A_\ep)}{\logeps^2}=\frac{E_\ep(u_\ep)}{\logeps^2}+\mathcal M(\frac{A_\ep}{\logeps},\frac{h_{ex}}{\logeps})+\mathcal I(\frac{ju_\ep}{\logeps},\frac{A_\ep}{\logeps})+\frac{\mathcal R(u_\ep,A_\ep)}{\logeps^2}\, .
\eeq
 Recall that \eqref{eq:gammalimE} asserts
 $$
 \frac{1}{\logeps^2}E_\ep(u_\ep)\xrightarrow{\Gamma} E(v)\, ,
 $$
 with $E(v)$ defined in \eqref{eq:ev}. Note further that $\mathcal M$ is lower semicontinuous with respect to the weak $\dot H^1_*$ convergence of $\tfrac{A_\ep}{\logeps}$, and hence, taking into account \eqref{eq:convA}, we readily deduce
\beq\label{eq:convM}
\mathcal M(\frac{A_\ep}{\logeps},\frac{h_{ex}}{\logeps})\xrightarrow{\Gamma}\mathcal M (A,h)\, .
\eeq
Moreover, by Sobolev embedding, \eqref{eq:convA} implies $\frac{A_\ep}{\logeps}\to A$  strongly in $L^p(\Omega)$, for any $1\le p<6$, whereas  \eqref{eq:convu} gives $\frac{ju_\ep}{\logeps}\rightharpoonup v$ weakly in $L^{2q/(q+2)}(\Omega)$. For $q\ge 3$ we have $2q/(q+2)\ge  6/5$, so that for any admissible sequence $(u_\ep,A_\ep)$ we have 
\beq\label{eq:convI}
\mathcal I(\frac{ju_\ep}{\logeps},\frac{A_\ep}{\logeps})\to \mathcal I (v,A)\, .
\eeq
Note finally that for the remainder term $\mathcal R(u_\ep,A_\ep)$, since $| 1- |u|^2|^{3/2} \le C W(u)$, 
$$
\begin{aligned}
\left|\mathcal R(u_\ep, A_\ep)\right|
&\le\int_\Omega \left|1-|u|^2\right| \, |A_\ep|^2dx\\
&\le C \ep^{4/3} \left(\int_\Omega \frac{W(u_\ep)}{\ep^2}dx\right)^{2/3}\left(\int_\Omega|A_\ep|^6dx\right)^{1/3}\\
&\le C\ep^{4/3} E_\ep(u_\ep)^{2/3}\| A_\ep\|_{L^6(\Omega)}^2  \\
&\le C\ep^{4/3}\logeps^{10/3} ,
\end{aligned}
$$
so that $\frac{1}{\logeps^2}\mathcal R(u_\ep,A_\ep)\le C(\ep \logeps)^{4/3}$ converges uniformly to $0$.

From the above considerations it follows immediately that
\beq\label{eq:gammaF}
\frac{\mathcal F_\ep(u_\ep,A_\ep)}{\logeps^2}\xrightarrow{\Gamma} E(v)+\mathcal I(v,A)+\mathcal M(A,h)\, ,
\eeq
which is formula \eqref{eq:F}.

\qed

\subsection{Some properties of the $\Gamma$-limit $\mathcal F$}\label{sect:curvature} In this section we derive the Euler-Lagrange equations for the functional $\mathcal F$ and deduce a curvature equation for the limiting vortex filaments. First of all notice that $\mathcal F$ is strictly convex and hence admits a unique minimizer $(v,A)$. We first make variations of $\mathcal F$ with respect to $A$. Standard computations yield
%\beq\label{eq:E-L-F1}
%\begin{cases}
%d^*(dA-h)=\mathbf{1}_\Omega (v-A)  &\text{in } \Omega \\
%d^*(dA-h)=0 &\text{in } \R^3\setminus\bar\Omega\\
%[(\star dA-h)_T=0 &\text{on }\bd\Omega\, ,
%\end{cases}
%\eeq
\beq\label{eq:E-L-F1}
\begin{cases}
d^*(dA-h)=\mathbf{1}_\Omega \cdot(v-A) &\text{in }\R^3\\
[(\star(dA-h))_\top]=[(dA-h)_N]=0 &\text{on }\bd\Omega\, ,
\end{cases}
\eeq 
where $\mathbf{1}_\Omega$ denotes the characteristic function of $\Omega$ and $[(dA-h)_N]$ denotes the jump across $\bd \Omega$ of the normal component of $(dA-h)$. Denoting  $j=\mathbf{1}_\Omega\cdot(v-A)$ the gauge-invariant supercurrent in $\Omega$ and $H=dA-h$, we recover from \eqref{eq:E-L-F1} Amp\`ere law  $d^*H=j$ in $\R^3$ for the magnetic field $H$, which has to be coupled with Gauss law for electromagnetism $dH=d(dA-h)=0$ in $\R^3$, and with the continuity condition $[H]=0$ on $\bd\Omega$, which is a consequence of $[H_N]=0$ (by \eqref{eq:E-L-F1}) and $[H_\top]=0$ on $\bd\Omega$ (by Gauss law $dH=0$).

\medskip
Let now $J(v)$ denote the convex and positively 1-homogeneous function $J(v):=||dv||$, and let $\bd J$ be its subdifferential. Making variations of $\mathcal F$ with respect to $v$ yields the differential inclusion
\beq\label{eq:E-L-F2}
0\in \frac{1}{2}\, \bd J(v) +v-A \, .
\eeq
Assume the minimizer $v$ is regular and spt$\, |dv|=\bar U$, with $U$ an open subset of $\Omega$. In particular, if $U$ is a proper subset of $\Omega$, then one may view $\Omega\cap\bd U$ as a kind of  free boundary. This situation has a counterpart in the 2-d case (see \cite{ss}, \cite{js2}).
Then \eqref{eq:E-L-F2} corresponds to
\beq\label{eq:variationv}
\frac{1}{2}\int_U\frac{dv}{|dv|}\wedge\star d\phi+\int_\Omega(v-A)\wedge\star\phi=0
\eeq
for any $\phi\in C^\infty(\Lambda^1\Omega)$ such that spt$\, \phi\subset \Omega\setminus\bd U$.
Testing \eqref{eq:variationv} with $\phi\in C^\infty_c(\Lambda^1(\Omega\setminus\bar U))$ we deduce $v=A$ in $\Omega\setminus \bar U$.
Testing now with those $\phi\in C^\infty(\Lambda^1(\Omega))$ such that spt$\, \phi\subset \bar U\setminus(\Omega\cap\bd U)$ and integrating by parts \eqref{eq:variationv} we further deduce
\beq\label{eq:variationvbis}
\int_{U}\left[\frac{1}{2}d^*\left(\frac{dv}{|dv|}\right)+v-A\right]\wedge\star\phi\, +\int_{\bd\Omega\cap \bar U}(\phi\wedge\star\frac{dv}{|dv|})_\top=0\, ,
\eeq 
whence
\beq\label{eq:curvature0}
\begin{cases}
d^*\left(\frac{dv}{|dv|}\right)=2(A-v) &\text{in } U\, ,\\
(\star\frac{dv}{|dv|})_\top=0 &\text{on }\bar U\cap\bd\Omega\, .
\end{cases}
\eeq
Notice that $\tau=\star\frac{dv}{|dv|}$ is the unit tangent covector field to the streamlines of the covector distribution $\star dv$, which correspond to the limiting vorticity. From \eqref{eq:curvature0} we obtain in particular
\beq\label{eq:curvature1}
\begin{cases}
\tau\wedge \star d\tau=2\tau\wedge(v-A)=2\tau\wedge j &\text{in }U\, ,\\
\tau_\top=0 &\text{on }\bar U\cap\bd\Omega\, .
\end{cases}
\eeq
Denoting respectively by $\vec\tau$ and $\vec\jmath$ the vector fields correpsonding to $\tau$ and $j$, we notice that $\star(\tau\wedge j)$ corresponds to $\vec\tau\times\vec\jmath$,  and  $\star d\tau$ corresponds to the vector field $\nabla\times\vec\tau$, so that $\star(\tau\wedge\star d\tau)$ corresponds to the curvature vector $\vec\kappa=\vec\tau\times (\nabla\times\vec\tau)$. We thus deduce the curvature equation \eqref{eq:curvature}.

\begin{remark} Notice that $d^*\tau=\star d(\frac{dv}{|dv|})=0$ (or equivalently $\nabla\cdot\vec\tau=0$) in $\Omega$. From \eqref{eq:curvature0} we deduce that $\tau$ satisfies the Hodge system
\beq\label{eq:tau}
\begin{cases}
d\tau=\star 2j &\text{in }\Omega\\
d^*\tau=0 &\text{in }\Omega\\
\tau_\top=0 &\text{on }\bd\Omega,
\end{cases}
\eeq
or respectively
\beq\label{eq:vectau}
\begin{cases}
\nabla\times\vec\tau=2\vec\jmath &\text{in }\Omega\\
\nabla\cdot\vec\tau=0 &\text{in }\Omega\\
\vec\tau_\top=0 &\text{on }\bd\Omega,
\end{cases}
\eeq
under the pointwise constraint $|\tau|=1$ (resp. $|\vec\tau|=1$) in spt$\, j$.
\end{remark}
\begin{remark} From \eqref{eq:E-L-F1}, \eqref{eq:curvature0} we recover in particular the continuity equation $d^*j=d^*(v-A)=0$ (or equivalently, $\nabla\cdot\vec\jmath=0$). If $A$ is in the Coulomb gauge $d^*A=0$ (which happens in particular if $A_{ex} = c x^1dx^2 - x^2 dx_1)$ and $A\in H^1_*$, so that $d^*(A-A_{ex})=0$), then it follows that $v$ satisfies 
\beq\label{eq:dstarv}
\begin{cases}
d^*v=0 &\text{in }\Omega \\
 v_N=0 & \text{on }\bd\Omega.
 \end{cases}
 \eeq
\end{remark}
%\beq\label{eq:curvature}
%\begin{cases}
%\vec\kappa=2\vec\tau\times\vec\jmath &\text{in }\Omega,\\
%\vec\tau_\top=0 &\text{on }\bd\Omega\, .
%\end{cases}
%\eeq

%%%%%%%%%%%%%%%%%%%%%%%%%%%%%%%%%%%%%%%%%%%%%%%%%%%%%
%%%%%%%%%%%%%%%%%%%%%%%%%%%%%%%%%%%%%%%%%%%%%%%%%%%%%%%%%%%%%%%%%

%\section{BOSE-EINSTEIN CONDENSATES}

%%%%%%%%%%%%%%%%%%%%%%%%%%%%%%%%%%%%%%%%%%%%%%%%%%%%
%%%%%%%%%%%%%%%%%%%%%%%%%%%%%%%%%%%%%%%%%%%%%%%%%%%%%

\section{APPENDIX}\label{sect:appendix}

In this Appendix we recollect basic facts and notation that we use throughout the paper, as well as background on differential forms, Hodge decompositions, minimal connections. We also provide the proofs of Lemma \ref{lem:nurest} and Lemma \ref{lem:mincon}.

\subsection{Differential forms} For $0\le k\le n$, let $\Lambda^k\R^n$ be the space of $k$-covectors in $\R^n$, i.e. $\theta\in\Lambda^k\R^n$ if $\theta=\sum\theta_I dx^I$, where $dx^I:=dx^{i_1}\wedge ... \wedge dx^{i_k}$, $1\le i_1<...<i_k\le n$. For $\theta,\beta\in\Lambda^k\R^n$, their inner product is given by $(\theta,\beta):=\sum \theta_I\cdot\beta_I$.
 
 \smallskip
 
 Let $\Omega\subset\R^n$ be a smooth bounded open set. We will denote  by $C^{\infty}(\Lambda^k\Omega):=C^\infty(\Omega;\Lambda^k\R^n)$ the space of smooth $k$-forms on $\Omega$. Similarly we denote by $L^p(\Lambda^k\Omega)$, $W^{1,p}(\Lambda^k\Omega)$ the spaces of $k$-forms of class $L^p$ and $W^{1,p}$ respectively. For $\omega\in C^\infty(\Lambda^k\Omega)$, denote by $\omega_\top\in C^\infty (\Lambda^k\partial\Omega)$ its tangential component\footnote{i.e. $\omega_\top:=i^*\omega$, where $i:\partial\Omega\to\Omega$ is the inclusion map} on $\bd\Omega$, and by $\omega_N:=\omega_{|\partial\Omega}-\omega_\top$ its normal component on $\bd\Omega$. The operators $\omega\mapsto\omega_\top$ and $\omega\mapsto\omega_N$ extend to bounded linear operators $W^{1,p}(\Lambda^k\Omega)\to L^p(\partial\Omega; \Lambda^k\R^n)$.
 The Hodge star operator $\star:\Lambda^k\R^n\to\Lambda^{n-k}\R^n$ is defined in such a way that $\theta\wedge\star\varphi=(\theta\, ,\varphi)dx^1\wedge...\wedge dx^n$. The $L^2$ inner product of $\omega\, ,\eta\in C^{\infty}(\Lambda^k\Omega)$ is defined by 
$$\langle \omega\, ,\eta\rangle:=\int_\Omega(\omega\, ,\eta) d\mathcal L^n=\int_\Omega\omega\wedge\star\eta.$$
%When considering restrictions to (piecewise smooth) $m$-dimensional submanifolds with boundary $T\subset\Omega$ we will denote by $\omega\rest T$ the tangential component of $\omega_{|T}$ (so that in particular, $\omega_\top=\omega\rest\bd \Omega$), and by
%$$\langle \omega\rest T\, , \phi\rangle=\int_T(\omega\, ,\varphi)_Td\mathcal H^m\, ,$$
%where $(\cdot\, ,\cdot)_T$ denotes the restriction of the euclidean product $(\cdot\, ,\cdot)$ to the tangent space of $T$. 

 Let $T\subset\Omega$ be a piecewise smooth $m$-dimensional submanifold with boundary. Integration of (the tangential component of) a smooth $m$-form $\omega$ on $T$  will be denoted by 
 $\int_T\omega\equiv\int_T\omega_\top = \int_T i^*\omega,$ with $i:T\to\Omega$ the inclusion map. 
 %In the previous notations,
 %we may write $\int_T\omega=\langle\varphi\rest T\, , \star 1\rangle$, where $\star$ denotes the star operator on the tangent space of $T$.

The adjoint with respect to $\langle\cdot\, , \cdot\rangle$ of the $\star$ operator on $k$-forms is $(-1)^{k(n-k)}\star$.
%of the  wedge product $\wedge$ is the inner product, also denoted by $\rest$, i.e. 
%$\langle \omega\rest\xi\, , \varphi\rangle=\langle \xi\, ,\omega\wedge\varphi\rangle$, while the adjoint of $\star$ on $k$-forms is $(-1)^{k(n-k)}\star$.

\subsubsection{measure-valued forms}\label{sect:mvf}
A distribution-valued $k$-form $\mu$ is an element of the dual space\footnote{One can thus identify a distribution-valued $k$-form with a $k$-current, see \cite{F}, although we generally choose not to do so.} of $C^\infty(\Lambda^k\Omega)$, and we express the duality pairing through the notation $\langle\cdot\, ,\cdot\rangle$. In particular, we will say  that $\mu$ is a measure-valued $k$-form  (cf. \cite{bo}, Definition 2.1) if
 \beq\label{eq:measureform}
 \langle\mu\, ,\varphi\rangle\le C|| \varphi ||_\infty\qquad\forall\, \varphi\in C^\infty_c(\Lambda^k\Omega)\, .
 \eeq
 
 A measure-valued $k$-form $\mu$ can be represented by integration (cf. \cite{bo}, Proposition 2.2) as follows:
 \beq\label{eq:integration}
 \langle\mu\, ,\varphi\rangle=\int_\Omega (\nu,\varphi) \, d|\mu|\, ,
 \eeq
 where $|\mu|$ is the total variation measure of (the vector measure) $\mu$ and $\nu$ is a $|\mu|$-measurable $k$-form such that $(\nu,\nu)^{1/2}=: |\nu|=1$ $|\mu|$-a.e. in $\Omega$. We denote by $||\mu||:=|\mu|(\Omega)$ the total variation norm of $|\mu|$. It coincides with the $L^1$ norm $||\mu||_1=\int_\Omega |\mu|$ if $\mu\in L^1(\Lambda^k\Omega)$. 
We denote by $\mu\rest U$ the restriction of $\mu$ to $U\subset\Omega$, defined by
\beq%\label{eq:integration}
 \langle\mu\rest U\, ,\varphi\rangle=\int_U (\nu,\varphi) \, d|\mu|\, .
 \eeq
 Moreover, for $\eta$ a unit $k$-covector and $\mu$ a measure $k$-form in $\Omega$, the component along $\eta$ of $\mu$ is a signed measure denoted $(\mu,\eta)$ defined by
 \beq\label{eq:mueta}
 (\mu,\eta)(U):=(\mu(U),\eta)=\int_U(\nu,\eta)d|\mu|\qquad\forall U\Subset\Omega,
 \eeq
 with variation measure $|(\mu,\eta)|$ given by
 \beq\label{eq:varmueta}
 |(\mu,\eta)|(U)=\int_U|(\nu,\eta)|d|\mu| \qquad\forall U\Subset\Omega.
 \eeq
 \medskip

Notice that an oriented piecewise smooth $k$-dimensional submanifold $T\subset\Omega$ can be identified with a measure $k$-form $\widehat T$, whose action on smooth $k$-forms $\varphi$ is given by
 \beq\label{eq:hatt}
 \langle \widehat {T}\, , \varphi\rangle = \int_T \varphi\, .%\langle\varphi\rest T\, , \star 1\rangle\, 
 \eeq
 
 Let $d$ be the exterior differentiation operator, and $d^*=(-1)^{n(k+1)+1}\star d\star$ its adjoint with respect to $\langle\cdot\, , \cdot\rangle$, i.e. $\langle d\omega\, ,\eta\rangle=\langle\omega\, , d^*\eta\rangle$ for $\omega$ a $k$-form, and $\eta$ an $(n-k-1)$-form. We define the action of $d$ and $d^*$ on a measure-valued distribution $\mu$ by duality, so that $\langle d\mu, \eta\rangle := \langle \mu, d^*\eta\rangle$ and $\langle d^*\mu, \eta\rangle := \langle \mu, d\eta\rangle$ for $\eta$ with compact support.

Stokes' Theorem reads $\int_T d\varphi=\int_{\bd T}\varphi_\top$, for $\varphi$ a smooth $(k-1)$-form and $T$ as above.  Notice that by \eqref{eq:hatt}
 we have
\beq\label{eq:bddstar}
\langle \widehat {T}\, , d\varphi\rangle=\langle d^*\widehat {T}\, ,\varphi\rangle=\langle \widehat{\bd T}\, , \varphi\rangle\, ,
\quad\text{ so that  }\ \widehat{\bd T}=d^*{\widehat {T}}\, .
\eeq

A measure-valued $k$-form $\mu$ is said to be 
{\em closed} if $d\mu=0$, and it is {\em exact} if there exists a measure-valued  $k-1$-form $\psi$ such that 
$\mu = d\psi$.

 \smallskip
 
\subsubsection{the tangential part of measure-valued forms}\label{sect:mvftop}
Suppose that $\omega$ is a closed measure-valued $n-1$-form defined on an open subset $\Omega\subset \R^n$. If we fix an open $U\subset \Omega$ with piecewise smooth boundary $\bd U$, we will use the notation
$\omega_\top$ to denote the distribution defined by
\beq
 \int f \omega_\top  \ := \ \int_U df \wedge \omega \quad\mbox{ for all }
f\in C^\infty (U)\cap C(\bar U).
\label{qtop.def1}\eeq
Thus our definition states that $\omega_\top :=\star d(\chi_U \omega)$ in the sense of distributions, where $\chi_U$ is the characteristic
function of $U$. Although the notation $\omega_\top$ does not explicitly indicate the set $U$,
it will normally be clear from the context, and when it is  not, we will write for example ``$\omega_\top$ on $\bd U$".

In general $\omega_\top$ is a distribution supported on $\bd U$. We claim that 
\beq
\mbox{$\int f \omega_T$ depends only on $f|_{\partial U}$, for smooth $f$.
}
\label{omegatopclaim}\eeq To verify this, it suffices to check
that 
$\int_U df \wedge \omega = 0$ for $\omega$ as above, whenever $f=0$ on $\partial U$. Toward this end, let $\chi_\ep$ denote a smooth function with compact support in $U$, such that $0\le \chi_\ep\le 1$, $|\nabla \chi_\ep|\le C/\ep$,
$\chi_\ep (x) = 1$ if  $\dist(x,\partial U) \ge \ep$, and $\chi_\ep = 0$ if $\dist(x,\partial U)\le \ep/2$. Then
\[
\int_U df \wedge \omega \ = \  \lim_{\ep\to 0}\int_U \chi_\ep df \wedge \omega
\ = \ \lim_{\ep\to 0}\int_U f d\chi_\ep\wedge \omega 
\]
since $\omega$ is closed. Since $f$ is smooth and $f=0$ on $\partial U$, $|f d\chi_\ep| \le (C\ep)(C/\ep) \le C$
when $\dist(x,\partial U)< \ep$, 
so the right-hand side is bounded by $|\omega|( \mbox{supp} \, d\chi_\ep)$. Since $|\omega|$ has finite total mass by assumption, we easily conclude that there exists a sequence $\ep_k\searrow 0$ such that 
$\lim_{k\to \infty}\int_U \chi_{\ep_k} df \wedge \omega=0$, proving \eqref{omegatopclaim}.

It follows from \eqref{omegatopclaim} that expressions such as $\int_{\partial U}\omega_\top$ are well-defined.

In this paper it will often be the case that $\omega_\top$ is
a measure supported on $\partial U$, and when this holds,
we may also think of $\omega_\top$ as a measure-valued $(n-1)$-form on $\partial U$. In particular,
if $\omega$ is smooth enough, then $\int f \omega_\top$ agrees with the
classical expression discussed above, $\int_{\bd U} f(x)i^* \omega(x)$, 
where  $i:\partial U\to \Omega$ is the inclusion map.
 
 \smallskip
 
\subsubsection{harmonic forms}
 If $d\omega = d^*\omega = 0$ then $\omega$ is said to be harmonic. Denote by 
 $$\mathcal H^k\equiv\mathcal H^k(\Omega):=\{ \omega\in L^2\cap C^\infty(\Lambda^k\Omega)\, ,\ \ d\omega=0\, ,\ d^*\omega=0\}$$ 
 the space of harmonic $k$-forms on $\Omega$, and by
 $$
 \mathcal H^k_\top=\{ \omega\in\mathcal H^k\, ,\ \omega_\top=0\}\, ,\qquad \mathcal H^k_N=\{ \omega\in\mathcal H^k\, ,\ \omega_N=0\},
 $$
 the spaces of harmonic forms with vanishing tangential and normal components on $\bd\Omega$. Since $\star \omega_N=(\star\omega)_\top$ and $\star\star=(-1)^{k(n-k)}$, we have the bijections
 $$
 \star:\mathcal H^k_\top\to\mathcal H^{n-k}_N\, ,\qquad \star:\mathcal H^k_N\to\mathcal H^{n-k}_\top\, .
 $$ 
 Harmonic forms in $\mathcal H^k_\top\cup H^k_N$ are smooth up to $\partial\Omega$. Denote by $H(\omega)$ (resp. $H_\top(\omega)$, $H_N(\omega)$) the orthogonal projection of a $k$-form $\omega$ on $\mathcal H^k$ (resp. $\mathcal H^k_\top$, $\mathcal H^k_N$). With respect to an orthonormal basis $\{\gamma_i\}_{i=1,\dots,\ell}$ of $\mathcal H^k$ (resp. $\mathcal H^k_\top$, $\mathcal H^k_N$), the orthogonal projection is of course given by  $\sum_{i=1}^\ell \langle \omega\, ,\gamma_i\rangle\, \gamma_i$.
 
 The Laplace operator $-\Delta=dd^*+d^*d$ on smooth $k$-forms is  positive semidefinite, commutes with $\star$, $d$, $d^*$, and $h\in\mathcal H^k\Rightarrow -\Delta h=0$.
 
 %%%%%%%%%%%%%%%%%%%%%%%%%%%%%%%%%%%%%%%%%%%

\subsection{Hodge decompositions}\label{sect:hodge} For $\omega\in L^p(\Lambda^k\Omega)$, $1<p<+\infty$, we have the following Hodge decomposition, orthogonal with respect to $\langle\cdot\, ,\cdot\rangle$ (see e.g. \cite{iss}, Theorem 5.7, or \cite{M} for $p\ge 2$):
 \beq\label{eq:hodge}
 \omega=\gamma+d\alpha+d^*\beta\, ,
 \eeq
 where
 \beq\label{eq:hodgen}
  \gamma\in\mathcal H^k_N,\ \alpha\in W^{1,p}(\Lambda^{k-1}\Omega),\ \beta\in W^{1,p}(\Lambda^{k+1}\Omega),\ \beta_N=0.
 \eeq
Then $\gamma = H_N(\omega)$.  Moreover there exists a unique $\Psi\in W^{2,p}(\Lambda^k\Omega)$ such that
 \beq\label{eq:potential}
 -\Delta\Psi=\omega-H_N(\omega)\, ,\quad \Psi_N=0,\quad (d\Psi)_N=0\, ,
 \eeq
 and
 \beq\label{eq:ellipticestimates}
 || d\Psi ||_{1,p}+ || d^*\Psi ||_{1,p}\le C_p || \omega ||_p\, .
 \eeq
 We will write
 $%\label{eq:deltainvers}
 \Psi=-\Delta_N^{-1}(\omega-H_N(\omega))\, .
 $
 
 \medskip
 We may also decompose $\omega=\gamma+d\alpha+d^*\beta$ with
 \beq\label{eq:hodget}
 \gamma\in\mathcal H^k_\top,\ \alpha\in W^{1,p}(\Lambda^{k-1}\Omega),\ \beta\in W^{1,p}(\Lambda^{k+1}\Omega),\ \alpha_\top=0,
 \eeq
so that $\gamma=H_\top(\omega)$. In this case there exists a unique $\Psi\in W^{2,p}(\Lambda^k\Omega)$ such that
 \beq\label{eq:potentialt}
 -\Delta\Psi=\omega-H_\top(\omega)\, , \quad \Psi_\top=0,\quad (d^*\Psi)_\top=0\, .
 \eeq
Moreover, \eqref{eq:ellipticestimates} holds. We write in this case $\Psi=-\Delta^{-1}_\top(\omega-H_\top(\omega))$.
 
 The operator $-\Delta^{-1}_\top$ is self-adjoint on $\calH_\top^\perp$, and similarly
 $-\Delta^{-1}_N$ is self-adjoint on $\calH_N^\perp$.
  
 \begin{remark}\label{rem:hodgern}{\rm In case $\Omega=\R^n$, basic properties of harmonic functions imply that $\mathcal H^k = \{0\}$. For $\omega$ compactly supported  the potential $\Psi$ is given in particular by $\Psi=G*\omega$, where $G(x)=c_n|x|^{n-2}$ is the Poisson kernel on $\R^n$, $n\ge 3$. The Hodge decomposition of $\omega$ reads $\omega=d\alpha+d^*\beta$ with $\beta=G*d\omega$ and $\alpha=G*d^*\omega$. In this case $\alpha, \beta\in \dot W^{1,p}$ rather than $W^{1,p}$.
 }
 \end{remark}
 \medskip
 
 For $\omega\in L^1(\Lambda^k\Omega)$ or more generally a measure-valued $k$-form, the decomposition \eqref{eq:hodge} fails in general, but decompositions of the form \eqref{eq:potential}, \eqref{eq:potentialt} still hold, in view of this variant of \cite{bo}, Theorem 2.10:
  \begin{proposition}\label{prop:2.10} Let $\mu$ be a measure-valued $k$-form in $\Omega$. If 
$H_N(\mu)=0$, there exists a unique $\Psi\in W^{1,q}(\Lambda^k\Omega)$ $\forall\, q<n/(n-1)$, denoted by $\Psi=-\Delta^{-1}_N(\mu)$, such that 
 $$
 -\Delta\Psi=\mu\, ,\quad \Psi_N=0,\quad (d\Psi)_N=0\, ,
 $$
 so that in particular $H_N(\Psi)=0$. 
 
 \medskip
 \noindent
 If $H_\top(\mu)=0$, then there exists a unique $\Psi\in W^{1,q}(\Lambda^k\Omega)$ $\forall\, q<n/(n-1)$, denoted by $\Psi=-\Delta^{-1}_\top(\mu)$, such that
 $$
 -\Delta\Psi=\mu\, ,\quad \Psi_\top=0,\quad (d^*\Psi)_\top=0\, ,
 $$
 and in particular $H_\top(\Psi)=0$.
 
 \medskip
 \noindent
 In both cases, we have  
 \beq\label{eq:stampac2}
 ||d\Psi ||_q+|| d^*\Psi||_q\le C_q ||\mu||\qquad\forall q<\frac{n}{n-1}\, .
 \eeq 
 \end{proposition}
  \begin{proof} The proof of Proposition \ref{prop:2.10} follows exactly the duality argument {\it \`a la} Stampacchia carried out in \cite{bo}, taking into account the elliptic estimates \eqref{eq:ellipticestimates} for the operators $-\Delta_N$ and $-\Delta_\top$, and observing that they are self-adjoint.
\end{proof}

\begin{corollary}\label{cor:poincare} 
A measure-valued $k$-form $\mu$ is exact if and only if
$d\mu = 0$ and $H_N(\mu)=0$. In addition, if $\mu$ is exact then 
$\mu = d\zeta$, for $\zeta := d^*(-\Delta_N)^{-1}\mu \in \cap_{1\le q <n/n-1} L^q(\Lambda^{k-1}(\Omega))$, and $||\zeta||_{q}\le C_q||\mu||$.

Similarly, a measure-valued $k$ form $\mu$ is co-exact (that is, can be written $\mu = d^*\psi$ for some
measure-valued $k+1$-form $\psi$) if and only if  $d^*\mu=0$ and $H_\top(\mu)=0$, and
if these conditions hold, then $\mu = d^*\zeta$ for $\zeta = d(-\Delta_\top)^{-1}\mu
\in \cap_{1\le q <n/n-1} L^q(\Lambda^{k+1}\Omega)$,
and $||\zeta||_{q}\le C_q||\mu||$.
\end{corollary}
 
\begin{proof}
If $d\mu = 0$ and $H_N(\mu) = 0$ then we appeal to Proposition \ref{prop:2.10} and define $\zeta=d^*(-\Delta_N^{-1}\mu)$, and it follows that $\mu = d\zeta$.
Conversely, $\mu= d\psi$ in $\Omega$ for some measure-valued $k-1$-form $\psi$, then
it is clear that $d\mu=0$ in $\Omega$, and if $\varphi \in \mathcal{H}_N^k$, then for $\chi_\ep$ as in
the proof of \eqref{omegatopclaim}, 
\[
\int \phi\cdot \mu = 
\lim_{\ep\to 0} \int \chi_\ep \varphi \cdot d\psi = 
\lim_{\ep\to 0} \int d^*( \chi_\ep \varphi ) \cdot \psi .
\]
Next, the fact that $\varphi\in \mathcal{H}_N^k$ and properties of $\chi_\ep$ imply that
$|d^*( \chi_\ep \varphi )| = |d\chi_\ep \wedge \star \varphi| \le C$, independent of $\ep$. We then conclude as in the proof of \eqref{omegatopclaim} that $\int \phi\cdot \mu = 0$, and hence that $H_N(\mu) =0$.

The assertions about co-exact forms are proved in exactly the same way.
 \end{proof}
 \begin{remark}\label{rem:poincare}{\rm In case $\Omega=\R^n$, $\mu$ compactly supported, we have in particular
$\zeta=d^*(G*\mu)$ (resp. $\zeta=d(G*\mu)$).
}
\end{remark}

 \begin{remark}\label{rem:weaktrace}
If $\varphi$ is a smooth $k$-form and $\varphi_N = 0$ (resp. $\varphi_\top=0$), then $(d^*\varphi)_N = 0$ (resp. $(d\varphi)_\top=0$).
The form $\zeta$ of Corollary \ref{cor:poincare} is only in $L^q$, and so does not have
a normal (resp. tangential) trace, but can be shown to satisfy $\zeta_N = 0$ (resp. $\zeta_\top=0$) in a sort of distributional sense, as a consequence of the fact that $\zeta = d^*\Psi$ (resp. $\beta=d\Psi$)
for $\Psi=-\Delta^{-1}_N\mu\in W^{1,q}$, with $\Psi_N=0$ (resp. $\Psi=-\Delta^{-1}_\top\mu$, $\Psi_\top=0$).

This distributional trace (of which our definition \eqref{qtop.def1} of $q_\top$ for a closed measure-valued $n-1$-form $q$ is a special case) 
is strong enough to provide uniqueness assertions in the setting of Corollary \ref{cor:poincare}. For example, if $d\mu = 0$, then there is a {\em unique} $\zeta\in  L^q(\Lambda^{k-1}\Omega)$
satisfying $d\zeta = \mu, d^*\zeta = 0$, and  $\zeta_N=0$ in the distributional sense.
\end{remark}

\begin{remark}\label{rem:green} Through the Green operators $-\Delta^{-1}_N$ (resp. $-\Delta^{-1}_\top$), one obtains an integral expression for the linking number of a $k$-cycle and a (relative) $(n-k-1)$-boundary (resp. a relative $k$-cycle with a $(n-k-1)$-boundary) in $\Omega$ (see e.g. \cite{dR}). Let for instance $\Gamma$ be a relative $(n-k-1)$-boundary in $\Omega$, i.e.
$\Gamma=\bd R+\Gamma^\prime$ with $R\subset\Omega$ and $\Gamma^\prime\subset\bd\Omega$.
One immediately verifies that $H_\top(\widehat{\Gamma})=0$, and hence $H_N(\star\widehat{\Gamma})=0$.
Let $\beta=-\Delta^{-1}_N(\star\widehat{\Gamma})$. Hence we have $d^*\beta\in L^p(\Lambda^1\Omega)$ for $p< \frac n{n-1}$ and $\beta$ is smooth outside $\Gamma$. Hence, for a $k$-cycle $\gamma\subset{\Omega\setminus\Gamma}$ we have $0=\widehat{\bd\gamma}=d^*\widehat{\gamma}$, and moreover
\beq\label{eq:linking}
\begin{aligned}
\int_\gamma d^*\beta  &=\langle d^*\Delta_N^{-1}(\star \hat \Gamma)\, ,\widehat{\gamma}\rangle=\langle \widehat{\Gamma}\, , \star d(-\Delta^{-1}_N\widehat{\gamma})\rangle=\langle \widehat{\bd R}\, , \star d(-\Delta^{-1}_N\widehat{\gamma})\rangle\\
&=\langle \widehat{R}\, , \star d^*d(-\Delta^{-1}_N\widehat{\gamma})\rangle=\langle \widehat {R}\, , \star \widehat {\gamma} +\star\Delta^{-1}_N (dd^*\widehat{\gamma})\rangle\\
&=\langle\widehat{R}\, ,\star\widehat{\gamma}\rangle=\langle \widehat{\gamma}\rest R\, , \star 1\rangle=\sum_{a_i\in \gamma\cap R}\star(\tau_\gamma\wedge\star\tau_R(a_i))\in\Z\, .\\
\end{aligned}
\eeq
Observe that in case $\Gamma=\bd R\subset\Omega$ is a $(n-k-1)$-boundary in $\Omega$, we have $H(\widehat\Gamma)=0$, hence we may consider $\beta=-\Delta^{-1}(\star\widehat\Gamma)=G*(\star\widehat\Gamma)$ with $G$ the Poisson kernel in $\R^n$, and deduce for $d^*\beta$  the integral representation
\beq\label{eq:biotsavart1}
 d^*\beta=G*(\star d\widehat{\Gamma})=(\star dG)*\widehat{\Gamma}=\int_{\Gamma}\star dG(x-\cdot), 
\eeq
which in the case $n=3$, $k=1$ reads more familiarly
\beq\label{eq:biotsavart2}
d^*\beta=\sum_{i,j,k =1}^3 {4\pi} dx^i \ \ep_{ijk}\int_{\Gamma_h^\ell}\frac{(x_j-y_j)dy^k}{|x-y|^3}\, .
\eeq
Following \eqref{eq:linking}, we thus deduce the Biot-Savart formula for the linking number link$(\Gamma\, ,\gamma)$ of $\Gamma=\bd R$ with a $k$-cycle $\gamma$ in $\Omega$, namely
\beq\label{eq:biotsavart3}
\int_\gamma d^*\beta=\int_{\gamma_x}\int_{\Gamma_y}\star dG(x-y)=\langle \widehat R\, , \star\widehat\gamma\rangle=\sum_{a_i\in \gamma\cap R}\star(\tau_\gamma\wedge\star\tau_R(a_i))\in\Z\, .
\eeq
\end{remark}
Notice that the integral formula \eqref{eq:biotsavart3} gives link$(\Gamma\, ,\gamma)$ also when $\Gamma$ is just a cycle, i.e. $\bd\Gamma=0$, not necessarily a boundary. In fact, considering $\gamma\times \Gamma\subset \R^n_x\times\R^n_y$, we have $\bd (\gamma\times\Gamma)=0$ in $\R^n\times\R^n$, and $\star dG(x-y)=|S^{n-1}|^{-1}\cdot\psi^*(d\sigma)$, where $\psi:\gamma\times\Gamma\to S^{n-1}\subset\R^n$ is given by $\psi(x,y)=\frac{x-y}{|x-y|}$ and $d\sigma$ is the volume form of $S^{n-1}$. Hence
\beq\label{eq:biotsavart4}
\int_{\gamma_x}\int_{\Gamma_y}\star dG(x-y)=\frac{1}{|S^{n-1}|}\int_{\gamma\times \Gamma} \psi^*(d\sigma)=\text{deg}(\psi)\in\Z\, .
\eeq
%%%%%%%%%%%%%%%%%%%%%%%%%%%%%%%%%%
%The operator $\Delta^{-1}_N$ can be used to compute the linking number of two relative 1-cycles in %$\Omega$
%$\gamma_1,\gamma_2$. In fact, 

%The Green operator $G_N$ (resp. $G_\top$), characterized by the identity $I-H_N=-\Delta_N\cdot G_N$ (resp. $I-H_\top=-\Delta_\top\cdot G_\top$), is in turn represented in coordinates by $G(x,y)=\sum G_I(x,y) dx^I\otimes dy^I\in\Lambda^k(\Omega)\otimes\Lambda^k(\Omega)$, and  verifies $-\Delta_y\cdot G(x,y)=\delta(x-y)\sum dx^I\otimes dy^i$ and $(G(x,y)_N)=(d_yG)_N=0$ (resp. $G(x,y)_\top=(d_y^*G)_\top=0$) for $y\in\bd\Omega$. We have in particular $G(x,y)=G(y,x)$, and  for $H_N(\omega)=0$, we may hence write $-\Delta_N^{-1}\omega(x)=\int_\Omega G(x,y)\wedge\star\omega(y)$.
%\end{remark}

 %\medskip
%%%%%%%%%%%%%%%%%%%%%%%%%%%%%%%%%%%%%%%%%%%%%%%%%%%%%%%%%%%%
\subsection{Representation of harmonic 1-forms.} We describe next the spaces $\mathcal H^1_N$, (resp. $\mathcal H^1_\top$), of harmonic $1$-forms on $\Omega\subset\R^n$ with zero normal (resp. tangential) component on $\bd\Omega$. Since $\mathcal H^{n-1}_N=\star \mathcal H^{1}_\top$ (resp. $\mathcal H^{n-1}_\top=\star \mathcal H^{1}_N$), this yields also a representation for harmonic $(n-1)$-forms.

\begin{lemma} {\rm(Description of $\calH^1_\top$).}
Let $(\partial \Omega)_i$, $i=0,\ldots,b$  denote the connected components of
$\partial \Omega$. Then
$\gamma\in \calH^1_\top$ if and only there exist constants $c_1,\ldots, c_b$ such that
$\gamma = d\phi$, where $\phi$ is the unique harmonic function in $\Omega$ such that
$\phi \equiv c_i$ on $(\partial \Omega)_i$ for $i\ge 1$, and $\phi = 0$
on $(\partial \Omega)_0$.
%Also, $\calH^2_N = \star \calH^1_T$. (notation is inconsistent; elsewhere $\star$ maps forms to %currents.)
\label{L:H2N}\end{lemma}

\begin{proof}
In fact $\calH^1_\top$ is isomorphic to the first relative de Rham cohomology group of $\Omega$,
that is $H^1_{dR}(\Omega;\partial \Omega)$,
(see for example \cite{GMS} vol. 1, Corollary 1, section 5.2.6) and $H^1_{dR}(\Omega,\partial\Omega)\simeq \R^{b}$, as it is shown in Lemma \ref{lem:2completed} below. Finally, the family of 1-forms described in the above statement span a $b$-dimensional subspace of $\mathcal H^1_\top$.
\end{proof}

\begin{lemma}\label{L:H1N}
{\rm (Description of $\calH^1_N$).}
Let $\kappa$ denote the dimension of $\calH^1_N$. Then there exists an
an orthogonal basis
$\{H_j\}_{j=1}^\kappa$  
for
$\calH^1_N$
normalized  so that for each $j$ there exists
a $\R/\Z$-valued function $\phi_j$ such that $H_j = d\phi_j$, so that
$e^{i2\pi \phi_j}$ is well-defined.
\end{lemma}

\begin{proof}
In fact $\calH^1_N$ is isomorphic to the first de Rham cohomology group  $H^1_{dR}(\Omega)$, which  in 
turn is isomorphic to Hom$(H_1(\Omega, \Z), \R)$, and these are all finitely generated.
(See e.g. \cite{GMS} vol.1, Corollary 1 in section 5.2.6 and Theorem 3 in Section 5.3.2). 
%Moreover, de Rham's Theorem states that $H^1_{dR}(\Omega)$ and Hom$(H_1(\Omega, R), R)$ are isomorphic (see e.g. \cite{GMS} vol.1, Corollary 3, section 5.3.2), which implies that $H^1_{dR}(\Omega)$ and $H_1(\Omega, R)$ have the same dimension, say $\kappa$. 
It follows that if   $\{\gamma_i\}_{i=1}^\kappa$ are cycles that
form a basis for $H_1(\Omega; \Z)$, then 
there exists a (unique) basis $\{H_i\}_{i=1}^\kappa$ for $\calH^1_N$
such that $\int_{\gamma_j}H_j  = \delta_{ij}$ for $i,j=1,\ldots, \kappa$.
We now fix $x_0\in \Omega$ and define $\phi_j(x) := \int_{\gamma(x_0,x)}H_j$,
$j=1\ldots,\kappa$, where $\gamma(x_0,x)$ is any path in $\Omega$ that starts at $x_0$ and ends at $x$.
If $\gamma'(x_0,x)$ is another such path, then $\gamma(x_0,x) - \gamma'(x_0,x)$ is homologous to an integer linear combination of the $\gamma_i$'s, so that $ \int_{\gamma(x_0,x)}H_j- \int_{\gamma'(x_0,x)}H_j 
\in \Z$.
%Then then defining property of $H_j$ implies that $\phi_j(x)$ is well-defined (that is, independent of the choice of path $\gamma(x_0,x)$) modulo $\Z$. 
Thus $\phi_j$ is well-defined as a function $\Omega\to \R/ \Z$. It is immediate that
$H_j = d\phi_j$.  \end{proof}

\begin{remark}\label{rem:restrictedharmonic} Although this fact is not needed in this paper,
we remark that if $H\in\mathcal H^{k}_N$, and $K=H_\top=H\rest\bd\Omega$ is its tangential component on $\bd\Omega$, then $K$ is a harmonic $k$-form in $\bd\Omega$. (A special case of this fact is used in the proof of Lemma \ref{L:H1N} above.) Indeed, since $dH=0$ and $(dH)_N=0$ we have $dK=(dH)_\top=dH-(dH)_N=0$. Moreover, one can check that $d\star_\top K=(d\star H)_\top$ since $H_N=0$, where $\star_\top$ denotes the star operator on the tangent space of $\bd\Omega$. Hence $d\star_\top K=0$ and the conclusion follows.
\end{remark}

We describe next an exactness criterion for closed $(n-1)$-forms in $\Omega\subset\R^n$. 

\begin{lemma}\label{lem:dp} A measure-valued $(n-1)$ form $q$ on a smooth bounded open set $\Omega\subset \R^n$
is exact if and only if 
$dq=0$ and  $\int_{(\bd\Omega)_i}q_\top=0$ for every connected component $(\bd\Omega)_i$ of $\bd\Omega$. \end{lemma}

\begin{proof} Let $\gamma\in \mathcal H^{n-1}_N$, so that $\star\gamma\in \mathcal H^1_\top$ and hence, by Lemma \ref{L:H2N}, $\star\gamma=d\varphi$, where $\Delta\varphi=0$ in $\Omega$ and $\varphi\equiv c_i$ on the $i$-th connected component $(\bd\Omega)_i$. Then
\beq
\begin{aligned}
\langle q\, ,\gamma\rangle &=\int_\Omega q\wedge\star\gamma=\int_\Omega q\wedge d\varphi 
%\\ &
\ \overset{\eqref{qtop.def1}, \eqref{omegatopclaim}}= \ \sum_{i=1}^b c_i\int_{(\bd\Omega)_i} q_\top
%\\ &=0\, .
\end{aligned}
\eeq
We deduce that $H_N(q)=0$ if and only if $\int_{(\bd\Omega)_i} q_\top=0$ for every $i$. The conclusion now follows from Corollary \ref{cor:poincare}.
\end{proof}

%%%%%%%%%%%%%%%%%%%%%%%%%%%%%%%%%

\subsection{Proof of Lemma \ref{L:H2N} completed}  We  need the following easy result, whose proof uses the language of algebraic topology (see e.g. \cite{spa}).

\begin{lemma}\label{lem:2completed} Let $U$ be a connected Lipschitz domain in $\R^n$, such that $\partial U$ has $b+1$ connected components. Then $H^1_{dR}(U,\partial U)\simeq\R^{b}$.
\end{lemma}

\begin{proof} From the exact sequence in singular homology for the pair $(\bar U\, ,\partial U)$ we have
\beq\label{eq:exact1}
H_1(\partial U)\xrightarrow{i_*} H_1(\bar U)\xrightarrow{\Phi_*} H_1(\bar U ,\partial U)\xrightarrow{\partial_*} H_0(\partial U)\xrightarrow{i_*^0} H_0(\bar U)\rightarrow 0
\eeq
which gives rise to the short exact sequence
\beq\label{eq:exact1'}
0\rightarrow \text{Im}\, \Phi_*\rightarrow H_1(\bar U,\partial U)\rightarrow \text{Ker}\, i_*^0\rightarrow 0\, .
\eeq
By hypothesis we have $H_0(U)=\Z$, $H_0(\partial U)=\Z^{b+1}$, and \eqref{eq:exact1} implies Ker$\, i^0_*=\Z^{b}$.
By the Mayer-Vietoris exact sequence for $V=\bar U$, $W=\R^n\setminus U$ we have
\beq\label{eq:exact2}
H_2(V\cup W)\rightarrow H_1(V\cap W)\xrightarrow{(i_*,i_*)} H_1(V)\oplus H_1(W) \rightarrow H_1(V\cup W)
\eeq
which yields, since $V\cup W=\R^n$ is contractible,
\beq\label{eq:exact3}
0\rightarrow H_1(\partial U)\xrightarrow{(i_*,i_*)} H_1(\bar U)\oplus H_1(\R^n\setminus\bar U)\rightarrow 0\, ,
\eeq
so that $(i_*,i_*)$ is an isomprphism. In particular $i_*=\pi_{1}\circ (i_*,i_*)$ is onto, hence $H_1(\bar U)=$Im$\, i_*$=Ker$\, \Phi_*$, which yields Im$\, \Phi_*=0$, so that \eqref{eq:exact1'} implies that $H_1(\bar U\, , \partial U)$ is isomorphic to Ker$\, i_*^0=\Z^{b}$. From the regularity assumption\footnote{actually it sufficient for $U$ to be a Lipschitz neighborhood retract in $\R^n$} on $U$ we have in particular $H_1(\bar U,\partial U)\simeq H_1(U,\partial U)$. Finally, from the relation
\beq\label{eq:cohomology}
H^1(U,\partial U;\R)=\text{Hom}\, ( H_1(U,\partial U);\R)=\text{Hom}\, (\Z^{b};\R)\simeq \R^{b}
\eeq
the conclusion follows, since the first singular relative cohomology group with real coefficients $H^1(U,\partial U;\R)$ is isomorphic to the  first de Rham relative cohomology group $H^1_{dR}(U,\partial U)$.

\end{proof}

%%%%%%%%%%%%%%%%%%%%%%%%%%%%%%%%%%%%%%%%%
%%%%%%%%%%%%%%%%%%%%

%\subsection{A second order elliptic problem.} Let $U\subset\R^3$ be a smooth domain such that $\bd U$ is connected. Then by Lemma \ref{L:H2N} we have $\mathcal H^1_\top=\{0\}$, and the following elliptic problem admits a unique solution:

%\beq\label{eq:2ndorder}
%\left\{
%\begin{aligned}
%-\Delta\Psi &=\omega &\text{in }  U\\
%\Psi_\top &=\xi  &\text{ on }  \bd U\\
%(d^*\Psi)_\top &=A  &\text{ \ \ on } \bd U\, .\\
%\end{aligned}
%\right.
%\eeq

%%%%%%%%%%%%%%%%%%%%%%%%%%%%%%%%%%%%%%%%%%%%%%%%%%%%%%%%%%%%

\subsection{Proof of Lemma \ref{lem:mincon}.}\label{sect:mincon} 
\mbox{ }
%We consider the case $\zeta$ is smooth function supported in $\bd K$. The conclusion will then follow by a density argument. 

\smallskip
%\noindent{\bf Step 1.} We first claim that $\{ \alpha: d\alpha = 0 \mbox{ in }K, \alpha_\top = \zeta\mbox{ on }\bd K\}$ is nonempty.

%\noindent
% In fact, since $\int_{\bd K}\zeta=0$, consider a solution $\varphi\in C^\infty(\Lambda^3K)$ of the Neumann problem $-\Delta\varphi=0$ in $S$, $(d^*\varphi)_\top=\zeta$ on $\bd K$. Such a solution clearly exists if $\zeta$ is a smooth function of $\partial K$, and the general case of a measure supported on $\partial K$ follows by a density argument. Then $\xi=d^*\varphi\in C^\infty(\Lambda^2K)$ is the unique solution of 
%$$
%\begin{cases}
%d\xi = 0 & \text{in } K\\ 
%d^*\xi=0 & \text{in }  K \\
% \xi_\top=\zeta & \text{on }\bd K\\
%\end{cases}
%$$
%which in particular proves the claim.

\smallskip
\noindent
{\bf Step 1.}  We have: $
\inf \{ || \alpha ||_{L^1(\Lambda^2K)}\, , \ d\alpha=0\text{ in } K\, ,\ \alpha_\top=\zeta\ \text{on } \bd K\, \}=|| \zeta ||_{\dot W^{-1,1}(K)}
$,
where
\[
||\zeta ||_{\dot W^{-1,1}(K)}=\sup \left\{ \int \varphi\, \zeta  \ : \  \varphi\in W^{1,\infty}_c(\R^3)\, ,\ ||d\varphi||_{L^\infty(K)}\le 1\right\}.
\]
\noindent
This follows by a straightforward modification of an argument in  Federer \cite{F}.
We provide a sketch:
define  a linear functional acting on $C^\infty_c(\R^3)$ by
\[
A(\varphi):=\int_{\bd K}\varphi \,\zeta,\quad\quad \varphi\in C^\infty_c(\R^3).
\]
Given any measure-valued $2$-form $\alpha$, we similarly define a linear functional $B_\alpha$
acting on $C^\infty_c(\Lambda^1\R^3)$ by
\[
B_\alpha(\psi)=\int_{K}\psi\wedge\alpha\, , \quad\quad\psi\in C^\infty_c(\Lambda^1\R^3).
\]
And generally, for a linear functional $C$ on  $C^\infty_c(\Lambda^1\R^3)$,
we define $\bd C(\varphi) := C(d\varphi)$ for $\varphi\in C^\infty_c(\R^3)$.
Then the definitions (see \eqref{qtop.def1} in particular) imply that 
$A=\bd C$ and $\|C\|<\infty$ if and only if  $C = B_\alpha$ for some measure-valued $2$-form $\alpha$ such that  $d\alpha=0$ in $K$ and $\alpha_\top=\zeta$ on $\bd K$. Next, we note that 
$||\zeta ||_{\dot W^{-1,1}(K)}={\bf F}_{hom,K}(A)\, ,$
where ${\bf F}_{hom,S}(A)$ denotes the {\it homogeneous} flat norm of  $A$ in $K$, see \cite{F}. 
Then as observed in section 4.1.12 of \cite{F} in a slightly different setting, the Hahn-Banach Theorem implies that
 $$
%\begin{aligned}
{\bf F}_{hom,K}(A) =\min \{\, || C ||\, ,\ \text{spt}\, C\subset K\, ,\ \bd C=A\, \}
$$
and this translates to our claim, in view of our earlier remarks.

\smallskip
\noindent
{\bf Step 2.} We claim that $|| \zeta ||_{\dot W^{-1,1}(K)}\le C || \zeta ||_{W^{-1,1}(\R^3)}$,%\le C|| \zeta ||_{W^{-1,1}(\Lambda^2\bd K)}$.
where
\[
||\zeta ||_{W^{-1,1}(\R^3)}=\sup \{ \int_{ \R^3}\varphi \zeta\, ,\ \varphi\in W^{1,\infty}_c(\R^3)\, ,\ ||\varphi||_{W^{1,\infty}(\R^3)}\le 1\}.
\]
%Toward this end, first note that because $K$ is convex, if $\| d\varphi\|_{L^\infty(K)}\le 1$, then $\varphi$ is $1$-Lipschitz on $K$, and hence there exists $\tilde \varphi: \R^3\to \R$ with compact support, such  that $\tilde \varphi = \varphi$ in $K$, and $\|\tilde  \varphi \|_{W^{1,\infty}(\R^3)} =\| \varphi \|_{W^{1,\infty}(K)}$.
It suffices to show that there exists $C>0$ such that, for any $\varphi\in W^{1,\infty}_c(\R^3)$ with $\| d\varphi\|_{L^\infty(K)}\le 1$, there exists $\psi\in W^{1,\infty}_c(\R^3)$ such that
\beq\label{eq:comparenorms}
\int \varphi \zeta = \int \psi\zeta
\quad\quad
\mbox{ and }\| \psi\|_{W^{1,\infty}(\R^3)} \le C.% \|d\varphi \|_{L^\infty(K)\}
\eeq
Indeed, given $\varphi$ such that $\|d\varphi\|_{L^\infty(K)}<\infty$, we fix $x_0\in K$ and
we define $\psi(x) = \varphi(x) - \varphi(x_0)$ for $x\in K$.
% (The definition of $\psi$ on $\R^3\setminus K$ will  be discussed in a moment.) 
Since $K$ is convex, $\varphi$ and hence $\psi$ are $1$-Lipschitz on $K$, 
so that $|\psi(x)| \le |x-x_0| \le \mbox{diam}(K)$ in $K$.
Next, we extend $\psi$ to  $\R^3\setminus K$, such that the extended function is still $1$-Lipschitz and
moreover satisfies $\|\psi\|_{L^\infty(\R^3)} \le \mbox{diam(K)}$, and has compact support.

Since $\zeta$ is a measure supported on $\bd K$, clearly $\int \psi\zeta$ depends only on the behavior of $\psi$ in $\bd K$,
and hence $\int \psi \,\zeta=\int (\varphi - \varphi(x_0))\,\zeta=\int \varphi \,\zeta$, since $\int_{\bd K}\zeta=0$, proving \eqref{eq:comparenorms}
\qed

\subsection{Proof of Lemma \ref{lem:nurest}.}\label{sect:nurest} 
{\bf Step 1}. We will show below that
there exists a piecewise smooth oriented 2-manifold with boundary $S=S_\ep$ such that 
\beq\label{eq:ScapU}
\bd S=M_\ep-M'_\ep\qquad\text{in }U\quad\text{and } \calH^2(S\cap U)\le C \ell\cdot E_\ep(u_\ep;\Omega)\le C\ell g_\ep,
\eeq
%Moreover, there exists a polyhedral surface $S_\ep=\sum_k \alpha_k P_k$ , for suitable 2-simplices $P_k\equiv P_{k,\ep}$ and coefficients $\alpha_k\equiv\alpha_{k,\ep}\in\R$, such that $\partial S_\ep=M_\ep-M'_\ep$ in $U$
 %and 
 %$$
 %||S_\ep||(U)=\sum_k |\alpha_k|\cdot \text{vol}(P_k\cap U)\le \gamma ||\nu_\ep-\nu'_\ep||_{W^{-1,1}(U)},
 %$$
with $C>0$ independent of $\ep$ and $U$. (See the proof of Proposition \ref{prop:3.1} for notation used here and below.) We first complete the proof of the lemma, assuming \eqref{eq:ScapU}.

We may assume that $S$ intersects transversally the level set $f^{-1}(t)$ for a.e. $t$, since if not, we can arrange that this condition is satisfied after an arbitrarily small perturbation of $S$ that leaves $\partial S$ fixed.
Noting that $f^{-1}(t)$ coincides with $\bd C^t$ for a.e. $t$, we deduce that $S\cap \bd C^t$ is piecewise smooth for a.e. $t>0$. 
 
Since $f$ is $1$-Lipschitz, the same is true for $f\rest S$, so that  $|\nabla (f\rest S)|\le 1$ a.e., and
$$
\calH^2((S\cap C^{N\ell})\cap U)\ge\int_{(S\cap C^{N\ell})\cap U}|\nabla (f\rest S)|d\mathcal{H}^2=\int_0^{N\ell} \mathcal{H}^1((S\cap\bd C^t)\cap U)dt\, ,
$$
by the coarea formula.
%$$
%\begin{aligned}
%\int_0^{N\ell}||S_\ep\cap \bd C^t||(U)\, dt &=\sum_k |\alpha_k|\int_0^{N\ell}\mathcal{H}^1((P_k\cap U)\cap\bd C^t)dt\\
%&\le \sum_k|\alpha_k|\text{vol}((P_k\cap U)\cap C^{N\ell})=||S_\ep\cap C^{N\ell}||(U)\\
%&\le ||S_\ep||(U)\, .
%\end{aligned}
%$$
We deduce that there exists $t_\ep$ s.t. 
\beq\label{eq:ScapCt}
\calH^1((S\cap \bd C^{t_\ep})\cap U)\le (N\ell)^{-1}\calH^2(S\cap U)\le C N^{-1}g_\ep\, .
\eeq
%and hence, by our choice of $N$ in \eqref{eq:3.21}, we have
%\beq\label{eq: Sto0}
%\calH^1((S\cap\bd C^{t_\ep})\cap U)\le \logeps^{-1}\to 0 \qquad\text{as }\ep\to 0.
%\eeq
In $U$ it holds
\beq\label{eq:cutS}
\begin{aligned}
\bd(S\cap C^{t_\ep}) &=(\bd S)\cap C^{t_\ep}+S\cap(\bd C^{t_\ep})\\
&= (M_\ep-M'_\ep)\cap C^{t_\ep}+ S\cap(\bd C^{t_\ep}) \\
&=M_\ep-M'_\ep\cap C^{t_\ep}+S\cap (\bd C^{t_\ep})\, . \\
\end{aligned}
\eeq
In particular, for $\phi\in C^\infty_c(\Lambda^1 U)$, we have $$
\langle \nu_\ep-\nu'_\ep\rest C^{t_\ep}\, ,\phi\rangle=\int_{S\cap C^{t_\ep}}d\star\phi-\int_{S\cap \bd C^{t_\ep}}\star\phi\, ,
$$
(using the definitions \eqref{eq:nu} and
\eqref{eq:nu'}),
whence
\beq\label{eq:nu'restclose}
\begin{aligned}
||\nu_\ep-\nu'_\ep\rest C^{t_\ep}||_{W^{-1,1}(U)} &\le \calH^2(S\cap C^{t_\ep}\cap U)+\calH^1(S\cap\bd C^{t_\ep}\cap U)\\
&\le (1+(N\ell)^{-1})\calH^2(S\cap U)\le C(\ell+N^{-1})g_\ep
\end{aligned}
\eeq
by \eqref{eq:ScapCt} and \eqref{eq:ScapU}. This gives precisely \eqref{eq:nurest}.

{\bf Step 2}. To conclude, we supply the proof of our earlier claim \eqref{eq:ScapU}.

Let $g(x)=|\text{dist}(x,R_1)|^{-1}+|\text{dist}(x,R_1^*)|^{-1}$. By the coarea formula, we have
\beq\label{eq:coareau}
\begin{aligned}
\int_{B_1}ds\int_{u_\ep^{-1}(s)}g(x)d\calH^1(x) &=\int_\Omega g(x)|Ju_\ep|dx
\le \int_\Omega g(x)e_\ep(u_\ep)dx\, ,
\end{aligned}
\eeq
 so that by a mean-value argument, \eqref{eq:meanvalue1}, and \eqref{eq:meanvalueE}, we deduce from \eqref{eq:coareau} that there exists a regular value $s$ of $u_\ep$ such that
 $|s|< 1/2$ and, denoting $M_s:=u_\ep^{-1}(s)$, we have
 \beq\label{eq:meanvalueu}
 \int_{M_s}g(x)d\calH^1(x)=\int_{M_s}\frac{d\calH^1(x)}{|\text{dist}(x,R_1)|}+\int_{M_s}\frac{d\calH^1(x)}{|\text{dist}(x,R_1^*)|}\le \frac{KE_\ep(u_\ep;\Omega)}{\pi\delta\ell}.
 \eeq
Define as in \cite{abo}, Lemma 3.8 (i), the map $\Phi:\R^3\setminus R_1\to R'_1$ and, accordingly, the map $\Phi^*:\R^3\setminus R^*_1\to {R^*_1}'$. Set $\Psi(t,x)=(1-t)x+t\Phi(x)$, $\Psi^*(t,x)=(1-t)x+t\Phi^*(x)$,  and define
$S_1=\Psi([0,1]\times M_s)$ and $S_2=\Psi^*([0,1]\times M_s)$.
Note, following  \cite{abo}, Lemma 3.8 (ii), that 
since $M_s$ has no boundary in $U$, we have $\bd S_1=\Phi_\#M_s-M_s$ and $\bd S_2=\Phi_\#^*M_s-M_s$ in $U$.
However, from  \cite{abo}, Lemma 3.8 (i), we know that $\Phi_\#M_s = M_\ep$ the point being that
the intersection number of $M_s$ with any $2$-face $Q_i$ agrees with $(-1)^{\sigma_i} d_{Q_i}$,
due to orientation conventions and elementary properties of topological degree.
Similarly $\Phi^*_\#M_s = M_\ep'$, so if we define $S := S_1-S_2$, then $\bd S =  M_\ep- M_\ep'$ in $U$, which
is the first part of \eqref{eq:ScapU}. Following the proof of \cite{abo}, Lemma 3.8 (ii), we readily deduce that
\beq\label{eq:areaS}
\calH^2(S\cap U)=\calH^2(S_1\cap U)+\calH^1(S_2\cap U)\le C\ell^2\int_{M_s}g(x)d\calH^1(x)\, .
\eeq
Combining \eqref{eq:areaS} and \eqref{eq:meanvalueu}, claim \eqref{eq:ScapU} follows.

\qed
%%%%%%%%%%%%%%%%%%%%%%%%%%%%%%%%%%%%%
%%%%%%%%%%

%%%%%%%%%%%%%%%%%%%%%%%%%%%%%%%%%%%%%%%%%%%%%%%%%%%%%%%%%%%

\end{document}